\def\enableLlncsCode{}

\def\enableLlncsCodeSlightlyImproved{}

\def\nicerThanLlncs{}

\def\displayArxivTexts{}

\documentclass[\ifdefined\enableLlncsCodeSlightlyImproved envcountsect, envcountsame\fi \ifdefined\nicerThanLlncs ,11pt\fi]{llncs}

\IfFileExists{personal_changes.tex}{\input{personal_changes}}{}
\input{include.tex}
\title{Non-Interactive and Non-Destructive Zero-Knowledge Proofs on Quantum States and \publishedOnly{\\}Multi-Party Generation of \publishedOnly{\\}Authorized Hidden \GHZ{} States}
\date{}

\usepackage{orcidlink}
\author{
  \arxivOnly{%
    Léo Colisson\,\orcidlink{0000-0001-8963-4656}\hspace*{0.05em}\inst{1}, Frédéric Grosshans\,\orcidlink{0000-0001-8170-9668}\hspace*{0.05em}\inst{1}, Elham Kashefi \inst{1,2}%
  }
}

\usepackage{orcidlink}
\institute{
  \arxivOnly{%
    Laboratoire d'Informatique de Paris 6 (LIP6), Sorbonne Université, \\
    4 Place Jussieu, 75252 Paris CEDEX 05, France \\ \{\href{mailto:leo.colisson@lip6.fr}{leo.colisson}, \href{mailto:frederic.grosshans@lip6.fr}{frederic.grosshans}\}@lip6.fr
    \and
    School of Informatics, University of Edinburgh, \\
    10 Crichton Street, Edinburgh EH8 9AB, UK
  }%
}

\begin{document}
{\def\addcontentsline#1#2#3{}\maketitle} 
\begin{abstract}
  We propose the first generalization of the famous Non-Interactive Zero-Knowledge (NIZK) proofs to \emph{quantum} languages (NIZKoQS) and we provide a protocol to prove advanced properties on a received quantum state non-destructively and non-interactively (a single message being sent from the prover to the verifier).

  In our second orthogonal contribution, we improve the costly Remote State Preparation protocols~\cite{CCKW_2019_qfactory,GV19_ComputationallySecureComposableRemote} that can classically fake a quantum channel (this is at the heart of our NIZKoQS protocol) by showing how to create a multi-qubits state from a single superposition.

  Finally, we generalize these results to a multi-party setting and prove that multiple parties can anonymously distribute a \GHZ{} state in such a way that only participants knowing a secret credential can share this state, which could have applications to quantum anonymous transmission, quantum secret sharing, quantum onion routing and more.

\end{abstract}
\keywords{Quantum Cryptography, Remote State Preparation, Zero-Knowledge, Learning With Errors}
\arxivOnly{\newpage \tableofcontents}

\newpage

\section{Introduction}

Due to the special no-cloning principle, quantum states appear to be very useful in cryptography. But this very same property also has drawbacks: when receiving a quantum state, it is nearly impossible for the receiver to efficiently check non-trivial properties on that state without destroying it. This allows a sender to send maliciously crafted states (potentially entangled with a larger system) without being detected.

To illustrate this, let us imagine the following simple goal. A receiver would like to obtain a quantum state $\ket{\psi}$ and verify, without destroying that state, that this state belongs to some ``\emph{quantum language}'', say the language composed of \BBHeightyFor{} states (so $\ket{\psi} \in \{\ket{0},\ket{1},\ket{+},\ket{-}\}$). Since any direct measurement would destroy that state, a first solution could be to use a generic quantum secure multiparty computing protocol  (QSMPC)~\cite{DNS12_ActivelySecureTwoParty,DGJ+20_SecureMultipartyQuantum,KKMO21_DelegatingMultiPartyQuantum} between the sender and the receiver in order to generate that state. However, these protocols are interactive and require at least $2$ messages (depending on the number of users and on the complexity of the prepared state, the number of rounds can increase significantly). Therefore, the following question was left open:
\begin{center}
  \emph{Is it possible to receive via a single message a quantum state and test non-trivial properties on it non-destructively?}
\end{center}

At the heart of our method are \emph{classical-client Remote State Preparation} (RSP) protocols~\cite{CCKW18,CCKW_2019_qfactory,GV19_ComputationallySecureComposableRemote}, \cite{GV19_ComputationallySecureComposableRemote} being itself based on the groundbreaking work of~\cite{mahadev2017classical}. These protocols can be used to fake a quantum channel using a purely classical channel. However, they are particularly heavy to run: a single qubit can be sent at a time at the expense of doing a very large and expensive quantum superposition. To prepare more complex states, the solution is typically to repeat the above protocol $n$ times before combining appropriately the different runs. Unfortunately, this method is very costly since the superposition needs to be recreated $n$ times. Therefore, a second question naturally arises, orthogonal to the first one:
\begin{center}
  \emph{Is it possible to classically send multi-qubits states while paying the cost of a single superposition?}
\end{center}


\subsection{Our Results}

In this work, we answer the two above interrogations in the affirmative, and provide various use cases together with a generalization of \RSP{} to multiparty :
\begin{enumerate}[listparindent=\parindent]
  \item We first provide a method to classically prepare large states on $n$ qubits (where a \GHZ{} state is hidden in between $\ket{0}$ and $\ket{1}$ qubits) using a single superposition. It allows us to improve the existing \RSP{} protocols, spreading the cost to create one state among the multiple qubits. We are therefore able to asymptotically obtain a quadratic improvement: the number of operations required to prepare this $n$ qubits state drops from $O(nMN)$ to $O((M+n)N) = O(nN)$, $M > N$ being very large constants.
  \item We lift the above construction to a multiparty setting in which each sender can control one part of the hidden \GHZ{} state. This construction is still non-interactive in the sense that a single message is sent from each sender to the receiver.
  \item We initiate the study of \emph{Non-Interactive and Non-Destructive Zero-Knowledge Proofs on Quantum States} (NIZKoQS). We obtain these guarantees by sending a classical NIZK proof on \emph{classical} instructions used to produce a quantum state. This approach has numerous advantages: the verification is \emph{non-destructive}, \emph{non-interactive}, can be done over a purely \emph{classical channel} (the sender does not need any quantum capabilities), and it is now \emph{possible to check much more involved properties} on the quantum state, potentially linked with secret classical information or involving \emph{multiple parties}. For instance, we can force the sender to send different categories of states depending on whether or not they know a classical password or signature, without revealing to the receiver to which category the state belongs to.



        We emphasize that this method is very generic and can be applied to any (non-interactive) \RSP{} protocol, but also to a non-interactive computation performed using the protocol from \cite{mahadev2017classical}\footnote{In that later case, we can prove any property on the output state before the application of the Pauli keys.  Note also that when dropping the non-interactive requirements, any property on the final state should be verifiable (and it is also possible to use more involved classical protocols like classical Multi-Party computation\arxivOnly{~\cite{GMW87_HowPlayANY}} to classically manipulate quantum states), but this is out of the scope of this paper.}. However, we focus in this paper on a generalization of~\cite{CCKW_2019_qfactory} which allows more efficient \GHZ{}-like state generation and can be extended to a multiparty protocol.

        While it is not the first time that \LWE{} based methods are used in quantum cryptography to protect the sender, to the best of our knowledge, it is the first time that this kind of cryptography is used to also protect the receiver. This allows us to \emph{bypass fundamental quantum limitations}: we expect this result to have further independent interests.
  \item We present a concrete usecase in the context of the \emph{generation of multi-party certified hidden \GHZ{} states}. We derive a protocol between a server and $n$ ``applicants''. At the end of the protocol, the applicants end-up sharing a \GHZ{} state in such a way that, if the server is honest, only the applicants knowing a classical secret (password/secret key/signature/\dots{}) can be part of the \GHZ{} (they are then said to be \emph{supported}). Others applicant will not be entangled in any way with the supported applicants, and therefore cannot disturb a potential future protocol. Our protocol has the following properties:
        \begin{itemize}
          \item The protocol is \emph{blind}, in the sense that no adversary (corrupting potentially the server and some applicants) can learn whether the honest applicants are supported or not.
          \item As soon as the server is honest, it is also impossible for any malicious or noisy applicant to significantly change the \GHZ{} state obtained by supported applicants. This is particularly interesting to avoid ``Denial Of Service'' kind of attacks where some noisy applicants would always force the protocol to abort.
        \end{itemize}

        One direct application is the following: a server can do a quantum secret sharing between parties that have been authorized, say, by some Certification Authority. The identity of these authorized parties stay hidden, even against a malicious coalition involving the server and applicants.
  \item We provide a series of simpler protocols that only provide blindness. In these protocol a classical trusted third party is in charge of choosing who is supported and who is not. These protocols could prove useful in the context of Anonymous Transmission, \emph{Quantum Onion Routing}\dots{}
  \item We give an \emph{instantiation of a cryptographic family} required in our protocol whose security relies on the hardness of \LWE{} (with superpolynomial modulus-to-noise ratio).
\end{enumerate}

\subsection{Related Works and Further Applications}\label{sec:relatedWorks}

Classical Zero-Knowledge proofs and Interactive Proofs systems have been introduced thirty years ago~\cite{GMR_1989_Knowledge_interactive_proofs, BM_1988_Arthur_Merlin}, and allow a prover to prove a statement to a verifier without revealing anything beyond the fact that this statement is true. Zero-Knowledge proofs have been proposed for any language in \npol~\cite{GMW_1991_Proof_yield_nothing_but_validity}. While Non-Interactive Zero-Knowledge (NIZK) proofs are known to be impossible in the plain model~\cite{FS87_HowProveYourself}, NIZK can be obtained in the Common Reference String model~\cite{BFM88_NoninteractiveZeroknowledgeIts} or in the Random Oracle model by using the famous Fiat-Shamir transformation~\cite{GO94_DefinitionsPropertiesZeroknowledge}. The security of classical Zero-Knowledge proofs were also extended to protect against malicious quantum provers~\cite{Watrous_2005_ZK_against_q_attacks,Unruh_2012_Proof_of_Knowledge,BS_2019_ZK_Ct_rounds}, and much work has been done to extend these protocols to remove interactivity~\cite{DFMS19_SecurityFiatShamirTransformation,LZ19_RevisitingPostquantumFiatShamir} or deal with multiple (potentially quantum) provers, targeting larger classes like $\mathsf{IP},\mathsf{QMA}$, $\mathsf{MIP}$ or $\mathsf{MIP^*}$, sometimes considering ``dequantized'' verifiers (\cite{\arxivOnly{IY_1987_Direct_minimum_knowledge_computations},BW_2016_zk_qma,GSY_2019_perfect_zk_multiprovers,CVZ_2020_non_interactive_zk_qma,ACGH20_NoninteractiveClassicalVerification,BG20_QMAhardnessConsistencyLocal,BCKM20_ComplexityTwoPartyQuantum,VZ_2020_classical_ZK_for_Quantum_comp,Shm20_MultitheoremMaliciousDesignatedVerifier,MY21_ClassicallyVerifiableDualMode}, see also the review~\cite{VW_2016_quantum_proofs}).

However, these works focus on \emph{classical} languages (even in \QMA{}, the language is still classical, even if the witness can be quantum). To prove similar properties on \emph{quantum} languages, a first solution is to use a generic quantum secure multiparty computing protocol  (QSMPC)~\cite{DNS12_ActivelySecureTwoParty,DGJ+20_SecureMultipartyQuantum,KKMO21_DelegatingMultiPartyQuantum}. However, these protocols are interactive, and to the best of our knowledge no work has been done to provide one-shot Zero-Knowledge proofs on quantum states.

In this paper, we use internally classical cryptography to leverage quantum cryptography. \cite{mahadev2017classical} introduced for the first time this notion, and this idea lead to the development of many quantum protocols based on classical cryptography~\cite{mahadev2018classical,BCMVV_2018_crypto_test_quantumness,brakerski_2018_quantum_FHE,MDCA_2020_device_independent_QKD_comp_assump}.

The first classical-client Remote State Preparation (RSP) protocol was proposed in \cite{CCKW18}. Informally, in this protocol, a purely classical client can create a state on a quantum server, in such a way that only the client knows the description of that qubit. This protocol named QFactory was later improved in \cite{CCKW_2019_qfactory}, and proven secure against an arbitrary malicious adversary. Another protocol~\cite{GV19_ComputationallySecureComposableRemote} can also be used to obtain a verifiable classical-client RSP. While this later paper provides stronger guarantees in term of verification, it cannot be directly extended to produce multi-party \GHZ{} states, requires interaction and offers a security scaling polynomially with the security parameter for both verifiability (which is certainly hard to avoid) and blindness. For these reasons, we will build our approach upon the QFactory protocol. Some general impossibility results on classical RSP protocols are also given in \cite{BCCKLMW_2020_security_limitations,\arxivOnlyNoDiff{MT_2020_trusted_center_rsp}}: therefore we will not consider the Constructive Cryptography framework in this paper.

In this work, we focus on the distributed generation of Greenberger--Horne--Zeilinger (\GHZ{}) state~\cite{GHZ_1989_going}. This state (and its special case, the Bell state) is very popular, and appears to be useful in many protocols, such as in Quantum Secret Sharing~\cite{hillery1999quantum}, Quantum Teleportation~\cite{BBGCJPW_1993_teleporting}, Entanglement Distillation~\cite{\arxivOnly{BBPS_1996_concentrating_partial_entanglement,}BBPSSW_1996_purification\arxivOnlyNoDiff{,BDSW_1996_mixed_state_entanglement_error_code}}, Device-Independent Quantum-Key-Distribution~\cite{MY98_QuantumCryptographyImperfect}, Anonymous Transmission~\cite{CW_2005_quantum_anonymous_transmissions}, Quantum Routing~\cite{\arxivOnlyNoDiff{PWD_2018_Modular,}MMG_2019_Distributing}. Our approach could be used as an initial step in order to run this kind of protocols between an unknown subset of authorized parties. We note that even if our protocols are computationally secure, one may still be able to obtain unconditional guarantees, for example when the source is not corrupted.

Our protocol could also be used to achieve new functionalities, such as a Quantum Onion-like Routing protocol in order to route a quantum message through an untrusted quantum network, hiding the exact taken route. The idea would be to ask to each intermediate server to blindly generate a large state in which a Bell pair is hidden: this Bell pair could then be used to teleport the qubit to the next server, without revealing its identity to the previous one.

\subsection{Paper Organization}

In the \cref{sec:overview} we give a generic overview of our result: \cref{subsec:initialSetup} presents the role of the different parties involved in the target protocols that will benefit from NIZQoQS and multiparty large state preparation, \cref{subsec:listProtocols} gives the list of the different protocols we realize, and \cref{subsec:quick_overview} gives an insight on the technical methods used in this paper. The preliminaries can be found in \cref{sec:preliminaries}.
Then, we give in \cref{sec:crypto_requirements} the cryptographic requirements of our protocols. In \cref{sec:protocols}, the reader will find the definition of our different protocols\publishedVsArxiv{ }{---including one desirable but impossible protocol---}with our main protocol \authBlindCanDist{} in \cref{subsec:authBlindCanDist}. \cref{sec:NIZKoQS} defines NIZKoQS and proves that our protocols exhibits NIZKoQS. Finally, the \cref{sec:functionConstruction} explains how to construct the required cryptographic primitives\publishedVsArxiv{}{: \cref{sec:construction_f_superpoly} shows how to obtain an \AssumpFct{} family, and \cref{sec:compilerAssumpFctCan} describes a generic way to construct a distributable \AssumpFctCanPrime{} family}.

\section{Technical overview and presentation of the setup}\label{sec:overview}

\subsection{Initial setup}\label{subsec:initialSetup}
Our two first orthogonal results (efficient large state preparation and NIZKoQS) will ultimately be combined to obtain multiple protocols allowing the (secret) distribution of a \GHZ{} state between $n$ clients. The protocols involve several parties:
\begin{itemize}
  \item There are $n$ \arxivOnly{(weakly\footnote{The applicants need only basics quantum skills, depending on the protocol, they may have nothing to do except receiving a state, or may need to apply a few gates.}) }quantum \emph{applicants} $a_i$. At the end of the protocols, a (secret) subset $\cS$ of these applicants---the \emph{supported}\footnote{We also call them \emph{supported} because they are part of the support of the hidden \GHZ{} state. We may also refer to this as being the \emph{support status} of an applicant.} applicants---will share a \GHZ{} state.
  \item A quantum\arxivOnly{\footnote{Bob needs to perform quite heavy quantum computations.}} server, \emph{Bob}. Bob is quantumly connected to all applicants, and it is the only quantum communication required in the upcoming protocols.
  \item \emph{Cupid}\footnote{Besides having a name starting with a 'C', Cupid, the roman god of love, is famous for sending arrows at the heart of Humans to designate the beloved among the applicants.} is an honest classical party (the only party that will always be considered as honest when present. In our last protocol, we won't need Cupid anymore). Cupid will be in charge of choosing the supported applicants that will share a \GHZ{} at the end of the protocol. We assume Cupid can communicate classically with the server, as well as with all applicants---using secure (i.e.\ confidential\footnote{Since secure channels may still leak the length of the exchanged message, in this paper we assume that for a given round, all messages have the same size.} and authenticated) channels.\\
        Note that if the supported applicants want to perform at the end some quantum protocols in order to use the previously shared \GHZ{} state, \emph{Cupid} may also be used as a proxy forwarding the messages to the good applicants in case classical communication is required between participants\footnote{In that case, it is also important to make sure that no adversary can detect the source and the destination of a message\publishedVsArxiv{.}{, otherwise it would be trivial for the adversary to obtain the set of supported applicants. This could be achieved by either asking to all non-supported applicants to also send/receive dummy encrypted messages to Cupid when a message is expected to be exchanged (it is the simplest solution but may not be the most efficient solution), by using a classical channel that is not controlled by the adversary to exchange these messages, by using a noisy cover channel in which communications can be made undetectable, or by relying on an anonymous onion-like routing protocol if the adversary controls only parts of the network and if enough unrelated messages are exchanged this way to make sure that the adversary cannot apply timing attacks/flow correlation attacks\dots{}}}.
\end{itemize}

Note also that these parties may not be always different entities. For example, when a user wants to send a qubit to a secret recipient, this user could be both considered as an applicant and as Cupid. Similarly, the server may want to be part of the applicants. Moreover, in our final protocol, the role of Cupid is distributed among all applicants: each applicant will be able to choose whether they want to be part of the \GHZ{} or not, and in that setting all parties may potentially be malicious. However, since the simpler protocols could be of independent interest, we also study them separately.

\subsection{List of the protocols proposed in this work}\label{subsec:listProtocols}

In the following, a canonical \GHZ{} state will be a state of the form $\frac{1}{\sqrt{2}}(\ket{0 \dots 0} + \ket{1 \dots 1})$ (we will drop all normalization factors later), a generalized \GHZ{} state is a \GHZ{} state in which we applied some local $X$ or $Z$ gates, and a hidden \GHZ{} is a permutation of a generalized \GHZ{} state tensored with some $\ket{0}$ and $\ket{1}$ qubits.

We propose $4$ different protocols in \cref{sec:protocols} (\blind{}, \blindSup{}, \blindCanSup{} and \authBlindCanDist{}) to distribute a \GHZ{} state. Note that we show in \cref{sec:NIZKoQS} how \blind{} and \authBlindCanDist{} can exhibit NIZKoQS.

\authBlindCanDist{} is the most involved protocol, and we also show \cref{subsec:blindCan} the impossibility of a desirable variant of these protocols, \blindCan{}. All these protocols derive from the \blind{} protocol, in which Cupid chooses the support status of each applicant, and at the end of the protocol each supported applicant is supposed to obtain a generalized \GHZ{} state, while non-supported applicant obtain random not entangled qubits in the computational basis (at that step no applicants know if they are supported or not). The server is simply used as an intermediate party in charge of creating and distributing states. The other protocols differ slightly:
\begin{itemize}
  \item The subscript $\cdotwild_{\tt can}$ denotes the fact that at the end of the protocol each supported applicant ends up with a canonical \GHZ{} state instead of a generalized \GHZ{} state.
  \item The superscript $\cdotwild^{\tt sup}$ denotes the fact that at the end of the protocol each applicant knows their own support status, chosen by Cupid.
  \item In the protocol \authBlindCanDist{} each applicant chooses their own support status (Cupid is not required anymore), and the server can check that each participant is ``authorized'', meaning that if they are supported, then they know a secret (note that this does \emph{not} reveal to the server if a given applicant is supported). To obtain this last property, we use NIZKoQS.
\end{itemize}
In term of security, we typically expect that no malicious group of applicants, potentially colluding with the server, should learn the support status of honest applicants\footnote{In \blindCanSup{} and \authBlindCanDist{}, we expect at least one supported applicant to be honest when the adversary is allowed to corrupt supported applicants (the identity of this honest applicant may be unknown to the adversary), otherwise there is a trivial attack against any such protocol.}. In \blind{}, we can even prove that applicants can't even learn their own support status. Moreover, in \authBlindCanDist{} we show that an honest server is guaranteed that any supported applicant ``knows'' the secret, and that no malicious or noisy applicant can significantly alter the state obtained by honest applicants.

A simple application of the \authBlindCanDist{} protocol would be to allow Bob to teleport a quantum state $\ket{\psi}$ to an unknown applicant who knows a secret password, or who has access to a signature provided by a Certification Authority. The applicant would be allowed to hide its identity to Bob, and Bob can be sure that only applicants knowing the password/signature could obtain the state $\ket{\psi}$. In addition, if several applicants know the secret, then Bob is in fact secret sharing its qubit $\ket{\psi}$ among all applicants knowing this secret~\cite{hillery1999quantum}. Other applications are discussed at the end of \cref{sec:relatedWorks}.

\subsection{Quick overview of our methods}\label{subsec:quick_overview}

\noindent\textbf{Cryptographic assumptions.}\quad We need to use a classical cryptographic family of functions $\{f_k\colon \cX \rightarrow \cY\}_{k \in \cK}$, together with a function $h\colon \cX \rightarrow \{0,1\}^n $ having several properties. The exact list of requirements is given in \cref{sec:crypto_requirements}, but here are the important properties. For any $\vect{d}_0 \in \{0,1\}^n$ (intuitively representing the support status of all applicants, with $\vect{d}_0[i] = 1$ iff the $i$th applicant is supported), we can generate using a function $\Gen(1^\lambda, \vect{d}_0)$ an index $k$ and a trapdoor $t_k$ such that:
\begin{itemize}
  \item $f_k$ is $2$-to-$1$ (i.e.\ for all $x$, $|f_k^{-1}(f_k(x))| = 2$\arxivOnly{, we will also generalize this definition to approximate $\delta$-$2$-to-$1$ functions}).
  \item $f_k$ can be efficiently computed given $k$, but should be hard to invert without $t_k$. Moreover, it should be hard to obtain any information on $\vect{d}_0$ given $k$.
  \item Given the trapdoor $t_k$, $f_k$ can be efficiently inverted.
  \item For any $x \neq x'$ such that $f(x) = f(x')$, $h(x) \xor h(x') = \vect{d}_0$.
\end{itemize}
We will say that such a family is \AssumpFct{} (more details in \cref{def:fct_requirements}). While these properties will be enough to understand our ``binding'' approach, in practice we will also require other properties in order to allow each applicant to choose their own support status. More precisely, we will require the existence of a local generation procedure $(k^{(i)},t_k^{(i)}) \gets \GenLocal(1^\lambda,\vect{d}_0[i])$ such that $k = (k^{(1)},\dots,k^{(n)})$ (since $k^{(i)}$ is enough to fix $\vect{d}_0[i]$, we may write $\vect{d}_0[i] \eqdef \vect{d}_0(k^{(i)})$). Moreover $t_k^{(i)}$ can be used to obtain partial information about the preimages of $f_k$: one can, for example, obtain the $i$-th bit of $h(x)$ given $f_k(x)$ and $t_k^{(i)}$. These additional properties will be important to allow the creation of a multi-party state, and we will say that such family is a distributable \AssumpFctCan{} family.\\

\noindent\textbf{Efficient large \RSP{}.}\quad The first key point of our method is to bind a classical message to a quantum state. Instead of receiving directly a quantum state, the receiver receives classical instructions encrypting a class of quantum states and derives from it the final quantum state. More precisely, since in our case we are interested in the preparation of hidden \GHZ{} states (i.e.\ a state which is equal, up to some local operations, to a permutation of $(\ket{0 \dots 0} + \ket{1 \dots 1})\otimes \ket{0 \dots 0}$), we proceed as follows: If denote by $\vect{d}_0$  the support of the \GHZ{} state (i.e.\ qubit $i$ is part of the \GHZ{} iff $\vect{d}_0 = 1$), then the encryption of the quantum state will be $k$, were $k$ is generated using $\Gen(1^\lambda, \vect{d}_0)$. Then, in order to produce the quantum state, the receiver will run the unitary corresponding to $x \mapsto (h(x), f_k(x))$ in superposition, obtaining the state:
\begin{align}
  \sum_x \ket{x}\ket{h(x)}\ket{f_k(x)} = \sum_y (\ket{x}\ket{h(x)} + \ket{x'}\ket{h(x')})\ket{y}
\end{align}
where in the right hand side, we sum over the elements $y$ in the image of $f_k$, and $(x,x')$ are the two preimages of $y$ (reminder: the function is $2$-to-$1$). After measuring the last register, we obtain a $y$, and the following quantum state (where $x$ and $x'$ are the two preimages of $y$):
\begin{align}
  \ket{x}\ket{h(x)}+\ket{x'}\ket{h(x')}
\end{align}
Now, we measure the first register in the Hadamard basis, and we obtain:
\begin{align}
  \ket{h(x)}+(-1)^\alpha\ket{h(x')}
\end{align}
for some $\alpha \in \{0,1\}$, depending on the outcome of the measurement. This state is now a hidden \GHZ{} state, whose support is equal to $\vect{d}_0$ due to the property $h(x) \xor h(x') = \vect{d}_0$. Indeed, when $\vect{d}_0[i] = 0$, $h(x)[i] = h(x')[i]$, so we can factor out this qubit. More precisely, in order to come back to a canonical \GHZ{} on the other qubits, one just needs to apply an $X$ gate on the $i$-th qubit for all $i$ such that $1=h(x)[i] \neq h(x')[i]$ (the trapdoor would be required to learn $h(x)$), followed by a $Z$ operation on one supported qubit if $\alpha = 1$.

Note that here we managed to prepare a $n$-qubit state using a single superposition, which improves the efficiency of existing \RSP{} protocols. With the actual implementation of $f_k$, it drops the required number of steps from $O(nMN)$ to $O((M+n)N)$ where $M > N$, $N$ and $M$ being very large parameters of $f_k$.\\

\noindent\textbf{NIZQoQS.}\quad It is now possible to turn our Non-Interactive \RSP{} protocol into a NIZKoQS in order to prove properties on the generated state. The key element to note is that a single \emph{classical} message $k$ is sent to the receiver. This message can be interpreted as the instructions to follow to produce a given hidden \GHZ{} state whose support is $\vect{d}_0$, and it is therefore possible to send a classical NIZK proof---using your favorite post-quantum NIZK construction---about these instructions to indirectly (and therefore non-destructively) obtain guarantees on the produced quantum state. That way, the sender simply needs to send, together with $k$, a NIZK proof (where $t_k$ is part of the secret witness) proving that \emph{(i)} the message $k$ is indeed a \AssumpFct{} function\footnote{In our construction, it boils down to proving that the trapdoor $t_k$ has small enough singular values.} \emph{(ii)} that $\Auth(\vect{d}_0,w) = 1$, where $\Auth$ can be any efficiently computable function, and $w$ any secret witness, depending on the wanted property on $\vect{d}_0(t_k)$. This last function and witness could be virtually anything, like ensuring that \emph{There exists only two indices $i \neq j$ such that $\vect{d}_0(i)=\vect{d}_0(j) = 1$} (i.e.\ it proves that the final state contains only two entangled qubits forming a Bell state, $w$ is not needed here), or \emph{Either $\vect{d}_0[42]=0$ or ($\vect{d}_0[42]=1$ and I know the private key corresponding to the bitcoin wallet \href{https://www.blockchain.com/btc/address/12c6DSiU4Rq3P4ZxziKxzrL5LmMBrzjrJX}{12c6DSiU4Rq3P4ZxziKxzrL5LmMBrzjrJX})} (i.e.\ it proves that the 42th qubit is entangled to the rest of the \GHZ{} only if the sender is Satoshi Nakamoto\dots{} of course without revealing to the receiver if the sender is Satoshi Nakamoto; $w$ being here the private key of the wallet). This kind of property will be particularly interesting in the \authBlindCanDist{} protocol.\\

\noindent\textbf{Usage in protocols.}\quad We can use the above ideas in order to achieve the above protocols, notably \authBlindCanDist{}. Each applicant will be asked to sample $(k^{(i)},t_k^{(i)}) \gets \GenLocal(1^\lambda, \vect{d}_0)$, then $k^{(i)}$ will be sent to the server. In order to prove that they are authorized, each applicant will also include a NIZK proof confirming, as explained in the last paragraph, that they know a classical witness $w_i$ such that $\Auth(\vect{d}_0[i],w_i) = 1$ (the NIZK proof also checks that $\vect{d}_0[i]=\vect{d}_0(k^{(i)})$ and that $k^{(i)}$ is honestly prepared). Since this protocol is a Zero-Knowledge protocol, the server will not be able to learn any information about $\vect{d}_0$. In return, the server will have the guarantee that it will indeed produce a hidden \GHZ{} state distributed among authorized applicants. Therefore, the server can run the quantum circuit described above, and distribute each qubit to the corresponding applicant. In order to come back to a canonical \GHZ{}, each applicant will use their local trapdoor $t_k^{(i)}$ to compute $h(x)$ in order to apply the corresponding $X$ correction. Moreover, in order to compute the $\alpha$ needed to apply the $Z$ correction, all parties will need to run a Multi-Party Computation (MPC). Since $\alpha$ could leak some information about the state, each supported applicant will in fact obtain from the MPC a linear secret sharing of $\alpha$, i.e.\ a random $\hat{\alpha}_i$ such that $\xor_{i \in \cS} \hat{\alpha}_i = \alpha$ ($\cS$ being the set of supported applicants). That way, each supported applicant will be able to locally apply a $Z^{\hat{\alpha}_i}$ correction, and it will be equivalent to applying a single $Z^\alpha$ gate on the overall state.\\

\noindent\textbf{How to construct the cryptographic assumptions.}\quad In order to implement the family $\{f_k\}_{k \in \cK}$, we proceed in two steps: we first build a family which is \AssumpFct{} (see intuitive definition above, or \cref{def:fct_requirements}). The $\delta$ is needed because in practice the functions are only $2$-to-$1$ with a probability $1-\delta$, where $\delta$ can be made negligible by relying on the hardness of $\GapSVP{}_\gamma$ with a superpolynomial ratio $\gamma$. We then provide a quite generic method to turn a \AssumpFct{} family into a \AssumpFctCanPrime{} distributable family.

The security of our construction reduces to the hardness of the \LWE{} assumption (more details in \cref{subsec:LWE}), and builds on ideas introduced in \cite{CCKW_2019_qfactory}. The starting point is the trapdoor construction provided by \cite{MP11}. They provide an algorithm to generate a matrix $\vect{A} \in \Z_q^{M \times N}$ ($q$ will be a power of two) that looks random, together with a trapdoor matrix $\vect{R}$. If the noise $\vect{e} \in \Z_q^M$ is sufficiently small\arxivOnly{\footnote{In the actual construction, we also require $\vect{s}$ to be small because we rely on the equivalent but more efficient normal-form of \LWE{}, but for simplicity we use the classic \LWE{} problem in this overview.}}, the function $g_{\vect{A}}(\vect{s}, \vect{e}) \eqdef \vect{A}\vect{s} + \vect{e}$ is injective. Moreover, given the trapdoor $\vect{R}$, one can easily invert the function $g_{\vect{A}}$, otherwise inverting this function is hard: $g_{\vect{A}}(\vect{s}, \vect{e})$ is indistinguishable from a random vector given $\vect{A}$.

From that, we can first see how to get a $\delta$-$2$-to-$1$ family of functions. Note that the larger the noise $\vect{e}$ is, the larger $\delta$ is. So a perfect (but insecure) $2$-to-$1$ family would use $\vect{e}=\vect{0}$: therefore, to better understand this construction, it may be practical to imagine that $\vect{e}=\vect{0}$ during a first reading. The idea of the construction is to sample first an image vector $\vect{y}_0 \eqdef \vect{A}_u \vect{s}_0 + \vect{e}_0 \in \Z_q^M$ (by sampling first $\vect{s}_0$ uniformly at random over $\Z_q^N$ and $\vect{e}_0 \in \Z_q^m$ according to a discrete Gaussian $\cD_{\Z, \alpha q}^m$), and then to define for all $\vect{s} \in \Z_q^N$, $c \in \{0,1\}$ and errors $\vect{e}$ belonging to some set $E \subseteq \Z_q^M$ (to be defined), the function $f_{\vect{A}_u,\vect{y}_0}(\vect{s}, \vect{e} ,c \in \{0,1\}) \eqdef \vect{A}_u \vect{s} + \vect{e} + c \times \vect{y}_0$. That way, if all vectors in $E$ are small enough, $f_{\vect{A}_u,\vect{y}_0}$ has at most two preimages, one for $c=0$ and one for $c=1$: $f(\vect{s},\vect{e},0) = \vect{A}_u \vect{s}+\vect{e} = \vect{A}_u(\vect{s}-\vect{s}_0) + (\vect{e}-\vect{e}_0) + \vect{y}_0 = f(\vect{s}-\vect{s}_0,\vect{e}-\vect{e}_0,1)$. We remark that in order to have two preimages, we want to make sure that both $\vect{e} \in E$ and $\vect{e}-\vect{e}_0 \in E$: we will rely on \LWE{} with superpolynomial modulus to noise ratio to ensure this is true for a superpolynomially large fraction of inputs. Then, in order to obtain the bit string $\vect{d}_0 \in \{0,1\}^n$ and the function $h$ such that for any two preimages $x,x'$, $h(x) \xor h(x')=\vect{d}_0$, we will
update the previous construction and now sample $\vect{y}_0$ as follows: We will first sample additional lines $\vect{A}_l \sample \Z_q^{n \times N}$
to add to the matrix $\vect{A}_u$. Then, we will sample $\vect{s}_0$ uniformly at random over $\Z_q^N$ (as before) and $\vect{e}_0 \in \Z_q^{M+n}$ will be sampled according to the discrete Gaussian $\cD_{\Z, \alpha q}^{M+n}$. Finally, if we also denote by $\vect{d}_0$ the binary vector in $\Z_q^n$ composed of elements of $\vect{d}_0$, we compute $\vect{y}_0 \eqdef \SmallBlockMatrix{\vect{A}_u}{\highlightChanges{\vect{A}_l}} \vect{s}_0+\vect{e}_0+\highlightChanges{\frac{q}{2} \SmallBlockMatrix{\vect{0}^M}{\vect{d}_0}}$. We update similarly our function $f$ by adding a parameter $\vect{d} \in \{0,1\}^n$:
\begin{align}
f_{\vect{A}_u,\vect{A}_l,\vect{y}_0}(\vect{s},\vect{e},c,\highlightChanges{\vect{d}}) \eqdef \SmallBlockMatrix{\vect{A}_u}{\highlightChanges{\vect{A}_l}} \vect{s} + \vect{e} + \highlightChanges{\frac{q}{2} \SmallBlockMatrix{\vect{0}^M}{\vect{d}}} + c \times \vect{y}_0
\end{align}
Then, we can remark that, because $q$ is even:
\begin{align}
  f_{\vect{A}_u,\vect{A}_l,\vect{y}_0}(\vect{s},\vect{e},0,\vect{d})
  &= \SmallBlockMatrix{\vect{A}_u}{\vect{A}_l} \vect{s} + \vect{e} +\frac{q}{2} \SmallBlockMatrix{\vect{0}^M}{\vect{d}}\\
  &= \SmallBlockMatrix{\vect{A}_u}{\vect{A}_l} (\vect{s}-\vect{s}_0) + (\vect{e}-\vect{e}_0) +\frac{q}{2} \SmallBlockMatrix{\vect{0}^M}{\vect{d} \xor \vect{d}_0} + \vect{y}_0\\
  &= f_{\vect{A}_u,\vect{A}_l,\vect{y}_0}(\vect{s}-\vect{s}_0,\vect{e}-\vect{e}_0,1,\vect{d} \xor \vect{d}_0)
\end{align}
and that therefore (skipping a small technicality) for any two preimages $(\vect{s},\vect{e},0,\vect{d})$ and $(\vect{s}',\vect{e}',1,\vect{d}')$, we have $\vect{d} \xor \vect{d}' = \vect{d}_0$. So we just need to define $h(\vect{s},\vect{e},c,\vect{d}) = \vect{d}$ to get the \XOR{} property. Finally, the indistinguishability property is simple to obtain: since $\vect{A}_u$ is indistinguishable from a random matrix, and $\vect{A}_l$ is actually a random matrix, thus $\vect{A} \eqdef \SmallBlockMatrix{\vect{A}_u}{\vect{A}_l}$ is  indistinguishable from a random matrix. Therefore, under the decision-$\LWE{}_{q,\cD_{\Z, \alpha q}}$ assumption, $\vect{A} \vect{s}_0 + \vect{e}_0$ is indistinguishable from a random vector. Therefore, since adding a constant vector to a uniformly sampled vector does not change its distribution, one cannot distinguish $\vect{A} \vect{s}_0 + \vect{e}_0$ from $\vect{A} \vect{s}_0 + \vect{e}_0 + \SmallBlockMatrix{\vect{0}^M}{\vect{d}_0}$, or from any vector of the form $\vect{A} \vect{s}_0 + \vect{e}_0 + \SmallBlockMatrix{\vect{0}^M}{\vect{d}}$. We provide in \cref{sec:construction_f_superpoly} a more in-depth analysis in order to properly handle the noise, we find an explicit set of parameters allowing $f$ to be $\negl[\lambda]$-$2$-to-$1$, and we give a method to prove that a maliciously sampled $f$ is indeed $\negl[\lambda]$-$2$-to-$1$ and has the \XOR{} property. This is required to turn it into a distributable $\delta'$-\AssumpFctCanNoDelta{} family.

Finally, one can obtain a distributable $\delta'$-\AssumpFctCanNoDelta{} family from a \AssumpFct{} family (which has an additional property that the two preimages have the form $x = (0,\bar{x})$ and $x' = (1, \bar{x}')$). The idea is to sample one \AssumpFct{} function for each $\vect{d}_0[i]$. Then, we can define our new function to be $f_{k^{(1)},\dots,k^{(n)}}(c,\bar{x}^{(1)},\dots,\bar{x}^{(n)}) = f_{k^{(1)}}(c,\bar{x}^{(1)}) |\dots | f_{k^{(n)}}(c,\bar{x}^{(n)})$. More details can be found in \cref{sec:functionConstruction}.

\subsection{Open questions}

Since this work introduces NIZKoQS, there are many questions left open:

\begin{itemize}
\item \textbf{Characterization of the quantum languages verifiable non-interactively}: Using classical-client Universal Blind Quantum Computing (UBQC) and ZK proofs, classical-client Quantum Fully Homomorphic Encryption (QFHE) and ZK proofs, or directly quantum MPC, one could certainly prove ZKoQS statements on arbitrary languages (even using three messages with the QFHE+ZK solution). However, this must be written formally, in particular by proving that revealing additional information regarding the quantum state like the one-time pad does not weaken the security. Moreover, characterizing completely\===with either possibility or impossibility results\===the set of quantum languages that can be verified \emph{non-interactively} is an open question. This paper solves this problem for a specific quantum language, and using QFHE+ZK we could certainly characterize a broader set of quantum languages (basically any property which is true up to a quantum one-time pad, characterizing notably entanglement). However, the question of the feasibility of NIZKoQS for broader quantum languages is left open.
\item \textbf{Minimal assumptions}: In the present paper, we rely heavily on the hardness of the \LWE{} problem. However, studying the properties obtainable with weaker assumptions (like one-way functions) could be a thrilling exploration path.
\item \textbf{Statistical vs computational soundness}: In the protocol presented in this paper, an unbounded verifier could learn completely the produced quantum state. Finding NIZKoQS protocols secure against an unbounded verifier could be another interesting question, with application to statistically secure protocols.
\item \textbf{Applications}: We have shown how NIZKoQS can help boosting the round-efficiency of protocols sharing hidden \GHZ{} states. However, we expect NIZKoQS to be of paramount importance to achieve optimal round complexity in a much wider variety of protocols as they basically avoid cut-and-choose approaches. Similarly, applications of our protocols sharing \GHZ{} states should be further explored.
\item \textbf{\RSP{} efficiency and applications}: We showed how \RSP{} protocols can produce multi-qubits states with an asymptotically quadratic improvement. However, further improving this bound (potentially using \LWE{} with a polynomial noise ratio) and finding multi-qubit states that are both efficient to prepare and compatible with blind quantum computing is an exciting challenge.
\end{itemize}

\section{Preliminaries}\label{sec:preliminaries}
\subsection{Notation}\label{sec:notation}
The applicants and Cupid are designated with the gender neutral pronoun \emph{they}, while the server---a machine---is designated by \emph{it}.

If $\vect{b} \in \{0,1\}^n$ is a bit string or a vector, and $i$ is an integer, we may use $\textbf{b}_i$ or
\arxivOnly{\footnote{Note however that we often refer in this paper to a bit string $\vect{d}_0 \in \{0,1\}^n$ (we use this notation in order to keep a consistent notation with previous works), so to avoid confusion between $\vect{d}_0$ (the bit string) and $\vect{d}_0$ (the first bit of $d$), we will always use the bracket notation $\vect{d}[i]$ or $\vect{d}_0[i]$ when the $d$ letter is involved. Since $\vect{d}_0$ will also be interpreted as a vector in the second part of the paper, and to highlight the difference with a previous work in which $d_0$ was a bit string, we will use a bold font for this kind of bit string.}}
$\vect{b}[i]$ to denote the $i$-th bit of $\vect{b}$. The symbol $|$ will be used to denote string concatenation, and we define $\mybar{\vect{b}} \eqdef (b_1\xor 1) | \dots | b_n\xor 1)$.\arxivOnly{ We define also $X^b \eqdef X^{b[1]} \otimes X^{b[2]} \dots \otimes X^{b[n]}$ and $Z_i$ is the operator that applies the $Z$ gate on the $i$-th qubit, i.e.\ $Z_i = \underbrace{I \otimes \dots \otimes I}_{i-1} \otimes Z \otimes I \otimes\dots \otimes I$. For any permutation $\sigma$, we denote by $\PERM_\sigma$ the quantum unitary that performs a permutation of the qubits, such that the $i$-th qubit is mapped on $\sigma(i)$.} For any bit strings $x$, $x'$ and $b$ of same length, we also define $\langle b, x\rangle \eqdef \bigoplus_i b_i x_i$,  $x \xor x' \eqdef (x_1 \xor x'_1) | \dots | (x_n \xor x'_n) $.

\arxivOnly{
  $\R$ is the set of reals, $\Z$ the set of integers,
} $\Z_q$ is the set of integers modulo $q$ and $[n] \eqdef \{1,\dots,n\}$. For any element $x \in \Z_q$, if $\hat{x}$ is the representant of $x$ in $[-\frac{q}{2}, \frac{q}{2})$ we will define the modular rounding $\RoundMod_q(x) = 0$ if $\hat{x} \in [-\frac{q}{4}, \frac{q}{4})$ and $\RoundMod_q(x) = 1$ otherwise. We extend this notation when $\vect{x}$ is a vector by doing it component wise. Similarly, for any element $\vect{x} \in \Z_q^n$, $\|\vect{x}\|_\infty$ will be understood in a modular way: $\|x\|_\infty \eqdef \max_i \hat{x}[i]$. For any vector $\vect{x} \in \R^n$, $\vect{x}^T$ is the transpose operation, $\|\vect{x}\|_2 \eqdef \sqrt{\vect{x}^T \vect{x}}$ is the Euclidean norm of $\vect{x}$. For any matrix $\vect{R} \in \R^{n\times m}$, $\|\vect{R}\|_2 \eqdef \sup_{\vect{x}} \|\vect{R}\vect{x}\|_2/\|\vect{x}\|_2$ is the spectral norm of $\vect{R}$, and $\sigma_{\max}(\vect{R})$ is the highest singular value of $\vect{R}$. Note that $\sigma_{\max}(\vect{R}) = \|\vect{R}\|_2$. For an Hilbert space $\cH$, $\linearOp_\circ(\cH)$ is the set of density operators on $\cH$ (positive linear operators on $\cH$ with trace $1$).

\PPT{} stands for Probabilistic Polynomial Time and \QPT{} stands for Quantum Polynomial Time. We denote by $\negl[\lambda]$ any function that goes asymptotically to $0$ faster than any inverse polynomial in $\lambda$: we say that this function is \emph{negligible}. We say that a probability is \emph{overwhelming} when it is $1-\negl[\lambda]$. We use the Landau notation $O(\cdot)$, and we say that $f(x) = \tilde{O}(g(x))$ if there exists $k$ such that $f(x) = O(g(x) \log^k g(x))$. We use $\omega(\sqrt{\log N})$ to denote a function, fixed across all the paper, such that $\lim_{N \rightarrow \infty} \sqrt{\log N}/\omega(\sqrt{\log N}) = 0$. For instance, we can take $\omega(\sqrt{\log N}) = \log N$.

For any finite set $X$, we denote $x \sample X$ when $x \in X$ is sampled uniformly at random. The uniform distribution on $X$ is also denoted by $\cU(X)$, and the statistical distance between the distributions $\cD_1$ and $\cD_2$ is $\Delta(\cD_1,\cD_2) \eqdef \frac{1}{2} \sum_{x}|\cD_1(x)-\cD_2(x)|$. For any distribution $\cD$, we write $x \gets \cD$ if $x$ is sampled according to $\cD$\arxivOnly{, and $y \eqdef f(x)$ when $f$ is deterministic or if it is the definition of $y$}. If an adversary $\cA$ interacts $n$ times with its environment, without loss of generality we may decompose $\cA = (\cA_1, \dots, \cA_n)$ and denote $y_1 \leftarrow \cA_1(x_1); y_2 \leftarrow \cA_2(x_2); \dots; y_n \leftarrow \cA_n(x_n)$ where we implicitly assume that $\cA_{i}$ gives its internal state to $\cA_{i+1}$. In the hybrid games, we will also strike out some lines to denote differences between two games: it only removes the content of the current line, without changing the line numbering of the other lines.


\subsection{Generalized hidden \GHZ{} state}\label{subsec:defGHZ}

In this paper we assume basic familiarity with quantum computing~\cite{nielsen2002quantum}.
\publishedVsArxiv{
  A \emph{canonical \GHZ{}} (Greenberger--Horne--Zeilinger) quantum state~\cite{GHZ_1989_going} is a state of the form $\frac{1}{\sqrt{2}}( \ket{0 \dots 0} + \ket{1 \dots 1})$ (for simplicity, we will drop all the normalization factors). A \emph{generalized \GHZ{}} state $\ket{\GGHZ{}_{\alpha,\vect{d}}} \eqdef \ket{\vect{d}} + (-1)^\alpha \ket{\mybar{\vect{d}}}$ is a \GHZ{} state on which we applied some $X$ or $Z$ local gates, $(\alpha,\vect{d})$ being called the \emph{key}. A \emph{hidden (generalized) \GHZ{}} $\ket{\HGHZ{}_{\alpha,\vect{d},\vect{d}'}} \eqdef \ket{\vect{d}} + (-1)^\alpha \ket{\vect{d}'}$ is a permutation of a generalized \GHZ{} tensored with some $\ket{0}$ or $\ket{1}$ qubits. The set of qubits $\{i \mid \vect{d}[i] \neq \vect{d}'[i]\}$ (or simply $\vect{d}_0 \eqdef \vect{d} \xor \vect{d}'$) is called the \emph{support} of the \GHZ{} because they all belong to the original generalized \GHZ{}.
}{
  The \GHZ{} (Greenberger--Horne--Zeilinger) quantum state~\cite{GHZ_1989_going} is a generalization of the Bell state on multiple qubits. This state turns out to be useful in many applications, going from Quantum Secret Sharing~\cite{hillery1999quantum} to Anonymous Transmission~\cite{CW_2005_quantum_anonymous_transmissions}.

  \begin{definition}[(Canonical) \GHZ{} state]
    We denote by $\ket{\GHZ{}_n}$ a $n$ qubits (canonical) \emph{\GHZ{} state}, i.e.\ a state of the form:
    \[\ket{\GHZ{}_n} \eqdef \frac{1}{\sqrt{2}}(\ket{00\dots0} + \ket{11\dots1}) \]
  \end{definition}

  \begin{definition}[\GGHZ{}: Generalized \GHZ{} state]
    A quantum state on $n$ qubits is said to be a \emph{generalized \GHZ{} state} (sometimes abbreviated as \GGHZ{}), and usually denoted $\ket{\GGHZ{}_{\alpha,d}}$ if it is a \GHZ{} state up to an eventual $\pi$ phase, and up to local $X$ gates , i.e.\ if there exist a bit string $\vect{d} \in \{0,1\}^n$ and $\alpha \in \{0,1\}$ such that:
    \[\ket{\GGHZ{}_{\alpha,\vect{d}}} \eqdef X^{d}Z_1^\alpha\ket{\GHZ{}_n} = \frac{1}{\sqrt{2}}(\ket{\vect{d}} + (-1)^\alpha \ket{\mybar{\vect{d}}}) \]
    We call $(\alpha,\vect{d})$ the \emph{\GHZ{} key}, since with the help of $(\alpha,\vect{d})$ it is possible to turn a generalized \GHZ{} state into a canonical \GHZ{} state using only local $X$ and $Z$ gates. The term ``key'' is used since the generalized \GHZ{} state can be seen as a kind of one-time padded (OTP) canonical \GHZ{} state.
  \end{definition}

  \begin{definition}[{\HGHZ{}: Hidden generalized \GHZ{} state}]
    A quantum state $\ket{\phi}$ on $n$ qubits is said to be a \emph{hidden (generalized) \GHZ{} state} (sometimes abbreviated as \HGHZ{}) if it is a permutation of a state composed of one generalized \GHZ{} state and the tensor product of qubits in the computational basis, i.e.\ if there exist a permutation $\sigma$, a bit $\alpha \in \{0,1\}$, a bit string $\vect{d}$, and two integers $(l, m) \in \N^2$ such that:
    \[\ket{\phi} = \PERM_\sigma(\ket{\GGHZ{}_{\alpha,\vect{d}}} \otimes \ket{0}^{l} \otimes \ket{1}^{m}) \]
    Equivalently, it means that there exist two bit strings $(\vect{d}, \vect{d}') \in (\{0,1\}^n)^2$ and a bit $\alpha \in \{0,1\}$ such that:
    \[\ket{\phi} = \frac{1}{\sqrt{2}} (\ket{\vect{d}} + (-1)^\alpha \ket{\vect{d}'}) \eqqcolon \ket{\HGHZ{}_{\alpha,\vect{d},\vect{d}'}}\]
    For such a state $\ket{\HGHZ{}_{\alpha,\vect{d},\vect{d}'}}$, the qubits that are not in the computational basis (and therefore that are part of the original \GGHZ{} state) are said to be in the \emph{support} of the \HGHZ{} state, and it is easy to see that the support is:
    \[\supp(\ket{\HGHZ{}_{\alpha,\vect{d},\vect{d}'}}) \eqdef \{i \mid \vect{d}[i] \neq \vect{d}'[i] \}\]
    We will assimilate the string $\vect{d}_0 \eqdef \vect{d} \xor \vect{d}'$ (bitwise \XOR{}) to the support of the \HGHZ{} since $\vect{d}_0[i] = 1$ iff $i \in \supp(\ket{\HGHZ{}_{\alpha,\vect{d},\vect{d}'}})$. Moreover, we will call $(\alpha,\vect{d},\vect{d}')$ the \emph{\GHZ{} key} since it is possible the come back to a canonical \GHZ{} state tensored with some $\ket{0}$ qubits by using $(\alpha,\vect{d},\vect{d}')$.
  \end{definition}
}
\subsection{Classical Zero-Knowledge proofs and arguments for \npol{}}\label{subsec:ZK}
\pgfkeys{/prAtEnd/local custom defaults/.style={category=introZK}}

In the first (NIZKoQS) and last protocol presented in this paper, we need to use a classical Zero-Knowledge (ZK) protocol for \npol{} (to obtain NIZKoQS, we also expect the protocol to be non-interactive, but interactivity does not change security or correctness). Intuitively, in a ZK protocol for a language $\lang \in \npol{}$, a prover must convince a verifier that a word $x$ belongs to $\lang$, in such a way that the verifier should not learn anything more about $x$ beyond the fact that $x$ belongs to $\lang$. Because $\lang$ is in \npol{}, $\lang$ is described by a relation $\cR_\lang$, in such a way that a word $x$ belongs to $\lang$ iff there exists a witness $w$ such that $w \in \cR_\lang(x)$. Moreover, deciding if a witness $w$ belongs to $\cR_\lang(x)$ must be doable in polynomial time.

We will now formalize the above security statements, taking definition from \cite{BS_2019_ZK_Ct_rounds}. But first, we introduce a few notations: if $(\P{},\V{})$ is a protocol between two parties $\P{}$ and $\V{}$, we denote by $\OUT_\V{}\langle \P{},\V{} \rangle(x)$ the output of the party $\V{}$ after having followed the protocol with $\P{}$, $x$ being a common input of $\P{}$ and $\V{}$. An honest verifier outputs a single bit ($1$ if they accept and $0$ if they reject), but a malicious verifier can output an arbitrary quantum state.

\begin{definition}[Computational indistinguishability~\cite{BS_2019_ZK_Ct_rounds}]
  Two maps of quantum random variables $X \eqdef \{X_i\}_{\lambda \in \N, i \in I_\lambda}$ and $Y \eqdef \{Y_i\}_{\lambda \in \N, i \in I_\lambda}$ over the same set of indices $I = \cup_{\lambda\in\N} I_\lambda$ are said to be \emph{computationally indistinguishable}, denoted by $X \approx_c Y$, if for any non-uniform quantum polynomial-time distinguisher $D \eqdef \{(D_\lambda, \rho_\lambda)\}_{\lambda \in \N}$, there exists a negligible function $\mu$ such that for all $\lambda \in \N$, $i \in I_\lambda$,
  \begin{align}
    |\pr{D_\lambda(X_i,\rho_\lambda)=1} - \pr{D_\lambda(Y_i, \rho_\lambda)=1}| \leq \mu(\lambda)
  \end{align}
\end{definition}

\begin{definition}[{Post-Quantum Zero-Knowledge Classical Protocol \cite[definitions 2.1 and 2.6]{BS_2019_ZK_Ct_rounds}}]\label{def:postquantumZK}
  Let $(\P{},\V{})$ be a protocol between an honest \PPT{} prover $\P{}$ and an honest \PPT{} verifier $\V{}$. Then $(\P{},\V{})$ is said to be a Post-Quantum Zero-Knowledge (ZK) Classical Protocol for a language $\lang \in \npol{}$ if the following properties are respected:
  \begin{enumerate}
    \item \textbf{Perfect Completeness}: For any $\lambda \in \N$, $x \in \lang \cap \{0,1\}^\lambda, w \in \cR_\lang(x)$,
          \begin{align}
            \pr{\OUT_\V{}\langle \P{}(w), \V{}\rangle(x) = 1} = 1
          \end{align}
    \item \arxivOnly{\textbf{Soundness}: The protocol satisfies one of the following.}\arxivOnlyNoDiff{
          \begin{itemize}
            \item }\textbf{Computational Soundness}: For any non-uniform \QPT{} malicious prover $\P{}^* = \{(\P{}^*_\lambda, \rho_\lambda)\}_{\lambda \in \N}$, there exists a negligible function $\mu(\cdot)$ such that for any security parameter $\lambda \in \N$ and any $x \in \{0,1\}^\lambda \setminus \lang$,
                  \begin{align}
                    \pr{\OUT_\V{}\langle \P{}^*_\lambda(\rho_\lambda), \V{}\rangle(x)=1}\leq \mu(\lambda)
                  \end{align}
                  A protocol with computational soundness is called an argument.
            \arxivOnly{\item \textbf{Statistical Soundness}: There exists a negligible function $\mu(\cdot)$ such that for any unbounded prover $\P{}^*$, any security parameter $\lambda \in \N$ and any $x \in \{0,1\}^\lambda \setminus \lang$,
                  \begin{align}
                    \pr{\OUT_\V{}\langle \P{}^*, \V{}\rangle(x)=1}\leq \mu(\lambda)
                  \end{align}
                  A protocol with statistical soundness is called a proof.}
          \arxivOnlyNoDiff{\end{itemize}}
    \item \textbf{Quantum Zero Knowledge}: There exists a \QPT{} simulator $\Sim$ such that for any \QPT{} verifier $\V{}^* = \{(\V{}^*_\lambda,\rho_\lambda)\}_{\lambda \in \N}$,
          \begin{align}
            \{\OUT_{\V{}^*_\lambda}\langle \P{}(w), \V{}^*_\lambda(\rho_\lambda)\rangle(x)\}_{\lambda,x,w} \approx_c \{\Sim(x,\V{}^*_\lambda,\rho_\lambda)\}_{\lambda,x,w}
          \end{align}
          where $\lambda \in \N$, $x \in \lang \cup \{0,1\}^\lambda$, $w \in \cR_\lang(x)$, and $\V{}^*$ is given to $\Sim$ by sending the circuit description of $\V{}^*$.
  \end{enumerate}
  A Non-Interactive ZK protocol will be denoted NIZK.
\end{definition}

In our last protocol, in order to get stronger guarantees, we may also want to ensure that the prover ``knows'' the secret. This is known as Proof of Knowledge, but since we do not really rely on it in the proofs, we only describe it in more details in \cref{def:QPoK}.
\inAppendix{
  In our last protocol, in order to get stronger guarantees, we may also want to ensure that the prover ``knows'' the secret. Therefore, to get stronger guarantees, we will require the ZK protocol to also be a Proof of Knowledge protocol. Intuitively, we would like to check that any malicious prover $\P^*$ that can convince a verifier with a non-negligible probability has the witness $w$ ``encoded in its source code or memory''. We formalize this notion by saying that there exist a \QPT{} circuit \K, the extractor, which can recover $w$ with non-negligible probability from a full description of $\P^*$ and its input.

  This is usually enough, since being able to obtain a witness with non-negligible probability is usually enough to break the security: for example it could be used to forge a signature and break the unforgeability property as explained in \cite{Unruh_2012_Proof_of_Knowledge}.
  More precisely:

  \begin{definition}[{Post-Quantum Zero-Knowledge Proof of Knowledge \cite{Unruh_2012_Proof_of_Knowledge}}]\label{def:QPoK}
    We say that a Post-Quantum Zero-Knowledge protocol $(\P{},\V{})$ for a relation $\cR_\lang$ is a \emph{Proof of Knowledge} protocol, if it is \emph{quantum extractable} with knowledge error $\kappa = \negl[\lambda]$, i.e.\ if there exists a constant $d > 0$, a polynomially-bounded function $p > 0$, and a \QPT{} $\K$ such that for any interactive \QPT{} malicious prover $P^*$, any polynomial $l$, any security parameter $\lambda \in \N$, any state $\rho$, and any $x \in \{0,1\}^\lambda$, we have:
    \begin{align}
      &\pr{\OUT_V\langle\P^*(\rho), \V\rangle(x) = 1} \geq \kappa(\lambda) \nonumber{} \\
      \Longrightarrow&\pr{w \in \cR_\lang(x) \middle| w \gets \K(\P^*,\rho,x)} \nonumber{}
      \geq \frac{1}{p(\lambda)}\left(\pr{\OUT_V\langle\P^*(\rho), \V\rangle(x) = 1}-\kappa(\lambda)\right)^d
    \end{align}
  \end{definition}
}

Several Post-Quantum ZK protocols have been proposed in the literature \cite{Watrous_2005_ZK_against_q_attacks,Unruh_2012_Proof_of_Knowledge,BS_2019_ZK_Ct_rounds} and have been shown to obey properties similar to both \cref{def:postquantumZK,def:QPoK}. Moreover, \cite{LZ19_RevisitingPostquantumFiatShamir,DFMS19_SecurityFiatShamirTransformation} explain how to obtain quantum-secure NIZK (which are also Proof of Knowledge) using the Fiat-Shamir transformation and the hardness of the \LWE{} problem in a Quantum Random Oracle model. In the following, we are agnostic of the used NIZK protocol and we assume the existence of a NIZK protocol obeying \cref{def:postquantumZK,def:QPoK}.

\subsection{Classical Multi-Party Computations}\label{subsec:MPC}
\pgfkeys{/prAtEnd/local custom defaults/.style={category=introMPC}}

In the last protocol presented in this paper, we also need to use a classical\footnote{But of course post-quantum secure.} Multi-Party Computation (MPC) protocol $\Uppi$. A MPC protocol works as follows: given $n$ (public and deterministic) functions $(f_1,\dots,f_n)$, at the end of the protocol involving $n$ parties $\P{}_1,\dots,\P{}_n$, we expect party $\P{}_i$ to get $f_i(x_1,\dots,x_n)$, where $x_j$ is the (secret) input of the party $\P{}_j$. Moreover, we expect that no party can learn anything more than what they can already learn from $f_i(x_1,\dots,x_n)$. For simplicity, we define $f(x_1,\dots,x_n) = (f_1(x_1,\dots,x_n),\dots,f_n(x_1,\dots,x_n))$.

We can formalize the above security statements using the usual (quantum) real/ideal world paradigm.%
\publishedOnly{
  We will use the precise definition given in~\cite{ABG+21_PostQuantumMultiPartyComputation}, but for completeness it has been copied in \cref{appendix:introMPC}. However, the definition is in spirit close to the Zero-Knowledge property: a MPC protocol is said to be secure if:
  \begin{align}
    \{\Real_{\Uppi,\cA}(\lambda,\vec{x},\rho_\lambda)\}_{\lambda \in \N} \approx_c \{ \Ideal_{f,\Sim}(\lambda,\vec{x},\rho_\lambda) \}_{\lambda \in \N}
  \end{align}
  where  $\vec{x}$ contains the inputs of the parties, $\rho_\lambda$ is a non-uniform advice, $\Real_{\Uppi,\cA}$ is the output of the adversary $\cA$ in a real protocol execution ($\cA$ can corrupt arbitrary many parties at the beginning of the protocol to fully control them), and $\Ideal_{f,\Sim}$ is the view produced by a simulator (that can depend on $\cA$) in an ``ideal world''; this simulator being also able to corrupt arbitrary many parties in order to have access to the inputs/outputs of $f$ for these parties. This definition captures the fact that no attack can do more damage than what is possible by changing the inputs of corrupted parties and reading their outputs. See \cref{appendix:introMPC} for more details.
}%
\inAppendixIfPublished{
  Informally, the protocol $\Uppi$ will be said to be secure if it is impossible to distinguish two ``worlds''. On the one hand, we have a \emph{real world} in which an adversary $\cA = \{\cA_\lambda\}_{\lambda \in \N}$ can corrupt a subset $\cM \subsetneq [n]$ of parties and interact in an arbitrary way with the other honest parties. On the other hand, we have an \emph{ideal world}, in which a simulator $\Sim$ interacts with a functionality, this functionality behaving as a trivially-secure trusted third-party. If these two worlds are indistinguishable, it means that the protocol is ``secure'' because any secret obtained from the real world would also be obtainable from the ideal world\dots{} which is impossible because the ideal world is trivially secure.

  More precisely, we define, following \cite{ABG+21_PostQuantumMultiPartyComputation}, the real and ideal world as follows, where $\vec{x} \eqdef (x_1,\dots,x_n)$ is the inputs of the parties:

  \begin{definition}[$\Real_{\Uppi,\cA}(\lambda,\vec{x},\rho_\lambda)$]
    $\cA_\lambda$ is given $\rho_\lambda$, and gives a subset $\cM \subsetneq [n]$ of corrupted (malicious) parties. Then $\cA_\lambda$ receives the inputs $x_i$ of all corrupted parties $\P{}_i$ ($i \in \cM$), sends and receive all the messages on the behalf of these corrupted parties, and communicates in an arbitrary quantum polynomial time way with the honest parties that follow the protocol $\Uppi$. At the end of the interaction, $\cA_\lambda$ outputs an arbitrary state $\rho$, and we define as $\vec{y}$ the output of the honest parties $\P{}_j$, $j \notin \cM$. Finally, we define $\Real_{\Uppi,\cA}(\lambda,\vec{x},\rho_\lambda)$ as the random variable corresponding to $(\rho, \vec{y})$.
  \end{definition}

  \begin{definition}[$\Ideal_{f,\Sim}(\lambda,\vec{x},\rho_\lambda)$]
    $\Sim$ (playing the role of the adversary) receives $\rho_\lambda$, outputs a set $\cM \subsetneq [n]$ of corrupted parties, interacts with a trusted party (called the ideal functionality) defined below, and outputs at the end a state $\rho$. The ideal functionality also outputs at the end a message $\vec{y}$ corresponding to the output of the trusted party. We then define $\Ideal_{f,\Sim}(\lambda,\vec{x},\rho_\lambda)$ as the random variable corresponding to $(\rho, \vec{y})$. Now we define the ideal functionality:
    \begin{itemize}
      \item The ideal functionality receives the set $\cM \subsetneq [n]$ a subset of corrupted parties, and for each party $\P{}_i$, it receives an input $x_i'$: if $\P{}_i$ is honest ($i \notin \cM$), we have $x_i'=x_i$, otherwise $x_i'$ can be arbitrary.
      \item Then, it computes $(y_1,\dots,y_n) \eqdef f(x_1',\dots,x_n')$, and sends $\{(i,y_i)_{i \in \cM}\}$ to the simulator.
      \item The simulator can choose to abort by sending a message $\bot$ to the ideal functionality. Otherwise it sends a ``continue'' message $\top$. If the message received by the ideal functionality is $\bot$, then it outputs $\bot$ to each honest party, which is formalized by outputting $\vec{y} \eqdef \{(i,\bot)\}_{i \notin \cM}$. Otherwise, it outputs $\vec{y} \eqdef \{(i, y_i)\}_{i \notin \cM}$.
    \end{itemize}
  \end{definition}

  \begin{definition}[Secure MPC~\cite{ABG+21_PostQuantumMultiPartyComputation}]\label{def:MPC}
    Let $f$ be a deterministic function with $n$ inputs and $n$ outputs, and $\Uppi$ be an $n$-party protocol. Protocol $\Uppi$ securely computes $f$ if for every non-uniform quantum polynomial-time adversary $\cA = \{\cA_\lambda\}_{\lambda \in \N}$ corrupting a set of at most $n-1$ players, there exists a non-uniform quantum polynomial-time ideal-world adversary $\Sim$ such that for any combination of inputs $\vec{x} \in (\{0,1\}^*)^n$ and any non-uniform quantum advice $\rho = \{\rho_\lambda\}_{\lambda \in \N}$,
    \begin{align}
      \{\Real_{\Uppi,\cA}(\lambda,\vec{x},\rho_\lambda)\}_{\lambda \in \N}
      \approx_c \{ \Ideal_{f,\Sim}(\lambda,\vec{x},\rho_\lambda) \}_{\lambda \in \N}
    \end{align}
  \end{definition}
}

\subsection{Introduction to the Learning With Errors (\LWE{}) problem}\label{subsec:LWE}
\pgfkeys{/prAtEnd/local custom defaults/.style={category=introLWE}}

\paragraph{The \LWE{} problem.}
The Learning With Errors (\LWE{}) problem was introduced in~\cite{Regev2005}. We briefly recall the definitions here, more details and useful lemmas can be found in \cref{appendix:introLWE}.

  \begin{definition}[Continuous and Discrete Gaussian]
    For any $s \in \R_{>0}$ and any vector $\vect{x} \in \R^N$, we define $\rho_s(\vect{x}) \eqdef \exp(-\pi \vect{x}^T \vect{x}/s^2)$. We define the continuous Gaussian distribution on $\R^N$ as $\cD_{s}^N(\vect{x}) \eqdef \rho_s(\vect{x})/s^N$ and the discrete Gaussian distribution on $\Z_q^N$ as $\cD_{\Z,s}^N(\vect{x}) \eqdef \rho_s(\vect{x})/\left(\sum_{\vect{y} \in \Z_q^N} \rho_s(\vect{y})\right)$. Note that sampling from $\cD_s^N$ is equivalent to sampling each coordinate from $\cD_s \eqdef \cD_s^1$.
  \end{definition}

  \begin{definition}[Learning With Errors problem (\LWE{})]
    Let $(N,M)$ be two integers, $\alpha \in (0,1)$, $q = q(N) \in \N_{\geq 2}$ be a modulus. For a distribution $\chi$ on $\R$ (respectively on $\Z_q$), the decision-$\LWE{}_{q,\chi}$ problem (matrix version) is to distinguish a sample of the uniform distribution $\cU(\Z_q^{M \times N} \times [0,q)^M)$ (respectively $\cU(\Z_q^{M \times N} \times \Z_q^M)$) from $(\vect{A}, \vect{A}\vect{s}+\vect{e} \bmod q)$, where $\vect{A} \in \Z_q^{M \times N}$ and $\vect{s} \in \Z_q^N$ are sampled uniformly at random, and $\vect{e}$ is sampled from $\chi^M$.
  \end{definition}

\inAppendix{
  \subsection{Definitions}
  We provide here a more in-depth introduction to \LWE{}.
  \begin{definition}[Learning With Errors (\LWE{})~\cite{Regev2005}]\label{def:DiscreteLWE}~\\
    Let $N \in \N$ \footnote{While usually the parameters $N$ is written in lowercase, we will use this notation here since $n$ already represents the number of applicants.}, $q = q(N) \in \N_{\geq 2}$ be a modulus and $\chi$ a distribution on $\R$ ($\chi$ may be continuous or $\chi \subseteq \Z$ (we will then say that $\chi$ is discrete), and will always be reduced modulo $q$). For any $s \in \Z_q^N$ we define $A_{\vect{s},\chi}$ as the distribution on $\Z_q^N \times [0,q)$ obtained by sampling $\vect{a} \in \Z_q^{N}$ uniformly at random, $e \sample \chi$ and outputting $(\vect{a}, \vect{a}^T \vect{s} + e \bmod q)$.

    We say that an algorithm solves the \emph{search-$\LWE{}_{q,\chi}$} problem (in the worst case, with overwhelming probability) if for any $\vect{s} \in \Z_q^N$, given an arbitrary number of samples from $A_{\vect{s},\chi}$, it outputs $\vect{s}$ with overwhelming probability. We say that an algorithm solves the \emph{decision-$\LWE{}_{q,\chi}$} problem (on the average, with non-negligible advantage) if it can distinguish with non-negligible advantage between the distribution $A_{\vect{s},\chi}$ where $\vect{s} \sample \Z_q^N$, and the uniform distribution $U \eqdef \cU(\Z_q^N \times [0,q))$ (when $\chi$ is discrete, we consider instead the uniform distribution $U \eqdef \cU(\Z_q^N \times \Z_q)$).

    We can also formulate this problem using matrices by grouping a fixed number $M \in \N$ of samples: For any $\vect{s} \in \Z_q^N$, we can sample $\vect{A} \sample \Z_q^{M \times N}$ and $\vect{e} \gets \chi^M$ a (typically small) vector where each of its component is sampled according to $\chi$. Then let $\vect{b} \eqdef \vect{A} \vect{s} + \vect{e}$. The \emph{search} problem consists in finding $\vect{s}$ given $(\vect{A},\vect{b})$. The \emph{decision} problem consists in deciding, when receiving a couple $(\vect{A},\vect{b})$, if $\vect{b}$ has been sampled uniformly at random (over $[0,q)^M$ if $\cX$ is continuous, or over $\Z_q^m$ if $\chi$ is discrete) or if $\vect{b}$ has been sampled according to the procedure described above (and therefore $\vect{b} = \vect{A} \vect{s} + \vect{e}$).
  \end{definition}

  Note that having access to less samples can only make the problem harder. On the other hand, one can show that having access to polynomially many samples is enough to generate arbitrary many further samples,
  with only a minor degradation in the error~\cite{GCV_2008_Trapdoors,ACPS_2009_Fast_crypto_primitives_circular_secure_enc,APS_2015_Concrete_hardness_LWE}. The worst-case and average-case search problems are in fact equivalent: \cite[Lem.~4.1]{Regev2005} shows how it is possible to turn a distinguisher that can solve the decision-\LWE{} problem with non-negligible advantage into a better distinguisher that can solve the decision-LWE{} problem with overwhelming probability.

  The distribution $\chi$ can be instantiated in many different ways: for example when $\chi$ is always equal to $0$, these problems are trivial, and when $\chi$ is uniform, they are impossible. In practice, $\chi$ is usually a (discrete or continuous) Gaussian~\arxivOnly{ (or more rarely a rounded Gaussian~\cite{Regev2005})}:

  \begin{definition}[Continuous and discrete Gaussian]
    For any $s \in \R_{>0}$ and any vector $\vect{x} \in \R^N$, we define $\rho_s(\vect{x}) \eqdef \exp\left(- \pi \left(\frac{\|\vect{x}\|_2}{s}\right)^2\right) = \exp(- \pi \vect{x}^T \vect{x}/s^2)$. By applying a linear transformation on $x$, we can generalize this notion: for any positive-definite matrix $\Sigma > \vect{0}$, we define:
    \begin{align}
      \rho_{\sqrt{\Sigma}}(\vect{x}) \eqdef \exp(- \pi \cdot \vect{x}^T \Sigma^{-1} \vect{x})
    \end{align}
    In particular, if $\Sigma = s^2 \vect{I}$, we have $\rho_s = \rho_{\sqrt{\Sigma}}$: therefore, $s$ will be used as a shortcut for $\sqrt{\Sigma} = \sqrt{s^2 \vect{I}}$. The normalization of the expression gives $\int_{\R^n} \rho_{\sqrt{\Sigma}}(\vect{x}) = \sqrt{\det \Sigma}$ and $\int_{\R^n} \rho_s(\vect{x}) = s^N$. We can now define the continuous Gaussian distribution:
    \begin{align}
      \cD_{s}^N(\vect{x}) \eqdef \rho_s(\vect{x})/s^N &\qquad \cD_{\sqrt{\Sigma}}(\vect{x}) \eqdef \rho_{\sqrt{\Sigma}}(\vect{x})/\sqrt{\det \Sigma}
    \end{align}
    Note that sampling from $\cD_s^N$ is equivalent to sampling each component from $\cD_s \eqdef \cD_s^1$. Moreover, due to our choice of normalization, $s$ and $\Sigma$ are not exactly equal to the standard deviation and to the covariance matrix: $\cD_s$ has standard deviation $\sigma \eqdef s/(\sqrt{2\pi})$ and the actual covariance of $\cD_{\sqrt{\Sigma}}$ is $\Sigma' \eqdef \Sigma/(2\pi)$.

    If $\Lambda \subseteq \R^N$ is a lattice (i.e.\ a discrete additive subgroup of $\R^N$), we define for any $\vect{c} \in \R^N$ the coset $\Lambda + \vect{c} = \{\vect{x} +\vect{c} \mid \vect{x} \in \Lambda\}$ and $\rho_{\sqrt{\Sigma}}(\Lambda+\vect{c}) \eqdef \sum_{\vect{x} \in \Lambda + \vect{c}} \rho_{\sqrt{\Sigma}}(\vect{x})$. We can define now the discrete Gaussian on $\Lambda + \vect{c}$ by simply normalizing $\rho_{\sqrt{\Sigma}}(\vect{x})$. For any $\vect{x} \in \R^N$, if $\vect{x} \notin \Lambda + \vect{c}$, $\cD_{\Lambda+\vect{c},\sqrt{\Sigma}}(\vect{x}) = 0$ and if $\vect{x} \in \Lambda + \vect{c}$:
    \begin{align}
      \cD_{\Lambda+\vect{c},\sqrt{\Sigma}}(\vect{x}) \eqdef \frac{\rho_{\sqrt{\Sigma}}(\vect{x})}{\rho_{\sqrt{\Sigma}}(\Lambda+\vect{c})}
    \end{align}
    In the following, for $\alpha \in (0,1)$, the discrete Gaussian on $\Z$, $\cD_{\Z, \alpha q} = \cD_{\Z, \sqrt{(\alpha q)^2 \vect{I}}}$, will be particularly important to sample the noise.
  \end{definition}

  The next lemma is useful to bound the length of a vector sampled according to a discrete Gaussian.
  \begin{lemmaE}[Particular case of {\cite[Lem.~1.5]{Ban_1993}\cite[Lem.~2.6]{MP11}}]\label{lem:boundGaussianDistrib}
    For any $s > 0$, we have:
    \begin{align}
      \pr{\|\vect{x}\|_2 \geq s \sqrt{n} \middle| \vect{x} \gets \disGaussCont{s}^n} \leq 2^{-n}
    \end{align}
  \end{lemmaE}
}

The decision-$\LWE_{q,\cD_{\alpha q}}$ problem is widely supposed to be hard to solve even for quantum computers and is the basic building block of many post-quantum cryptographic protocols~\cite{peikert2016decade}.%
\publishedVsArxiv{
  In particular, if $\alpha q > 2\sqrt{N}$, it is on average as hard as worst-case problems on lattices: notably it can be used~\cite{Regev2005,PRS17_PseudorandomnessRingLWEAny} to approximates the decision version of the shortest vector problem ($\GapSVP_\gamma$) to within $\gamma \eqdef \tilde{O}(N/\alpha)$.%
}{In particular, it is on average as hard as worst-case problems on lattices.}%
\inAppendixIfPublished{
  \begin{lemmaE}[Hardness of \LWE{}~\cite{PRS17_PseudorandomnessRingLWEAny}]\label{lem:GapSVPtoLWE}
    Let $N,q$ be integers and $\alpha \in (0,1)$ be such that $\alpha q > 2\sqrt{N}$. If there exists an algorithm that solves decision-$\LWE{}_{q,\cD_{\alpha q}}$, then there exists an efficient quantum algorithm that approximates the decision version of the shortest vector problem ($\GapSVP_\gamma$) \arxivOnly{and the shortest independent vectors problem ($\SIVP_\gamma$) }to within $\gamma \eqdef \tilde{O}(N/\alpha)$.
  \end{lemmaE}
}%
As there is no known algorithm to solve $\GapSVP_\gamma$ for $\gamma = \tilde{O}(2^{N^\eps})$ and $\eps \in (0,\frac{1}{2})$, we will assume (like~\cite{BGGHNSVV_2014_FHE_ABE}) that $\GapSVP_\gamma$ is secure for this kind of superpolynomial $\gamma$ (see \cref{appendix:bestAlgoLWE} for more details).
\inAppendix{
  \subsection{Best algorithms to solve \LWE{}}\label{appendix:bestAlgoLWE}
  To have strong security guarantees, we usually want $\gamma$ to be polynomial in $N$, but this is sometimes impossible. According to \cite{peikert2016decade} the best algorithm for solving these problems in polynomial time works for only slightly sub-exponential approximation factor $\gamma = 2^{\Theta(N \log \log N/\log N)}$. They also mention the algorithm~\cite{Schnorr_1987_Hierarchy} that provides a tradeoff between the approximation $\gamma$ and the running time: an approximation of $\gamma = 2^k$ can be obtained in time $2^{\tilde{\Theta}(N/k)}$. This suggest that there is no efficient algorithm for $\gamma = 2^{N^\epsilon}$ for $\epsilon \in (0,\frac{1}{2})$ (the algorithm provided by \cite{Schnorr_1987_Hierarchy} would indeed run in subexponential time $2^{N^{1-\epsilon}}$). In practice, many works rely on the security of \LWE{} when $\gamma$ is superpolynomial (\cite{BGGHNSVV_2014_FHE_ABE} uses for example the above assumption that $\gamma = 2^{N^{\epsilon}}$ for some $\epsilon \in (0,\frac{1}{2})$), and we will follow this same path here.
}

While \publishedVsArxiv{the above hardness assumption}{the hardness assumption given in \cref{lem:GapSVPtoLWE}} targets a continuous noise distribution, it can also be adapted to discrete Gaussians\arxivOnly{(note that the result is trivial for less used ``rounded Gaussians''~\cite[Lem.~4.3]{Regev2005})}.

In particular, \cite[Thm.~3.1]{Pei10_EfficientParallelGaussian} can be used to show that if decision-$\LWE_{q,\cD_{\alpha q}}$ is hard, then the discrete version decision-$\LWE_{q,\cD_{\Z, s}}$ is hard for $s \eqdef \sqrt{(\alpha q)^2 + \omega\left(\sqrt{\log \lambda}\right)^2}$ (details in \cref{cor:continuousToDiscrete}). Moreover, we can also choose to sample $\vect{s}$ according to $\cD_{\Z,s}$ without weakening the security: this is known as the Hermite Normal Form (details in \cref{lem:normalLWEreduction}).

\inAppendix{
  \begin{corollaryE}[{From continuous Gaussian to discrete Gaussian, corollary of \mbox{\cite[Thm.~3.1]{Pei10_EfficientParallelGaussian}}}][]\label{cor:continuousToDiscrete}
    Let $\lambda, q \in \N$, $\alpha \in (0,1)$. If $e_c$ is sampled according to $\cD_{\alpha q}$ and $e \gets e_c + \cD_{\Z - e_c, \omega\left(\sqrt{\log \lambda}\right)}$, then the marginal distribution of $e$ is within negligible statistical distance $\Delta = 1/(\exp(\pi \omega(\sqrt{\log \lambda})^2 -\ln(2\lambda))-1)=\negl[\lambda]$
    of $\cD_{\Z, s}$ with $s \eqdef \sqrt{(\alpha q)^2 + \omega\left(\sqrt{\log \lambda}\right)^2}$. Moreover, if there exists $x \in \Z_q$ (for example $x = \vect{a}^T \vect{s}$ for some $\vect{s} \in \Z_q^N$ and $\vect{a} \in \Z_q^N$) such that $e_c$ is distributed according to $x+\cD_{\alpha q} \bmod q$, then the statistical distance between the distribution $e \bmod q$ and $x + \cD_{\Z,s}\bmod q$ is $\Delta = \negl[\lambda]$. Finally, if $e_c$ is uniformly sampled over $[0,q)$, the marginal distribution of $e$ is uniform over $\Z_q$.
  \end{corollaryE}
  \begin{proofE}
    The first part of this corollary is a direct application of \cite[Thm.~3.1]{Pei10_EfficientParallelGaussian}: We define $\vect{c}_1 = \vect{0}$, $\Lambda_1 = \Z$, $\Sigma_2 = (\alpha q)^2 \vect{I}_1$, $\Sigma_1 = \omega(\sqrt{\log \lambda})^2 \vect{I}_1$ with $\sqrt{\Sigma_1} = \omega(\sqrt{\log \lambda}) \geq \eta_\eps(\Z)$, where $\eps$ is a negligible function of $N$ (the last inequality comes from \cite[Lem.~2.5]{Pei10_EfficientParallelGaussian}). When $\vect{x}_2(=e_c)$ is chosen according to a continuous Gaussian, the marginal distribution of $e$ is within statistical distance $8 \eps = \negl[\lambda]$ of $\cD_{\Z, \sqrt{\Sigma}}$ (the constant can actually be improved in this specific case), with $\Sigma = \Sigma_1 + \Sigma_2 = (\alpha q)^2 \vect{I}_1+ \omega(\sqrt{\log \lambda})^2 \vect{I}_1 = s^2\vect{I}_1$, which concludes the first part of the proof.

    To see that the equality also holds when $e_c \gets x + \cD_{\alpha q} \bmod q$ for some $x \in \Z$, we remark that for any $\bar{e} \in \Z_q$ and for any $e_c = x + e_{c}'\in \R$:
    \begin{align}
      p \eqdef& \pr{\bar{e} = e \bmod q \middle| e \gets e_c + \cD_{\Z - e_c, \omega\left(\sqrt{\log \lambda}\right)}}\\
      =& \pr{\bar{e} = e \middle| e \gets e_c + \cD_{\Z - e_c, \omega\left(\sqrt{\log \lambda}\right)} \bmod q}\\
      =& \pr{\bar{e} = e \middle| e \gets x+e_c' + \cD_{\Z - e_c', \omega\left(\sqrt{\log \lambda}\right)} \bmod q}
    \end{align}
    where the last equality comes from $\Z - e_c = \Z - (e_c' \bmod q)$. But we already know that $e_c' + \cD_{\Z - e_c', \omega\left(\sqrt{\log \lambda}\right)}$ is statistically close to $\cD_{\Z,s}$ from the first part of the corollary, which concludes this part of the proof.

    Now, let us assume that $e_c$ is sampled uniformly at random over $[0,q)$. Because $e_c$ can be uniquely decomposed into $e_c = e_{c,1} + e_{c,2}$ where $e_{c,1} \in \{0,\dots,q-1\}$ and $e_{c,2} \in [0,1)$, and because $\Z-e_c = \Z^n-e_{c,2}$ we have:
    \begin{align*}
      \begin{split}
        &\pr{\bar{e} = e \bmod q \middle| e_c \gets [0,q), e \gets e_c + \cD_{\Z-e_c, \omega\left(\sqrt{\log \lambda}\right)}}\\
        &\quad= \frac{1}{q} \int_0^1d e_{c,2} \sum_{e_{c,1} =0}^{q-1}
        \sum_{\substack{e \in \Z\\ \bar{e} = e \mathrlap{[q]}}} \frac{\rho_{\omega\left(\sqrt{\log \lambda}\right)}(e-e_{c,1}-e_{c,2})}{\rho_{\omega\left(\sqrt{\log \lambda}\right)}(\Z-e_{c,2})}
      \end{split}
    \end{align*}
    Similarly, because any integer $\hat{e} \in \Z$ can be decomposed uniquely into $\hat{e} = e - e_{c,1}$ where $e_{c,1} \in \{0,\dots,q-1\}$, and $e = \bar{e} \bmod q$, we can merge the two sums into a single sum over $\hat{e} \in \Z$, replace the $e-e_{c,1}$ with $\hat{e}$, and use the fact that $\sum_{\hat{e} \in \Z} \rho_{\omega}(\hat{e}-e_{c,2}) = \rho_\omega(\Z-e_{c,2})$ to conclude that the probability is equal to $1/q$: it corresponds to a uniform sampling over $\Z_q$.
  \end{proofE}

  While in the usual $\LWE{}$ problem, $\vect{s}$ is sampled uniformly at random over $\Z_q^N$, it is also possible to sample $\vect{s}$ according to a small Gaussian. Since this can be seen as a reformulation of the problem in its Hermite Normal form (HNF), this new sampling method is actually at least as secure as the initial uniform sampling, and it appears to be more efficient. Moreover, the construction given in \cite{MP11} will naturally be formulated in this form, so in this paper we will also sample $\vect{s}$ according to a small Gaussian (however, one can easily come back to the initial formulation as we will see later).

  \begin{lemmaE}[{Normal \LWE{} problem~\mbox{\cite[Lem.~2]{ACPS_2009_Fast_crypto_circular_enc}}}]\label{lem:normalLWEreduction}
    Let $q = p^k$ be a prime power. There is a deterministic polynomial-time transformation $T$ that, for arbitrary $\vect{s} \in \Z_q^N$ and error distribution $\chi$, maps $A_{\vect{s},\chi}$ to $A_{\bar{\vect{s}},\chi}$ where $\bar{\vect{s}} \gets \chi^N$, and maps $U(\Z_q^N \times \Z_q)$ to itself.
  \end{lemmaE}

  The idea of the proof given in \cite[Lem.~2]{ACPS_2009_Fast_crypto_circular_enc} (following \cite{MR_2009_lattice_based_crypto}, see also \cite[p.~23]{peikert2016decade}) is to first obtain and select enough samples $(\bar{\vect{A}}, \bar{\vect{y}} \eqdef \bar{\vect{A}} \vect{s} + \bar{\vect{s}})$, $\bar{\vect{s}}$ being sampled according to $\chi$, to ensure that $\bar{\vect{A}} \in \Z_q^{N \times N}$ is invertible. Then, any new sample $(\vect{a}, y \eqdef \langle \vect{a}, \vect{s}\rangle + e)$ can be updated into $(\vect{a}', b')$ where $\vect{a}' \eqdef -(\bar{\vect{A}}^T)^{-1}\vect{a}$ and $b' \eqdef y + \langle \vect{a}', \bar{\vect{y}} \rangle = \langle \vect{a}', \bar{\vect{s}} \rangle + e$.
}

\paragraph{The \cite{MP11} construction.}

In order to realize the primitives described in \cref{def:fct_requirements,def:GHZcanCapable}, we will use the trapdoor system presented in~\cite{MP11}. This work introduced an algorithm $\MPGen(1^\lambda)$ that samples a matrix $\vect{A}$ and a trapdoor $\vect{R}$. In addition, $\vect{A}$ is indistinguishable from a random matrix (without $\vect{R}$), and $g_{\vect{A}}(\vect{s}, \vect{e}) \eqdef \vect{A}\vect{s}+\vect{e}$ is injective and can be inverted given $\vect{R}$ for any $(\vect{s},\vect{e}) \in \ccX$ ($\ccX$ will be defined later as sets of elements having a small norm). The details of the construction and parameters required by the \cite{MP11} construction are provided in \cref{appendix:introLWE}.%

\inAppendix{%
  \subsection{The \cite{MP11} Construction}\label{sec:MP11construction}
  We give here more details on the construction introduced in~\cite{MP11}. Note that~\cite{MP11} uses a left-style multiplication $\vect{s}^T\vect{A}^T$, here we will prefer a right-style multiplication $\vect{A}\vect{s}$ for consistency. In this paper, we will also focus on the (more efficient) computationally-secure construction presented in \cite{MP11} (we also require the modulus $q \eqdef 2^k$ to be a power of $2$), but the same method should extend to other constructions with even $q$. We give in \cref{def:MP11construction} the construction we will use, and we explain after the intuition behind it.

  \begin{definition}[{\cite{MP11}}]\label{def:MP11construction}
    Let $\lambda \in \N$ be a security parameter, and $\cP_0 = (k, N, \alpha, r_{\max})$ with $(k,N) \in \N^2$, $\alpha \in (0,1)$ and $r_{\max} \in \R$, be some parameters that can depend on $\lambda$. We define $M \eqdef N(1+k)$, $q \eqdef 2^k$, $\cX_g \eqdef \left\{(\vect{s},\vect{e}) \in \Z_q^N \times \Z_q^M \middle| \left\|\SmallBlockMatrix{\vect{s}}{\vect{e}}\right\|_2 \leq r_{\max} \right\}$, $\vect{g} \eqdef \begin{bmatrix}
      1 & 2 & 4 & \dots & 2^{k-1}
    \end{bmatrix}^T
    \in \Z_q^{k}$, and the gadget matrix $\vect{G}$ as:
    \begin{align}
      \vect{G} \eqdef \vect{I}_n \otimes \vect{g} =
      \begin{bmatrix}
        \scalebox{.5}{\vdots} \\
        \vect{g}\\
        \scalebox{.5}{\vdots} \\
        & \ddots\\
        & &\scalebox{.5}{\vdots} \\
        & &\vect{g}\\
        & &\scalebox{.5}{\vdots} \\
      \end{bmatrix} \in \Z_q^{Nk \times N}
    \end{align}
    We define now in \cref{fig:constructionMP11} a procedure to sample a public matrix and its trapdoor $(\vect{A}, \vect{R}) \gets \MPGen_{\cP_0}(1^\lambda)$, and for any $(\vect{s}, \vect{e}) \in \cX_g$, we define $\vect{y} \eqdef g_{\vect{A}}(\vect{s}, \vect{e})$ and its inversion procedure $(\tilde{\vect{s}}, \tilde{\vect{e}}) \eqdef \MPDec(\vect{R}, \vect{A}, \vect{y})$.
    \begin{figure}[htb]
      \centering
      \begin{pchstack}[boxed, center,space=0.3cm]
        \begin{pcvstack}[space=0.3cm]
          {\normalfont\procedure[linenumbering]{$\MPGen_{\cP_0}(1^\lambda)$}{
              \hat{\vect{A}} \sample \Z_q^{N \times N}\\
              \vect{R} = \HorizBlockMatrix{\vect{R}_1}{\vect{R}_2} \gets \disGaussAQ^{Nk \times 2N}\\
              \vect{A} \eqdef \SmallBlockMatrix{\hat{\vect{A}}}{\vect{G} - \vect{R}_2\hat{\vect{A}} - \vect{R}_1} \in \Z_q^{M \times N}\\
              \pcreturn (\vect{A}, \vect{R})
            }}
          {\normalfont\procedure[linenumbering]{$g_{\vect{A}}(\vect{s}, \vect{e})$}{
              \pcreturn \vect{A}\vect{s} + \vect{e}
            }}
          {\normalfont\procedure[linenumbering]{$\MPDec_{\cP_0}(\vect{R} \eqdef \HorizBlockMatrix{\vect{R}_1}{\vect{R}_2}, \vect{A}, \vect{y})$}{
              \pclinecomment{Return $(\tilde{\vect{s}}, \tilde{\vect{e}}) \in \cX_g$ s.t. $\vect{y} = g_{\vect{A}}(\vect{s},\vect{e})$}\\
              \tilde{\vect{s}} \eqdef \InvertGadget_{\cP_0}(\HorizBlockMatrix{\vect{R}_2}{\vect{I}_{Nk}} \vect{y})\\
              \tilde{\vect{e}} \eqdef \vect{y} - \vect{A}\tilde{\vect{s}}\\
              \pcif \left\|\SmallBlockMatrix{\tilde{\vect{s}}}{\tilde{\vect{e}}}\right\|_2 > r_{\max}\  \pcreturn \bot \ \pcfi\\
              \pcreturn (\tilde{\vect{s}}, \tilde{\vect{e}})
            }}
        \end{pcvstack}
        \begin{pcvstack}[space=0.3cm]
          {\normalfont\procedure[linenumbering]{$\InvertSmallGadget_{\cP_0}(\vect{y} = \begin{bmatrix} y_0 & \dots y_{k-1}\end{bmatrix}^T)$}{
              \pclinecomment{Return $s \in \Z_q$ such that $\vect{y} = \vect{g}s + \vect{e}$}\\
              s \eqdef 0\\
              \pcfor i = k-1,\dots,0 \pcdo\\
              \t \pcif y_i - 2^i s \notin \left[-\frac{q}{4}, \frac{q}{4}\right) \bmod q\\
              \t \t s \eqdef s + 2^{k-1-i} \  \pcfi \pcendfor\\
              \pcreturn s
            }}
          {\normalfont\procedure[linenumbering]{$\InvertGadget_{\cP_0}(\vect{y} = \begin{bmatrix} \vect{y}_{1}^T & \dots \vect{y}_{N}^T\end{bmatrix}^T)$}{
              \pclinecomment{Return $\vect{s} \in \Z_q^{n}$ such that $\vect{y} = \vect{G}\vect{s}+\vect{e}$}\\
              \pcfor i = 1,\dots,N \pcdo\\
              \t s_i \eqdef \InvertSmallGadget_{\cP_0}(\vect{y}_i)\\
              \pcendfor\\
              \vect{s} \eqdef
              \begin{bmatrix}
                s_1 & \dots & s_n
              \end{bmatrix}^T\\
              \pcreturn \vect{s}
            }}
        \end{pcvstack}
      \end{pchstack}
      \caption{Construction from \cite{MP11}}
      \label{fig:constructionMP11}
    \end{figure}
  \end{definition}

  The idea of the construction given in \cref{def:MP11construction} is to use a gadget matrix $\vect{G}$ which is easy to invert even in the presence of noise, and then to hide this matrix inside a random looking matrix $\vect{A}$. $\vect{G}$ is easy to invert because $\vect{G}$ basically encodes all bits of the binary representation of each component of $\vect{s}$ in a different component (where a $1$ is encoded by $q/2+$noise, and $0$ is encoded by $0$+noise): the inversion of $\vect{G}$ is doable by a rounding operation, starting from the least significant bits of the components of $\vect{s}$. Then, as we will see, $\vect{R}$ can be used to invert $\vect{A}\vect{s} + \vect{e}$: with $\vect{R}$ we can obtain a vector of the form $\vect{G}\vect{s} + \vect{e}'$ ($\vect{e}'$ is small if $\vect{R}$ has sufficiently small singular values), and then since $\vect{G}$ is easy to invert we can obtain $\vect{s}$ easily. We formalize now these statements.

  \begin{lemmaE}[\cite{MP11}]\label{lem:MPGenIndistinguishableMatrix}
    If $\LWE{}_{q,\disGaussAQ}$ is hard and if $Nk = \poly[\lambda]$, then the matrix $\vect{A}$ obtained via $\MPGen$ is indistinguishable from a uniform random matrix.
  \end{lemmaE}
  \begin{proofE}
    For completeness, we sketch the proof given in \cite{MP11}. Since $\vect{G}$ is a fixed matrix, it is easy to subtract $\vect{G}$ from $\vect{A}$ and transpose the matrices: $\vect{A}$ looks random iff $(\hat{\vect{A}}^T, \hat{\vect{A}}^T\vect{R}_2^T+\vect{R}_1^T)$ looks random. But this is nearly an exact \LWE{} instance in its normal form ($\vect{R}_2$ is indeed sampled according to a small Gaussian). The only difference is that the $\hat{\vect{A}}$ samples are ``reused'' multiple times since $\vect{R}_i^T$ are matrices and not vectors. However, as shown in~\cite[Lem.~6.2]{PW_2011_lossy_trapdoor}, an hybrid argument can be made (by gradually replacing each column with random elements) to prove that it is still hard to distinguish it from a random matrix if $\LWE{}_{q,\disGaussAQ}$ is hard (since we obtain $Nk$ hybrid games, we need $Nk=\poly[\lambda]$).
  \end{proofE}

  \begin{remark}
    It is possible to easily translate the normal form into a more usual form in which $\vect{s}$ is sampled uniformly at random: one can sample a random invertible matrix $\vect{A}_r \in \Z_q^{N \times N}$, and define $\vect{A}' \eqdef \SmallBlockMatrix{\vect{I}}{\vect{A}}\vect{A}_r$.
  \end{remark}

  \begin{lemmaE}[][]\label{lem:MPInjectiveSingularValue}
    Let $\hat{\vect{A}} \in \Z_q^{N \times N}$, $(\vect{R}_1,\vect{R}_2) \in (\Z^{Nk \times N})^2$, $\vect{A} \eqdef \SmallBlockMatrix{\hat{\vect{A}}}{\vect{G} - \vect{R}_2\hat{\vect{A}} - \vect{R}_1} \in \Z_q^{M \times N}$. If the highest singular value $\sigma_{\max}(\vect{R})$ of $\vect{R}$ is such that $\sqrt{\sigma_{\max}(\vect{R})+1} < \frac{q}{4 r_{\max}}$, then $g_{\vect{A}}: \cX_g \rightarrow  \Z_q^M$ is injective and for all $(\vect{s},\vect{e}) \in \cX_g$, $\MPDec(\HorizBlockMatrix{\vect{R}_1}{\vect{R}_2}, \vect{A}, \vect{A}\vect{s} + \vect{e}) = (\vect{s}, \vect{e})$.

    Moreover, if we denote by $C \approx \frac{1}{\sqrt{2\pi}}$ the universal constant defined in \cite[Lem.~1.9]{MP11}, and if we have for the parameters $\cP_0$ defined in \cref{def:MP11construction}:
    \begin{align}
      \sqrt{\left(C \times \alpha q \times \sqrt{N}(\sqrt{k}+\sqrt{2}+1)\right)^2 + 1} \leq \frac{q}{4 r_{\max}}\label{eq:linkWithRmax}
    \end{align}
    then with overwhelming probability (on $N$) $\geq 1-2 e^{-N}$, we have $\sqrt{\sigma_{\max}(\vect{R})^2 +1 } < \frac{q}{4 r_{\max}}$ and therefore $g_{\vect{A}}$ is injective.
  \end{lemmaE}
  \begin{proofE}
    This proof can be obtained by combining theorems from \cite{MP11}. For completeness, we put the full proof of this lemma here. For the injectivity, we assume that there exist $(\vect{s},\vect{s}') \in S$ and $(\vect{e},\vect{e}') \in E$ such that $g_{\vect{A}}(\vect{s},\vect{e}) = g_{\vect{A}}(\vect{s}',\vect{e}')$. So $\vect{A}(\vect{s}-\vect{s}')+(\vect{e}-\vect{e}')=\vect{0}$, and it is enough now to prove that $\vect{s}=\vect{s}'$ to obtain $\vect{e}=\vect{e}'$. We multiply (on the left) this equation by $\HorizBlockMatrix{\vect{R}_2}{\vect{I}_{Nk}}$: we obtain $\vect{G}(\vect{s}-\vect{s}')-\vect{R}_1 (\vect{s}-\vect{s}')+\HorizBlockMatrix{\vect{R}_2}{\vect{I}_{Nk}}(\vect{e}-\vect{e}') = \vect{0}$, i.e.\ if we define $\tilde{\vect{R}} \eqdef \BlockMatrix{c|c|c}{-\vect{R}_1 & \vect{R}_2 & \vect{I}_{Nk} }$, $\tilde{\vect{e}} \eqdef \tilde{\vect{R}} \SmallBlockMatrix{\vect{s}-\vect{s}'}{\vect{e}-\vect{e}'}$ and $\tilde{\vect{s}} = \vect{s}-\vect{s}'$, we have $\vect{G}\tilde{\vect{s}}+\tilde{\vect{e}} = \vect{0}$. First, we remark that
    \begingroup
    \allowdisplaybreaks
    \begin{align}
      \sigma_{\max}(\tilde{\vect{R}^T})
      &= \|\tilde{\vect{R}^T}\|_2\nonumber\\
      &= \max_{x, \|x\|_2=1} \left\|\BlockMatrix{c}{-\vect{R}_1^T \tabularnewline \vect{R}_2^T \tabularnewline \vect{I}_{Nk} } x\right\|_2 \nonumber\\
      &= \max_{x, \|x\|_2=1} \sqrt{\left\|\BlockMatrix{c}{-\vect{R}_1^T \tabularnewline \vect{R}_2^T} x\right\|_2^2 + \|x\|_2^2}\nonumber\\
      &= \sqrt{\max_{x, \|x\|_2=1}\left\|\BlockMatrix{c}{\vect{R}_1^T \tabularnewline \vect{R}_2^T} x\right\|_2^2 + 1}\nonumber\\
      &= \sqrt{\max_{x, \|x\|_2=1}\left\|\vect{R}^T x\right\|_2^2 + 1}\nonumber\\
      &= \sqrt{\sigma_{\max}\left(\vect{R}^T\right)^2 + 1}\nonumber\\
      &= \sqrt{\sigma_{\max}\left(\vect{R}\right)^2 + 1}\label{eq:sigmaMaxR1R2}
    \end{align}
    \endgroup
    We prove now that for all $i$, $\tilde{\vect{e}}[i] \in (-\frac{q}{2}, \frac{q}{2})$: If we denote by $\vect{u}_i$ the vector such that $\vect{u}_i[i] = 1$ and for all $j \neq i$, $\vect{u}_i[j] = 0$, we have
    \begin{align}
      |\tilde{\vect{e}}[i]|
      &= |\vect{u}_i^T\tilde{\vect{e}}| = \left|\vect{u}_i^T \tilde{\vect{R}}\SmallBlockMatrix{\vect{s}-\vect{s}'}{\vect{e}-\vect{e}'}\right|
        = \left\langle \tilde{\vect{R}}^T \vect{u}_i, \SmallBlockMatrix{\vect{s}-\vect{s}'}{\vect{e}-\vect{e}'} \right\rangle\\
      \intertext{Using the Cauchy-Schwarz inequality we get:}
      |\tilde{\vect{e}}[i]|&\leq \|\tilde{\vect{R}}^T \vect{u}_i\|_2 \left\|\SmallBlockMatrix{\vect{s}-\vect{s}'}{\vect{e}-\vect{e}'}\right\|_2\\
      \intertext{Using $\|\vect{A}\|_2 = \sigma_{\max}(\tilde{\vect{R}}^T)$, and the definition of $\cX_g$ (\cref{def:MP11construction}), then \cref{eq:sigmaMaxR1R2} and the assumption on $\sigma_{\max}(\vect{R})$ we obtain:}
      |\tilde{\vect{e}}[i]|&\leq \sigma_{\max}(\tilde{\vect{R}^T}) \times (2 r_{\max}) < \frac{q}{2}
    \end{align}
    Now, if we write for all $i \in [N]$, $\tilde{\vect{s}}[i] = \sum_{j = 0}^{k-1} 2^j \tilde{\vect{s}}[i]_j$, with for all $j$, $\tilde{\vect{s}}[i]_j \in \{0,1\}$ (we also use the same notation for $\tilde{\vect{e}}$), we can prove, following the same path as the $\InvertSmallGadget$ algorithm that $\forall i,j$, $\tilde{\vect{s}}[i]_j = 0$. Let $i \in [n]$. If we consider only the line $l \eqdef (i-1)k+k$ of $\vect{G}\tilde{\vect{s}}+\tilde{\vect{e}}=\vect{0}$, we obtain $2^{k-1}\tilde{\vect{s}}+\tilde{\vect{e}}[l] = \frac{q}{2} \tilde{\vect{s}}[i]_0 + \tilde{\vect{e}}[i] \bmod q = \vect{0}$. Since $\tilde{\vect{e}}[l] \in (-\frac{q}{2}, \frac{q}{2})$, we can't have $\tilde{\vect{s}}[i]_0 = 1$, so $\tilde{\vect{s}}[i]_0 = 0$. We can then iterate the same process for $m = 1 \dots k-1$ with the line $l \eqdef (i-1)k+(k-m)$ to show that $\tilde{\vect{s}}[i]_{m} = 0$, i.e.\ $\tilde{\vect{s}}=0$, which concludes the proof of the injectivity of $g_{\vect{A}}$. Because this proof follows exactly the algorithm $\MPDec$, it is easy to see using the same argument that this algorithm correctly inverts any $\vect{y} = \vect{A}\vect{s}+\vect{e}$ with $(\vect{s},\vect{e}) \in \cX_g$.

    We prove now the second part of the theorem. Because $\vect{R}$ is sampled according to a discrete Gaussian of parameter $\alpha q$, so according to \cite[Lem.~2.8]{MP11}, this distribution is $0$-subgaussian, and therefore we can apply \cite[Lem.~2.9]{MP11} with, for example, $t=\sqrt{N/\pi}$ (we divide by $\pi$ just to simplify the probability, and to transform the $\leq$ into a $<$). So with probability $1-2\exp(-\pi t^2) = 1-2e^N$,
    \begin{align*}
      \sigma_{\max}(\vect{R}) \leq C \times \alpha q \times (\sqrt{Nk}+\sqrt{2N}+\sqrt{\frac{N}{\pi}}) < C \times \alpha q \times \sqrt{N}(\sqrt{k}+\sqrt{2}+1)
    \end{align*}
    To conclude the proof, we inject this equation inside \cref{eq:sigmaMaxR1R2} and use \cref{eq:linkWithRmax}.
  \end{proofE}
}

\section{Cryptographic requirements}\label{sec:crypto_requirements}
All the protocols are based on the existence of a (post-quantum secure) cryptographic family of functions, which is said to be \AssumpFct{}. Intuitively, in the definition below, the string $\vect{d}_0$ is such that $\vect{d}_0[i] = 1$ iff applicant $a_i$ is part of the support of the final \GHZ{}):
\begin{definition}\label{def:fct_requirements}
  Let $\lambda \in \N$ be a security parameter, and $n \in \N$. We say that a family of functions $\{f_k: \cX_\lambda \rightarrow \cY_\lambda \}_{k \in \cK_\lambda}$ with $\cX_\lambda \subseteq \{0,1\}^l$ is \AssumpFct{} if there exists a function $h\colon \cX_\lambda \rightarrow \{0,1\}^{n}$ ($h$ could be extended to depend on $k$) such that the following properties are respected:
  \begin{itemize}
  \item \textbf{efficient generation}: for all $\vect{d}_0 \in \{0,1\}^{n}$ a \PPT{} machine can efficiently sample $(k, t_k) \leftarrow \Gen(1^\lambda, \vect{d}_0)$ to generate (with overwhelming probability) an index $k \in \cK_\lambda$ and a trapdoor $t_k \in \cT_\lambda$.
  \item \textbf{efficient computation}: for any index $k$, the function $f_k$ is efficiently computable by a \PPT{} algorithm $\Eval(k,x)$.
  \item \textbf{trapdoor}: for any trapdoor $t_k$ and any $y$, there exists a procedure $\Dec$ that efficiently inverts $f_k$ when $y$ has two preimages. More precisely, if $y$ has exactly two distinct preimages, we have $\Dec(t_k, y) = f^{-1}(y)$. If the number of preimages is not 2, we expect $\Dec(t_k,y) = \bot$.
  \item \textbf{quantum input superposition}: there exists a \QPT{} algorithm that, on input $1^\lambda$ generates a uniform superposition $\sum_{x \in \cX_\lambda} \ket{x}$.\\
    For instance, if $\cX_\lambda = \{0,1\}^l$, it can be trivially done by computing $H^{\otimes l}\ket{0}^{\otimes l}$.
  \item \textbf{$\delta$-$2$-to-$1$}\footnote{This kind of functions are sometimes said to be $\delta$-$2$-regular.}: for all $k \in \cK$, when sampling an input $x$ uniformly at random in $\cX_\lambda$, the probability that $y \eqdef f_k(x)$ has exactly two distinct preimages (denoted by $x_y$ and $x_y'$ or simply $x$ and $x'$) is larger than $1-\delta$. When $\delta = 0$, we just say that the function is $2$-to-$1$.
  \item \textbf{\XOR{} of $h$}: for all $k$, there exists $\vect{d}_0 \in \{0,1\}^n$ such that for all $y$, if $y$ has exactly $2$ distinct preimages $x$ and $x'$ (i.e.\ $f_k^{-1}(y) = \{x, x'\}$ with $x \neq x'$), then:
          \[h(x) \xor h(x') = \vect{d}_0\]
          Moreover, if $k$ was obtained from $\Gen(1^\lambda,\vect{d}_0^*)$, then $\vect{d}_0 = \vect{d}_0^*$. We will always assume that $\vect{d}_0$ is easy to obtain from $t_k$ (it is always possible to append $d_0$ to $t_k$). Since, fixing $k$ fixes also $\vect{d}_0$, in the following we may use interchangeably $\vect{d}_0(k)$, $\vect{d}_0(t_k)$ or simply $\vect{d}_0$.
    \item
          \begin{minipage}[t]{.55\textwidth}
            \textbf{indistinguishability}:
            If $\Gen(1^\lambda, \cdot)$ is seen as en encryption function, then, similarly to \indcpa{} security\arxivOnly{ (see for example~\cite{Goldreich_2004_Foundation_cryptography})}, a quantum adversary cannot learn $\vect{d}_0$ from the index $k$, seen as an ``encryption'' of $\vect{d}_0$. More formally, if we formulate it as an indistinguishability game $\INDFCT_\Gen^\cA$, where $\cA = (\cA_1,\cA_2)$ is a non-uniform \QPT{} adversary (as explained in \cref{sec:notation}, $\cA_1$ gives implicitly its internal state to $\cA_2$), then any non-uniform \QPT{} adversary $\cA$ has only a \publishedVsArxiv{\makebox[0pt][l]{negligible advantage of winning this game:}}{negligible advantage of winning this game:} 
          \end{minipage}\hfill%
          \begin{minipage}[t]{.35\textwidth}
            \begin {pcimage}\label{game:indfct}
              {\normalfont \game[linenumbering]{$\INDFCT_{\Gen{}} ^\adv (\lambda) $}{%
                  (\vect{d}_0^{(0)},\vect{d}_0^{(1)}) \gets \cA_1(1^\lambda)\\
                  c \sample \bin\\
                  (k, t_k) \gets \Gen(1^\lambda, \vect{d}_0^{(c)})\\
                  \tilde{c} \gets \cA_2(k)\\
                  \pcreturn \tilde{c} = c
                }}
            \end{pcimage}%
          \end{minipage}
          \begin{align}
            \pr{\INDFCT_{\Gen{}}^{\adv}(\lambda)} \leq \frac{1}{2} + \negl[\lambda]
          \end{align}
  \end{itemize}
\end{definition}
Note that we provide in \cref{sec:functionConstruction} an explicit implementations of a \AssumpFct{} function.

The above properties are enough to prove the security of \blind{} and \blindSup{} against an arbitrary corruption of parties, and can be used to prove the security of \blindCan{} and \blindCanSup{} when the adversary corrupts only the server and the non-supported applicants. However, if we want to prove the security of these last two protocols in a stronger attack model, namely when the adversary can also corrupt some supported applicants, we also require our function to have a stronger property. Intuitively, the $\partialInfo{}$ function will list the messages to send to all applicants: if it contains a $\cross$ for applicant $i$, it means that applicant $i$ is not part of the support of the \GHZ{}, if it is a $0$ or a $1$, it means that the applicant gets a \GHZ{} canonical state---up to a local $X$ correction if it is a $1$---and, if it is a $\bot$, it means that the protocol aborts ``locally''.

\begin{remark}
  This local abort is interesting since it triggers when $y$ has only one preimage, and this means that the server is malicious with overwhelming probability\footnote{This is the case when $\delta = \negl[\lambda]$, which is possible to obtain using \LWE{} with superpolynomial noise ratio.}. Note that one may want to send this abort bit to all applicants, however it is not yet known if leaking this bit to other corrupted applicants could reduce the security of the protocol\publishedVsArxiv{.}{ (for example, it is not clear if a malicious server could force the protocol to abort when one specific honest applicant is not part of the \GGHZ{}). To avoid that issue, we introduce this notion of local abort, that tells locally to participants if the server was behaving honestly. Note that it is important to make sure that this abort bit don't leak to the adversary later, otherwise the security is not guaranteed anymore.}
\end{remark}

\begin{definition}\label{def:GHZcanCapable}
  A \AssumpFct{} family of function $\{f_k\}$ is said to be \AssumpFctCan{} if this additional property is respected:
  \begin{itemize}
    \item \textbf{indistinguishability with partial knowledge}: We want to show that, by leaking some information about the ``key'' of the \GHZ{} state owned by malicious applicants, we don't reveal additional information about the support status of applicants. More precisely, there exists a \PPT{} algorithm $\partialInfo\colon \cT_\lambda \times \cY_\lambda \rightarrow \{0,1,\cross,\bot\}^n$ with the following properties:
          \begin{itemize}
            \item
                  \textbf{correctness}: $\forall k \in \cK_\lambda, y \in \cY_\lambda$, and $v \gets \partialInfo(t_k, y)$:
                  \begin{itemize}
                    \item $y$ has exactly two preimages iff there is no partial abort (see discussion above): $|f_k^{-1}(y)| = 2$ iff $\bot \notin v$
                    \item for all $i$, if $v[i] \in \{0,1\}$ then $\vect{d}_0[i] = 1$, and if $v[i] = \cross$, then $\vect{d}_0[i] = 0$ (required to make sure $\cross$ is sent only to non-supported applicants and that $0$/$1$ is sent only to supported applicants).
                    \item if $|f_k^{-1}(y)| = 2$, $\exists u \in \{h(x),h(x')\}$, such that  $\forall i$, if $\vect{d}_0[i] = 1$ then $v[i] = u[i]$ (required to make sure that all corrections are correct).
                  \end{itemize}
          \end{itemize}
          \begin{minipage}[t]{0.64\linewidth}%
            \vspace*{-.7em}
            \lapbox[\textwidth]{0cm}{%
              \raisebox{-\height}{ 
                \begin{varwidth}{10cm}%
                  \begin{pcimage}\label{game:indpartialfct}
                    {\normalfont\game[linenumbering]{$\INDPARTIAL_{\Gen{},\partialInfo{}} ^\adv (\lambda) $}{%
                        (\cM, \vect{d}_0^{(0)},\vect{d}_0^{(1)}) \gets \cA_1(1^\lambda)\\
                        \pcif \exists i \in \cM, \vect{d}_0^{(0)}[i] \neq \vect{d}_0^{(1)}[i]: \pcreturn \pcfalse\ \pcfi
                        \label{line:Mcond}\\
                        c \sample \bin\\
                        (k, t_k) \gets \Gen(1^\lambda, \vect{d}_0^{(c)})\\
                        y \gets \cA_2(k)\\
                        v \gets \partialInfo(t_k, y)\\
                        \tilde{c} \gets \cA_3(\{(i,v[i])\}_{i \in \cM})\\
                        \pcreturn \tilde{c} = c
                      }}%
                  \end{pcimage}
                \end{varwidth}
              }
            }%
          \end{minipage}\begin{minipage}[t]{.36\linewidth}
            \begin{itemize}[wide, labelwidth=!, labelindent=0pt]
              \item
                    \textbf{security}:
                    The game on the left is impossible to win with non negligible advantage for any \QPT{} adversary $\cA = (\cA_1,\cA_2,\cA_3)$ (note that $\cM$ is intuitively the set of malicious corrupted applicants, and the condition line \ref{line:Mcond} is added because otherwise there is a trivial uninteresting way to distinguish).
            \end{itemize}
          \end{minipage}
  \end{itemize}
\end{definition}

For our protocol \authBlindCanDist{}, we also need to make sure that this family of functions can be computed in a distributed manner among users:
\begin{definition}\label{def:GHZdistCapable}
  A \AssumpFctCan{} family of function $\{f_k\}$ is said to be distributable if the above procedures can be computed after gathering partial results from the parties. More precisely:
  \begin{itemize}
    \item There exists $\GenLocal{}$, a ``local'' generation procedure such that sampling $(k,t_k) \leftarrow \Gen(1^\lambda,\vect{d}_0)$ can be done by first sampling for all $i$: $ (k^{(i)},t_k^{(i)}) \leftarrow \GenLocal(1^\lambda,\vect{d}_0[i])$ and defining $k \eqdef (k^{(1)},\dots,k^{(n)})$ and $t_k \eqdef (t_k^{(1)},\dots,t_k^{(n)})$. We will denote as $\cK_{\lambda,\texttt{Loc}}$ the set of such $k^{(i)}$, and we assume that $\cK = \cK_{\lambda,\texttt{Loc}}^{n}$.
    \item Similarly, there exists \partialInfoLocal{}, a ``local'' version of \partialInfo{} such that sampling $v \leftarrow \partialInfo{}(t_k,y)$ can be done by sampling for all $i$: $v[i] \leftarrow \partialInfoLocal{}(t_k^{(i)},y)$.
    \item Finally, there exists a method \partialAlpha{} such that for any bit string $b$ and for any $y$ such that $f_k^{-1}(y) = \{x,x'\}$ with $x \neq x'$ we have $\langle b, x \xor x' \rangle = \xor_i \partialAlpha(i, t_k^{(i)},y,b)$.
  \end{itemize}
  Moreover, we cannot assume anymore that people running these functions will be honest. Therefore, if we want to make sure that a non-supported malicious applicant cannot alter the state obtained by supported applicants (for example by providing a function which is not $\delta$-$2$-to-$1$), we also require the existence of a circuit $\CheckHonestTrapdoor_\lambda(\vect{d}_0[i], t_k^{(i)}, k^{(i)})$ that returns $\pctrue$ iff $k^{(i)} \in \cK_{\lambda,\texttt{Loc}}$\footnote{This is particularly important when $\cK_{\lambda,\texttt{Loc}} \subseteq \{0,1\}^*$ and when there exists no efficient algorithm to decide if a bit string $k^{(i)}$ is indeed an element of $\cK_{\lambda,\texttt{Loc}}$. For example, with our construction, it is easy to produce a key $k'$ that is indistinguishable from a key $k \in \cK$, and such that the function $f_{k'}$ is injective instead of $\delta$-$2$-to-$1$.}and if $k^{(i)}$ is the public key corresponding to the trapdoor $t_k^{(i)}$, embedding the bit $\vect{d}_0[i]$. This circuit can in particular be combined with a ZK protocol to prove in a Zero-Knowledge way that $k^{(i)} \in \cK_{\lambda,\texttt{Loc}}$.
\end{definition}

We also provide \cref{sec:compilerAssumpFctCan} a generic construction that turns a \AssumpFct{} family of functions into a $\delta'$-\AssumpFctCanNoDelta{} distributable family of functions, with $\delta' = 1-(1-\delta)^n\leq \delta n$. In particular, if $\delta$ is a negligible function of $\lambda$ as in \cref{sec:construction_f_superpoly} and $n = O(\lambda)$, $\delta'$ is negligible.

\section{Non-Interactive Zero-Knowledge Proofs on Quantum States}\label{sec:NIZKoQS}
\pgfkeys{/prAtEnd/local custom defaults/.style={category=proofsNIZKoQS}}

In this section we first define formally our new concept of Non-Interactive and Non-Destructive Zero-Knowledge proofs on Quantum States (NIZKoQS), and define a protocol achieving NIZKoQS. The more involved protocol \authBlindCanDist{} defined in the next section will also exploits NIZKoQS (but this protocol will have more than one message as it is also consuming the state produced by the NIZKoQS), while the other simpler protocols will only rely on the correctness (or completeness) of the following NIZKoQS protocol. Before giving the formal definition, let us motivate and describe informally the definition.

\paragraph{Quantum language.} In classical NIZK, a language $\lang$ is a set of strings, so similarly we will define a quantum language $\lang_\cQ$ as a set of quantum states. For instance, we could consider:
\begin{itemize}
\item the quantum language made of \BBHeightyFor{} states $\lang_\cQ^{\BBHeightyFor{}} = \{\ket{0}, \ket{1}, \ket{+}, \ket{-}\}$,
\item the quantum language $\lang_\cQ^{ex} = \{\ket{\vect{d}} \pm \ket{\vect{d} \xor \vect{d}_0} \mid (\vect{d},\vect{d}_0) \in (\{0,1\}^n)^2, \vect{d}_0[1] = 1, w_H(\vect{d}_0) = 2\}$ ($w_H$ denotes the Hamming weight), corresponding to all hidden \GHZ{} states whose first qubit is supported and where the support has size $2$ (i.e.\ only two qubits are entangled forming a Bell state, where one of these qubits is at the first position),
\item but we can also consider quantum languages referring to classical secrets, for instance if $p_k$ is a public key (say of a bitcoin wallet), and if $\Ver_{p_k}(s_k) = 1$ iff $s_k$ is the private key of $p_k$, we can define $\lang_\cQ^{p_k} = \{\ket{\vect{d}} \pm \ket{\vect{d} \xor \vect{d}_0} \mid (\vect{d},\vect{d}_0) \in (\{0,1\}^n)^2, \vect{d}_0[1] = 0 \lor (\vect{d}_0[1] = 1 \land \exists s_k, \Ver_{p_k}(s_k) = 1)\}$ that informally allows the prover to ``send'' a hidden \GHZ{} state where the first qubit is supported \emph{only} if the prover knows the private key of $p_k$.
\end{itemize}

Classically, both the prover and the verifier typically have a copy of the word $x$, and since information can be copied classically, the verification process cannot ``destroy'' $x$. Quantumly, this is not true anymore, therefore, instead of saying that all parties agree on $\rho$ before the protocol, what matters is that \emph{at the end} of the protocol, the verifier should end up with a $\rho \in \lang_\cQ$.

\paragraph{Relation, witness and quantum ZK.}

Classically, to check if a word $x$ belongs to a language $\lang$, we usually define a relation $\cR$ between a \emph{witness} $w$ and the word, saying that $x \in L \Leftrightarrow \exists w, w \cR x$. The prover typically knows the witness $w$ and the ZK property ensures that the verifier has no way to learn the witness $w$, formalizing the fact that the verifier learns nothing beyond the fact that the statement is true. Quantumly, we mimic this definition by defining a relation between classical \emph{witnesses} (or \emph{classes}\footnote{Unlike in classical ZK, the witness $\omega$ cannot be used to verify that a quantum state belongs to the quantum language, for this reason the term \emph{class} may be more appropriate. This also justifies the usage of a different notation $\omega$ instead of $w$.}) $\omega$ and quantum states $\cR$, saying that a quantum state $\rho$ belongs to a quantum language $\lang_\cQ$ if and only if there exists $\omega$ such that $\omega \cR \rho$. Similarly, we want to ensure that the verifier has no way to learn $\omega$.

However, even if our definition does not say anything more about witnesses, we need to choose them appropriately to obtain a meaningful and secure protocol. Moreover, at that stage it may not even be clear what could be used as a witness. For instance, in the quantum language $\lang_\cQ^{\BBHeightyFor{}}$ defined above, what would be the witness of $\ket{0}$? Because the ZK property ensures that no information leaks about the witness, while we typically want to ensure that no information is leaked about the received state, one could naively say that the witness is the classical description of the state. Unfortunately if each witness $\omega_\rho$ is in a $1$-to-$1$ correspondence to its corresponding state $\rho \in \lang_\cQ$, then it would be impossible to obtain the ZK property: for instance, given a random state in $\lang_\cQ^{\BBHeightyFor{}} = \{\ket{0}, \ket{1}, \ket{+}, \ket{-}\}$, it is possible to rule out one of the 4 states: for instance by measuring the state in the computational basis, if we measure $b$ we know that the state can't be $\ket{1-b}$ and therefore we know that the witness can't be $\omega_{\ket{1-b}}$, contradicting the ZK property.

To overcome this issue, a single witness $\omega$ must characterize a \emph{class} of states. For instance, for the language $\lang_\cQ^{\BBHeightyFor{}}$ we will define two witnesses $0$ and $1$ characterizing the basis of the state and we therefore define the relation $0 \cR \ket{0}$, $0 \cR \ket{1}$, $1 \cR \ket{+}$ and $1 \cR \ket{-}$. These classes will be used to characterize two of the three wanted properties:
\begin{itemize}
\item \textbf{Completeness} (or correctness): An honest prover should be able to choose $\omega$ and generate on the side of the verifier a state in $\lang_\omega \eqdef \{\rho \mid \omega \cR \rho\}$.
\item \textbf{Zero-Knowledge}: A malicious verifier should be unable to learn the witness $\omega$ with significant advantage over a random guess. Because of the completeness property, the verifier should therefore be unable to learn the class $\lang_\omega$ chosen by the prover that was supposed to contain the target state.
\item \textbf{Soundness}: Finally we also expect that if the prover is malicious, then an honest verifier will obtain a state in $\lang$ whenever it accepts. (Note that this property does not depend on $\cR$)
\end{itemize}

From the above properties, it clearly appears that when the relation $\cR$ is thinner (i.e.\ when $|\lang_\omega|$'s are smaller and the number of witnesses increases), we get a stronger result: indeed, an honest prover can choose more precisely the sent state and a malicious verifier is more confused as there are more classes to which a state could belong to. In particular, it is always possible to define a trivial NIZKoQS protocol for any language $\lang_\cQ$ if there is a single witness $\omega_0$ such that $\forall \rho \in \lang_\cQ, \omega_0 \cR \rho$: the prover would not do anything and the verifier would simply generating an arbitrary state in $\lang_\cQ$. However, the guarantees are quite poor in that case as the verifier can fully describe the state\dots{} For this reason, we will focus on non trivial relations, and we will always specify the relation associated to a quantum language.

Note that in some cases, it may be cumbersome to write separately the language and the relation, especially when the witness is an arbitrary label and when only the equivalence class formed by the relation matters. In that case, we may abuse notations and write directly $\lang = \{ \lang_\omega \}_\omega$, like $\lang_\cQ^{\BBHeightyFor{}} = \{\{\ket{0},\ket{1}\},\{\ket{+},\ket{-}\}\}$. This more succinct notation makes it clearer that the prover can choose in advance the basis (computation or Hadamard) of the state obtained by the verifier, that an honest verifier would always obtain a \BBHeightyFor{} state, and that a malicious verifier would be unable to learn the basis of the output state. For the other above examples of languages $\lang_\cQ^{ex}$ and $\lang_\cQ^{p_k}$, the witness that we will consider is the support $\vect{d}_0$ of the \GHZ{} state (therefore no malicious verifier will be able to learn any information about $\vect{d}_0$ beyond the fact that it respects the constraints that are specified in the language). We can now provide the formal definition of (NI)ZKoQS.

\begin{definition}[{Zero-Knowledge Proof on Quantum State (ZKoQS)}]\label{def:NIZKoQS}
  Let $(\P,\V)$ be a protocol between an honest \QPT{} prover\footnote{In our case the prover is actually \PPT{}.} $\P{}$ and an honest \QPT{} verifier $\V{}$ (that also outputs a final quantum state). Let $E = \cup_{n \in \N} \linearOp_\circ(\cH_n)$ be the set of finite dimensional quantum states (where $\cH_n$ is the Hilbert space of dimension $n$), $\cR \subseteq \{0,1\}^* \times E $ be a relation between bit strings (called \emph{witnesses} or \emph{classes}) and quantum states, and $\lang_\cQ = \{\rho \in E \mid \exists \omega, \omega \cR \rho\}$ be a \emph{quantum language} defined by $\cR$. Then $(\P_\lambda{},\V_\lambda{})$ is said to be a \emph{Zero-Knowledge proof on Quantum State} (ZKoQS) for $\lang_\cQ$ if the following properties are respected:
  \begin{enumerate}
    \item \textbf{Completeness}: There exists a negligible function $\mu(\cdotwild)$ such that for any $\lambda \in \N$ and $\omega$ such that $\exists \rho' \in \lang, \omega \cR \rho'$,
          \begin{align}
            \pr{a = 1 \text{ and } \omega \cR \rho \mid (a,\rho) \leftarrow \OUT_\V{}\langle \P_\lambda(\omega), \V_\lambda\rangle} = 1-\mu(\lambda)
          \end{align}
    \item \textbf{Soundness}: For any non-uniform \QPT{} malicious prover $\P{}^* = \{(\P{}^*_\lambda, \sigma_\lambda)\}_{\lambda \in \N}$, there exists a negligible function $\mu(\cdot)$ such that for any security parameter $\lambda \in \N$,
          \begin{align}
            \pr{ a = 1 \text{ and } \rho \notin \lang \mid (a,\rho) \leftarrow \OUT_\V\langle \P{}^*_\lambda(\sigma_\lambda), \V_\lambda\rangle}\leq \mu(\lambda)
          \end{align}
          When $\P^*$ is unbounded, it is called a \emph{proof} otherwise an \emph{argument}.
    \item \textbf{Quantum Zero Knowledge}: There exists a \QPT{} simulator $\Sim_\lambda$ such that for any \QPT{} verifier $\V{}^* = \{(\V{}^*_\lambda,\sigma_\lambda)\}_{\lambda \in \N}$,
          \begin{align}
            \{\OUT_{\V{}^*_\lambda}\langle \P_\lambda{}(\omega), \V{}^*_\lambda(\sigma_\lambda)\rangle\}_{\lambda,\omega} \approx_c \{\Sim_\lambda(\V{}^*_\lambda,\sigma_\lambda)\}_{\lambda,\omega}
          \end{align}
          where $\lambda \in \N$, $\omega \in \{\omega \mid \exists \rho, \omega \cR \rho\}$, and $\V{}^*$ is given to $\Sim_\lambda$ by sending the circuit description of $\V{}^*$.
  \end{enumerate}
  A Non-Interactive ZKoQS protocol---in which a single message is sent, from the prover to the verifier---will be denoted NIZKoQS.
\end{definition}

Now, we define a protocol where we can prove any property on the set of entangled qubits of a hidden \GHZ{} state in a NIZKoQS fashion. Note that we formalise the notion of ``any property'' by relying on a function $\Auth\colon \{0,1\}^n \times \cW_c \rightarrow \{0,1\}$ specifying if a given support $\vect{d}_0 \in \{0,1\}^n$ (which corresponds to the witness/class of the NIZKoQS) is allowed or not, where the decision can optionally depend on another classical secret witness $w \in \cW_c$. This additional optional set of witnesses $\cW_c$ is needed to allow NIZKoQS relying on classical secrets like in the $\cL_\cQ^{p_k}$ example provided above, where the $\cW_c$ would correspond to the set of private keys.

\begin{remark}\label{rq:whyHonestTrapdoor}
  Note that we also require here the existence of $\CheckHonestTrapdoor(\vect{d}_0,t_k,k)$ that will check if $k \in \cK$ and if $\vect{d}_0$ corresponds to the constant involved in the \XOR{} property of $k$. The reason is that we cannot anymore be sure that Alice (the prover) is honest: therefore we need to check that the $k$ sent by Alice was honestly prepared. This is particularly important when $\cK_{\lambda} \subseteq \{0,1\}^*$ and when there exists no efficient algorithm to decide if a bit string $k$ is indeed an element of $\cK_{\lambda}$. For example, with our construction, it is easy to produce a key $k'$ that is indistinguishable from a key $k \in \cK$, and such that the function $f_{k'}$ is injective instead of $\delta$-$2$-to-$1$. Our construction based on \LWE{} does guarantee that there exists a function $\CheckHonestTrapdoor$, which internally checks if the singular values of the trapdoor $\vect{R}$ are small enough and if the norm of $(\vect{s}_0,\vect{e}_0)$ is small enough (see \cref{lem:conditionsFkDelta2To1,lem:upperBoundDelta} for more details).
\end{remark}

\begin{protocol}
  \caption{\blindZK{}}\label{protocol:blindZK}
  \textbf{Assumptions}: There exists a \AssumpFctNegl{} family of functions (\cref{def:fct_requirements}), together with an efficient function $\CheckHonestTrapdoor_\lambda(\vect{d}_0, t_k, k)$ outputting \pctrue{} if $k \in \cK$ and if $\vect{d}_0 = \vect{d}_0(t_k)$.\\
  \textbf{Parties}: A classical sender/prover (Alice) and a quantum receiver/verifier (Bob).\\
  \textbf{Common inputs}: The size $n$ of the hidden \GHZ{} and an efficiently computable function $\Auth\colon \{0,1\}^n \times \cW_c \rightarrow \{0,1\}$ where $\{0,1\}^n$ corresponds to the witness defined in \cref{def:NIZKoQS} and $\cW_c$ is a set of other optional witnesses.\\
  \textbf{Alice's input}: The support $\vect{d}_0 \in \{0,1\}^n$ of the hidden \GHZ{} state and a witness $w_c \in \cW_c$ such that $\Auth(\vect{d}_0,w_c) = 1$.\\
  \textbf{Bob's output}: Bob can reject or accept and output a quantum state if he thinks that there exist $\vect{d}_0$ and $w$ such that $\Auth(\vect{d}_0,w_c) = 1$ and such that the quantum state is a hidden \GHZ{} state of support $\vect{d}_0$.\\
  \textbf{Protocol}:\\
  \begin{compressedList}
  \item Alice: Generate $(k,t_k) \gets \Gen(1^\lambda, \vect{d}_0)$ and a NIZK proof $\pi$ proving that $\CheckHonestTrapdoor(\vect{d}_0,t_k,k) \land \Auth(\vect{d}_0,w_c)=1$ (the witness being $(\vect{d}_0,t_k,w_c)$ and the word $k$) and send $(k,\pi)$ to Bob.
  \item Bob: Check that $\pi$ is correct (if not, reject), then perform the following operations (circuit in \publishedOnly{\cref{cref:appendixBlind}, }\cref{fig:circuit_bob}):
    \begin{compressedList}
    \item create the state $\sum_{x \in \cX_\lambda} \ket{x}\ket{h(x)}\ket{f_k(x)}$ by applying in superposition the unitary mapping $\ket{x}\ket{0}\ket{0} \mapsto \ket{x}\ket{h(x)}\ket{f_k(x)}$ on the uniform superposition of all inputs,
    \item measure the third register in the computational basis, which gives a $y$,
    \item measure the first register in the Hadamard basis, which gives a bit string $b$
    \end{compressedList}
    Then, output the resulting quantum state on the second register.
  \end{compressedList}
\end{protocol}

\inAppendixIfPublished{
  \begin{figure}[ht]
    \centering
    \includegraphics[]{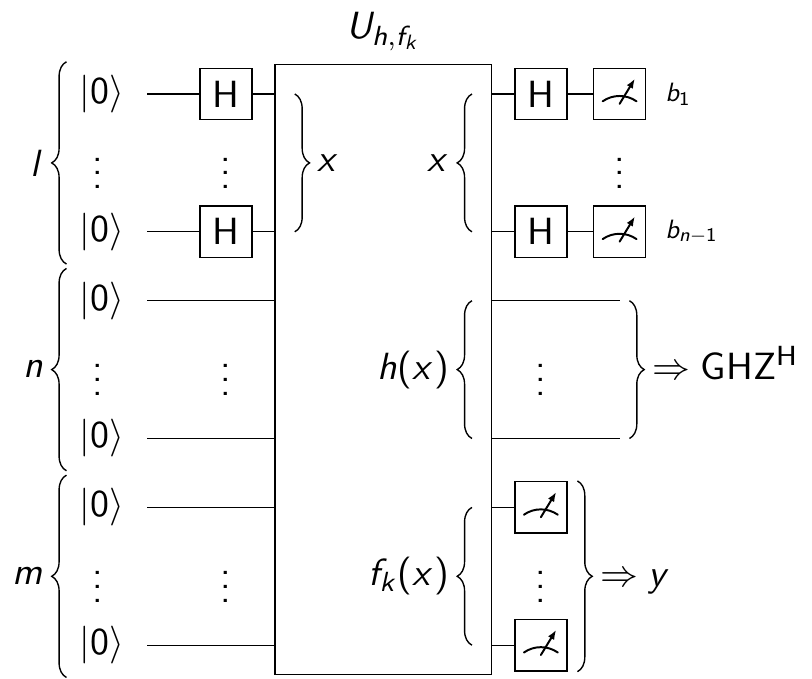}
    \caption{Circuit performed by the server Bob.}
    \label{fig:circuit_bob}
  \end{figure}
}

\begin{theorem}[NIZKoQS]\label{thm:NIZKoQS}
  Let $n \in \N$ be the size of the quantum state outputted by Bob in \blindZK{} and $\delta=\negl[\lambda]$. The protocol \blindZK{} (where Alice plays the role of the prover $\P$ and Bob the verifier $\V$) is a NIZKoQS for the quantum language, parameterized by the witnesses/classes $\vect{d}_0 \in \{0,1\}^n$, defined by all hidden \GHZ{} states $\rho$ on $n$ qubits whose support $\vect{d}_0$ is such that there exists $w_c$ such that $\Auth(\vect{d}_0, w_c)$ (defining the relation $\vect{d}_0 \cR \rho$). More precisely, when the verifier Bob accepts, with overwhelming probability the outputted by Bob is the hidden \GHZ{} state $\ket{\vect{d}}+(-1)^\alpha\ket{\vect{d}'}$, with
  \begin{align}
    \vect{d} = h(x) \qquad \vect{d}'= h(x') \qquad \alpha = \bigoplus_i b_i (x_i \xor x'_i) = \langle b, x \xor x'\rangle
  \end{align}
  In particular, since by definition of $f_k$ we have $\vect{d} \xor \vect{d}' = h(x) \xor h(x') = \vect{d}_0$, the support of the hidden \GHZ{} is $\vect{d}_0$ as explained in \cref{subsec:defGHZ}, where $\exists w_c, \Auth(\vect{d}_0,w_c)$.
\end{theorem}
\begin{proof}
  The protocol is non interactive since a single message $(k,\pi)$ is sent (the rest of the proof will also work for interactive protocols). First, for completeness we can see that when both parties are honest, the verifier accepts ($a=1$) with probability $1$ thanks to the perfect completeness of the classical NIZK protocol. Then, one can see that after the measurement of the third register, Bob gets on the first two registers the state $\sum_{x \in f^{-1}(y)} \ket{x}\ket{h(x)}$. But because the first register was created as a uniform superposition, $y$ has exactly two preimages $x,x'$ with probability $1-\delta$. So with probability $1-\delta$, the state obtained by an honest Bob is $(\ket{x}\ket{h(x)}+\ket{x'}\ket{h(x')})$. Then, after applying the Hadamard gates (preparing the measurement in the Hadamard basis), we obtain the state $\sum_b (-1)^{\langle b, x\rangle}\ket{b}\ket{h(x)} + (-1)^{\langle b, x'\rangle}\ket{b}\ket{h(x')}$. So after measuring a $b$, the second register contains the state:
  \begin{align}
    (-1)^{\langle b, x\rangle}\ket{h(x)}+(-1)^{\langle b, x'\rangle}\ket{h(x')} &= (-1)^{\langle b, x\rangle}(\ket{h(x)}+(-1)^{\langle b, x \xor x'\rangle}\ket{h(x')})
  \end{align}
  which is equal to the state $\HGHZ{}_{\alpha,\vect{d},\vect{d}'} = \ket{\vect{d}} + (-1)^\alpha\ket{\vect{d}'}$, once we get rid of the global phases and define $\vect{d}, \vect{d}'$ and $\alpha$ as above. Since we have $\vect{d} \xor \vect{d}' = h(x) \xor h(x') = \vect{d}_0$, the support of this hidden GHZ state is $\vect{d}_0$ as explained in \cref{subsec:defGHZ}, which concludes the proof of completeness.

  The soundness property relies on the soundness property of the classical NIZK protocol, and again on the correctness of the circuit performed by the server: the probability of accepting a $k$ which is not in $\cK_\lambda$ or such that there exists no $w_c$ such that $\Auth(\vect{d}_0, w_c)$ is negligible; the correctness of the protocol described above is enough to conclude that the hidden \GHZ{} has the expected properties.

  For the Zero-Knowledge property, we define the following simulator, where $\Sim_{ZK}$ is the simulator of the classical (NI)ZK protocol, and $k \rightSend \V_\lambda^*$ is the machine obtained by running $\V_\lambda^*$, and sending $k$ as first message:
  \begin{namedGame}[SimLambdaNIZKoQS]{\Sim_\lambda(\V_\lambda^{*}, \sigma_\lambda)}*
    \vect{d}'_0 \sample \{0,1\}^n\\
    (k',t_{k'}) \leftarrow \Gen(1^\lambda, \vect{d}'_0)\\
    \pcreturn \Sim_{ZK}(k', k' \rightSend \V_\lambda^*, \sigma_\lambda)
  \end{namedGame}
  To prove that the output of \refGame*{SimLambdaNIZKoQS} is indistinguishable from the real world, we define an hybrid distribution:
  \begin{namedGame}[NIZKoQSGameOne][\Game{1}]{\Game{1}(\vect{d}_0, \V_\lambda^{*}, \sigma_\lambda)}*
    (k,t_{k}) \leftarrow \Gen(1^\lambda, \vect{d}_0)\\
    \pcreturn \Sim_{ZK}(k, k \rightSend \V_\lambda^*, \sigma_\lambda)
  \end{namedGame}
  First, one can see that $\{\refGame*{NIZKoQSGameOne}(\vect{d}_0, \V_\lambda^{*},\sigma_\lambda)\}_{\lambda,\vect{d}_0} \approx_c \{\Sim_\lambda(\V_\lambda^*, \sigma_\lambda)\}_{\lambda,\vect{d}_0}$. Indeed, if a non-uniform distinguisher $D$ can distinguish between these two distributions, then we can use $D$ to break the game $\INDFCT_{\Gen{}} ^\adv (\lambda)$ by simply sending for any $\lambda$ a random $\vect{d}_0$ and the $\vect{d}_0$ which maximizes the distinguishing probability (anyway, $D$ is already non-uniform). Then, $\{\refGame*{NIZKoQSGameOne}(\vect{d}_0, \V_\lambda^{*},\sigma_\lambda)\}_{\lambda,\vect{d}_0} \approx_c \{\OUT_{\V{}^*_\lambda}\langle \P_\lambda{}(\vect{d}_0), \V{}^*_\lambda(\sigma_\lambda)\rangle\}_{\lambda,\vect{d}_0}$ since \refGame*{NIZKoQSGameOne} is exactly the same as the RHS, except that we replaced the actual ZK protocol with its simulator, which is an indistinguishable process by definition of ZK.
\end{proof}

\begin{remark}
  In particular, note that the above protocol provides a NIZKoQS for the examples given at the beginning of the section. The quantum language $\lang_\cQ^{\BBHeightyFor{}}$ can be certified by remarking that a \BBHeightyFor{} state is a hidden \GHZ{} state of size one, where $\vect{d}_0 \in \{0,1\}$ represents the basis, and therefore we just need to take $\Auth(\vect{d}_0,\_) = 1$. The language $\lang_\cQ^{ex}$ can be certified using $\Auth(\vect{d}_0,\_) = (\vect{d}_0[1] = 1 \land w_H(\vect{d}_0) = 2)$ and finally $\lang_\cQ^{p_k}$ can be obtained using $\Auth(\vect{d}_0,s_k) = (\Ver_{p_k}(s_k) = 1 \implies \vect{d}_0[1]=1)$.
\end{remark}

\section{The different protocols}\label{sec:protocols}
In this section, we define the protocols \blind{} (\cref{protocol:blind}), \blindSup{} (\cref{protocol:blindSup}), \arxivOnly{the impossible \blindCan{} (\cref{subsec:blindCan}), }\blindCanSup{} (\cref{protocol:blindCanSup}) and finally our main protocol \authBlindCanDist{} (\cref{protocol:authBlindCanDist}). \publishedOnly{We also prove in \cref{subsec:blindCan} the impossibility of \blindCan{}.}

\subsection{The protocol \blind{}}\label{subsec:blind}
\pgfkeys{/prAtEnd/local custom defaults/.style={category=proofsBlind}}

We define now the protocols \blind{} (\cref{protocol:blind}), which is the basic building block of all the other protocols. Note that this protocol can already easily be extended to exhibit NIZKoQS as shown in \cref{sec:NIZKoQS}. This protocol allows a trusted classical party, Cupid, to create a hidden \GHZ{} state on a server. This hidden \GHZ{} state is then distributed to all applicants (one qubit per applicant), who just need to store the received qubit.

\begin{protocol}[htbp]
  \caption{\blind{}}\label{protocol:blind}
  \textbf{Inputs}: Cupid gets as input $\vect{d}_0 \in \{0,1\}^n$, a bit string describing the final supported applicants: applicant $a_i$ will be supported iff $\vect{d}_0[i] = 1$. $\lambda \in \N$ is a public, fixed, security parameter.\\
  \textbf{Assumptions}: There exists a \AssumpFct{} family of functions with $\delta = \negl[\lambda]$.\\
  \textbf{Protocol}:%
  \begin{compressedList}
    \item Cupid: Generate $(k, t_k) \leftarrow \Gen(1^\lambda, \vect{d}_0)$, and send $k$ to the server.
    \item Server: Perform the quantum operations described in \blindZK{} (circuit in \publishedOnly{\cref{cref:appendixBlind}, }\cref{fig:circuit_bob}). Then, send $(y,b)$ to Cupid, and for all $i$, send the $i$-th qubit of the (unmeasured) second register to applicant $a_i$.
    \item All applicants: just receive and store the qubit sent by the server.
  \end{compressedList}
\end{protocol}


\begin{lemmaE}[Correctness of \blind{} and \blindSup{}][end if published]\label{lem:correctnessBlind}
  At the end of an honest run of protocol \blind{}, when $y$ has exactly two distinct preimages $x,x'$ (which occurs with probability $1-\delta$ according to \cref{def:fct_requirements}, which is overwhelming when we use the construction defined in \cref{sec:construction_f_superpoly}), the state shared between all applicants is a hidden generalized $\HGHZ_{\alpha,\vect{d},\vect{d}'}$ state, with:
  \begin{align}
    \vect{d} = h(x) \qquad \vect{d}'= h(x') \qquad \alpha = \bigoplus_i b_i (x_i \xor x'_i) = \langle b, x \xor x'\rangle
  \end{align}
  In particular, since by definition of $f_k$ we have $\vect{d}_0 = h(x) \xor h(x') = \vect{d} \xor \vect{d}'$, the support of the hidden \GHZ{} is $\vect{d}_0$.
\end{lemmaE}
The proof of this lemma is a direct consequence of the completeness of the protocol \blindZK{} proven in the first part of the proof of \cref{thm:NIZKoQS}.

\begin{lemmaE}[Security of \blind{}][end if published]
  At the end of a fully malicious interaction during the protocol \blind{}, where all applicants and the server can be fully malicious and can all collude together, the set of supported applicants is completely hidden. More precisely, if we define a game following the spirit of \indcpa{} security, no \QPT{} adversary $\cA = (\cA_1,\cA_2)$ can win the game $\INDBLIND{}$ with probability better than $\frac{1}{2} + \negl[\lambda]$.
\end{lemmaE}

\begin{namedGame}[INDBLIND][\INDBLINDtxt{}]{\INDBLINDtxt{}_{\Gen}^\adv (\lambda)}
  (\vect{d}_0^{(0)},\vect{d}_0^{(1)}) \gets \cA_1(1^\lambda)\\
  c \sample \bin\\
  (k, t_k) \gets \Gen(1^\lambda, \vect{\vect{d}}_0^{(c)})\\
  (y,b,\tilde{c}) \gets \cA_2(k)\\
  \pclinecomment{\text{No more interaction}}\\
  \pcreturn \tilde{c} = c  
\end{namedGame}
\begin{proofE}
  This is a trivial reduction to the indistinguishability property of the family $\{f_k\}$: since $y$ and $b$ are not used, we can remove them without changing the probability of winning, and we get exactly the game \INDFCT{} introduced \cref{game:indfct}. This game is impossible to win for probability better than $\frac{1}{2} + \negl[\lambda]$ by assumption on the family $\{f_k\}$, which ends the proof.
\end{proofE}

\subsection{The protocol \blindSup{}}\label{subsec:blindSup}
\pgfkeys{/prAtEnd/local custom defaults/.style={category=proofsBlindSup}}

We describe now the protocol \blindSup{} (\cref{protocol:blindSup}). In this protocol, all the applicants will obtain a qubit part of a hidden \GHZ{} state, and they will learn their own support status. However, they will not know the ``key'' of the hidden \GHZ{} state.

\begin{protocol}[htbp]
  \caption{\blindSup{}}\label{protocol:blindSup}
  \textbf{Inputs}: Same as \blind{}: Cupid gets $\vect{d}_0$ and $\lambda$.\\
  \textbf{Assumptions}: Same as \blind{} (\AssumpFct{} family with $\delta = \negl[\lambda]$).\\
  \textbf{Protocol}:%
  \begin{compressedList}
    \item Run the protocol \blind{}, so that Cupid gets $(b,y)$ and each applicant $a_i$ the $i$-th qubit
    \item Cupid: For all $i$, send $\vect{d}_0[i]$ to applicant $a_i$, so that each applicant knows whether they are supported or not.
  \end{compressedList}
\end{protocol}

Now, in order to prove the security of the \blindSup{} protocol, we first need to define what we mean by security. Since in this protocol Cupid reveal to all applicants their respective support status, we can't use the previous definition of security.

\begin{minipage}[t]{.5\linewidth}
  \begin{lemmaE}[Security of \blindSup{}][end if published]\label{lem:securityBlindSup}
    If we allow in the protocol \blindSup{} the fully malicious server Bob to corrupt an arbitrary subset of applicants, then the support status of the remaining honest applicants is completely hidden. More precisely, no \QPT{} adversary $\cA = (\cA_1,\cA_2,\cA_3)$ can win the game $\INDBLINDSUP$ with probability better than $\frac{1}{2} + \negl[\lambda]$. In the following, $\cM$ is the set of malicious applicants corrupted by Bob, and the condition $\forall i \in \cM, \vect{d}_0^{(0)}[i] = \vect{d}_0^{(1)}[i]$ is required to avoid a trivial uninteresting distinguishing strategy.
  \end{lemmaE}
\end{minipage}%
\begin{minipage}[t]{.5\linewidth}
  \begin{namedGame}*[INDBLINDSUP][\INDBLINDSUPtxt{}]{\INDBLINDSUPtxt{}_{\{f_k\}}^\adv (\lambda)}
    (\cM, \vect{d}_0^{(0)},\vect{d}_0^{(1)}) \gets \cA_1(1^\lambda)\\
    \pcif \exists i \in \cM, \vect{d}_0^{(0)}[i] \neq \vect{d}_0^{(1)}[i]:\\
    \t \pcreturn \pcfalse\ \pcfi\\
    c \sample \bin\\
    (k, t_k) \gets \Gen(1^\lambda, \vect{d}_0^{(c)})\\
    (y,b) \gets \cA_2(k)\\
    \pclinecomment{\text{The adversary has only access}}\\
    \pclinecomment{\text{to the messages sent by Cupid}}\\
    \pclinecomment{\text{to corrupted applicants:}}\\
    \tilde{c} \gets \cA_3(\{(i, \vect{d}_0^{(c)}[i])\}_{i \in \cM})\\
    \pcreturn \tilde{c} = c    
  \end{namedGame}
\end{minipage}

\begin{proofE}
  To prove the security of this scheme, we will assume by contradiction that there exists an adversary $\cA = (\cA_1, \cA_2, \cA_3)$ that can win the game $\INDBLINDSUP{}$ with probability $p_{\cA} \eqdef \frac{1}{2} + \frac{1}{\poly[\lambda]}$, and we will construct an adversary $\cA' = (\cA'_1, \cA'_2)$ that can win the game \INDBLIND{} with a non negligible advantage (which is impossible by assumption). So we define $\cA'_1(\lambda)$ as follows: $\cA'_1$ runs in a blackbox way $(\cM, \vect{d}_0^{(0)},\vect{d}_0^{(1)}) \gets \cA_1$, returns $(\vect{d}_0^{(0)},\vect{d}_0^{(1)})$ and keeps $(\cM, \vect{d}_0^{(0)})$ in its internal state. We then define:
  \begin{align}
    \cA'_2(k, \state_1 \eqdef (\cM, \vect{d}_0^{(0)})) \eqdef (y,b) \leftarrow \cA_2(k); \tilde{c} \gets \cA_3(\{(i, \vect{d}_0^{(0)}[i])\}_{i \in \cM}); \pcreturn \tilde{c}
  \end{align}
  It is then easy to see that $\cA'$ wins the game $\INDFCT$ with probability greater than $p_{\cA}$. Indeed, when $\cA_1$ outputs a $(\cM, \vect{d}_0^{(0)},\vect{d}_0^{(1)})$ that does not respect the condition $\forall i \in \cM, \vect{d}_0^{(0)}[i] = \vect{d}_0^{(1)}[i]$, then $\cA$ always lose (while $\cA'$ may win the game). Moreover, when the condition is respected, since $\{(i, \vect{d}_0^{(0)}[i])\}_{i \in \cM} = \{(i, \vect{d}_0^{(1)}[i])\}_{i \in \cM} = \{(i, \vect{d}_0^{(c)}[i])\}_{i \in \cM}$, we can replace the input of $\cA'_3$ with $\{(i, \vect{d}_0^{(c)}[i])\}_{i \in \cM}$: the game is now exactly equivalent to \INDBLINDSUP{}, so in that case $\cA'$ win with the exact same probability as $\cA$. Therefore, $\cA'$ wins the game \INDFCT{} with probability greater than $p_\cA = \frac{1}{2} + \frac{1}{\poly[n]}$: contradiction.
\end{proofE}

\arxivOnly{\subsection{The impossible protocol \blindCan{}}\label{subsec:blindCan}}
\pgfkeys{/prAtEnd/local custom defaults/.style={category=proofsBlindCan}}

\inAppendixIfPublished{
  One may be interested by a protocol \blindCan{}, that would make sure that all supported applicants share a canonical \GHZ{}, but that at the same time none of them know if they are part of the \GHZ{} or not. We state here that such security guarantee is impossible, and why it is therefore meaningful to use \blindCanSup{} protocols instead.
  \begin{lemmaE}[Impossibility of a secure \blindCan{} protocol][]\label{lem:impossibilitySecureBlindCan}
    There exists no protocol \blindCan{} such that, at the end of an honest interaction, all supported applicants share a canonical \GHZ{}, and such that none of them know their own support status. More precisely, there exists always an adversary $\cA$ that can win the game \refGame*{IMPOSSIBLEGAMEOne}.
    \begin{namedGame}[IMPOSSIBLEGAMEOne]{\texttt{ImpossibleGame1}}
      (\cM, \vect{d}_0^{(0)},\vect{d}_0^{(1)}) \gets \cA_1(1^\lambda)\pcskipln\\
      \pclinecomment{\text{Avoid trivial attack: check if at least one honest applicant is in the \GHZ{}}}\\
      \pcif (\forall i \notin \cM, \vect{d}_0^{(0)}[i] = 0) \text{ or } (\forall i \notin \cM, \vect{d}_0^{(1)}[i] = 0) \pcthen \pcreturn \pcfalse \ \pcfi\\
      \pclinecomment{\text{Run \blindCan{} with adversary.}}\\
      \tilde{c} \gets \cA_3()\\
      \pcreturn \tilde{c} = c      
    \end{namedGame}
  \end{lemmaE}
  \begin{proofE}
    So for simplicity, we do here a sketch of the proof, in order to give the general ideas. So when it comes to proving the security of the protocol, we realize that at least one of the supported applicants needs to be honest, otherwise it is trivial to distinguish any correct protocol: The attacker can always send $\vect{d}_0^{(0)} = \begin{pmatrix}1 & 1 & 0 \dots & 0 \end{pmatrix}$ and $\vect{d}_0^{(1)} = \begin{pmatrix}0 & 0 & 0 \dots & 0 \end{pmatrix}$, and at the end of any (correct) protocol run honestly, the attacker will get either a Bell pair on the first two qubits or two qubits not entangled. It is therefore easy to distinguish, so the condition in Game 1 is indeed required. But it is not enough: even if we assume that one applicant is honest (let's say the first one), it is still impossible to prove the security of the protocol.

    Indeed, let's consider an adversary that sends $\vect{d}_0^{(0)} = \begin{pmatrix}1 & \dots 1 \end{pmatrix}$ and $\vect{d}_0^{(0)} = \begin{pmatrix}1 & \dots & 1 & 0 \dots & 0 \end{pmatrix}$ (first half of the qubits is $1$ second half is $0$). Then, a first remark is that at the end of an honest protocol, all the qubits that are not entangled must be all equal, i.e.\ if $c=1$, the state obtained is $(\ket{0\dots0} + \ket{1\dots1}) \otimes \ket{0\dots 0}$ or $(\ket{0\dots0} + \ket{1\dots1}) \otimes \ket{1\dots 1}$. Indeed, if some qubits in the second half are different, then a measurement in the computational basis will reveal some different outcomes with high probability (while when $c=0$ all measurements are equal since the state is a canonical \GHZ{} state by the correctness property). But even in that case, it is still easy to distinguish: when we do the measurement, in the first case, we either get $1\dots 1$ or $0 \dots 0$. In the second case, however, the first part may be different compared to the second part, i.e.\ we can measure either a $0\dots 0$ or a $1\dots 10\dots 0$ with probability $\frac{1}{2}$. This last measurement is enough to distinguish, we can just ask to $\cA$ to measure the state in the computational basis: if all measurements are equal, $\cA$ picks $\tilde{c}$ uniformly at random, otherwise $\cA$ outputs $\tilde{c} = 1$. $\cA$ will succeed with non negligible advantage.
  \end{proofE}

  Therefore it is not possible to hide to an adversary its support status, so the best we can get is to prove that no adversary can learn the support status of the honest applicants, which is the goal of the protocol $\blindCanSup{}$.
}

\subsection{The protocol \blindCanSup{}}\label{subsec:blindCanSup}
\pgfkeys{/prAtEnd/local custom defaults/.style={category=proofsBlindCanSup}}

We present now the protocol \blindCanSup{} (\cref{protocol:blindCanSup}): at the end of the protocol, the supported applicants share a canonical \GHZ{} state, and each applicant knows their own support status.

\begin{protocol}[htbp]
  \caption{\blindCanSup{}}\label{protocol:blindCanSup}
  \textbf{Inputs}: Same as \blind{}: Cupid gets $\vect{d}_0$ and $\lambda$.\\
  \textbf{Assumptions}: There exists a \AssumpFctCan{} family of functions with $\delta = \negl[\lambda]$.\\
  \textbf{Protocol}:%
  \begin{compressedList}
    \item Run the protocol \blind{}, so that Cupid gets $(b,y)$ and each applicant $a_i$ the $i$-th qubit
    \item Cupid: Compute $v \gets \partialInfo(t_k,y)$, and if $f_k^{-1}(y) = \{x,x'\}$ with $x \neq x'$, compute $\alpha \eqdef \langle b, x \xor x'\rangle$ (otherwise, sample $\alpha$ randomly). Computes the supported set $\cS = \{i \mid \vect{d}_0[i] = 1\}$. Sample uniformly at random $\hat{\alpha} \leftarrow \{0,1\}^n$ such that $\alpha = \xor_{i \in \cS} \hat{\alpha}_i$. For all $i$, send $(\hat{\alpha}_i,v[i])$ to applicant $a_i$.
    \item All applicants: When receiving the message $(\hat{\alpha}_i,v[i])$:
    \begin{itemize}
      \item If $v[i] = \cross$, then it means that the applicant is not part of the support of the final \GHZ{}. The end.
      \item If $v[i] = \bot$, it is a local abort. It's likely that the server was malicious. Do not reveal this information to the server. The end.
      \item If $v[i] \in \{0,1\}$, it means that this applicant is part of the final \GHZ{} state. Apply $Z^{\hat{\alpha}_i}X^{v[i]}$ on the qubit sent by the server.
    \end{itemize}
  \end{compressedList}
\end{protocol}

\begin{lemmaE}[Correctness of \blindCanSup{}][end if published]
  If all parties are honestly running the \blindCanSup{} protocol, then at the end of the protocol, with probability $1-\delta$ (so with overwhelming probability when $\delta$ is negligible), all supported applicants share a canonical \GHZ{}, and all applicants know whether or not they are supported.
\end{lemmaE}
\begin{proofE}
  With probability $1-\delta$, the $y$ obtained by Cupid has exactly two preimages. In that case, due to the correctness property of \partialInfo{} given \cref{def:GHZcanCapable} (part 1) we get $\forall i$, $v[i] \neq \bot$, so $v[i] \in \{0,1,\cross\}$. Then, using \cref{lem:correctnessBlind}, we know that the state shared by all participants after the \blind{} part is $\ket{h(x)}+{(-1)}^\alpha\ket{h(x')}$, with $h(x) \xor h(x') = \vect{d}_0$. We can combine this using part 2 of the correctness property given in \cref{def:GHZcanCapable}, (that states that $v[i] \in \{0,1\}$ iff $\vect{d}_0[i] = 1$): because the set of supported participants is $\cS \eqdef \{i \mid \vect{d}_0[i] = 1\}$, we have for all $i \notin \cS$: $h(x)[i] = h(x')[i]$. Thus the register of each applicant $i \notin \cS$ is in a tensor product with all the other qubits, so we can factor them out, and consider only the state shared by applicants $i \in \cS$ (in that case $h(x')[i] = 1 \xor h(x)[i]$):
  \begin{align}
    \bigotimes_{i \in \cS} \ket{h(x)[i]} + (-1)^\alpha\bigotimes_{i \in \cS} \ket{h(x')[i]}
    &= \bigotimes_{i \in \cS} \ket{h(x)[i]} + (-1)^\alpha\bigotimes_{i \in \cS} \ket{1 \xor h(x)[i]}
  \end{align}
  After the corrections, the state becomes:
  \begin{align}
    \bigotimes_{i \in \cS} X^{v[i]}\ket{h(x)[i]} + (-1)^{\alpha \xor \bigoplus_{i \in \cS} \hat{\alpha}_i}\bigotimes_{i \in \cS} X^{v[i]}\ket{1 \xor h(x)[i]}
  \end{align}
  And due to the fact that $\alpha = \bigoplus_{i \in \cS} \hat{\alpha}_i$, we can get rid of the phase. Moreover, we can now use the part 3 of \cref{def:GHZcanCapable} which states that there exists $u \in \{h(x),h(x')\}$ such that if $\vect{d}_0[i] = 1$, then $v[i] = u[i]$. So if $\forall i \in \cS$, we have $u[i] = v[i] = h(x)[i]$, then after the correction we get the state $\ket{0 \dots 0}+\ket{1 \dots 1}$, which is a canonical \GHZ{}, and if $u[i] = v[i] = h(x')[i] = 1-h(x)[i]$, then we get $\ket{1 \dots 1} + \ket{0 \dots 0}$, which is the same canonical \GHZ{} state.
\end{proofE}
Similarly, we define now the security of \blindCanSup{}:

\begin{lemmaE}[Security of \blindCanSup{}][end if published]\label{lem:securityBlindCanSup}
  If we allow in the protocol \blindCanSup{} the fully malicious server to corrupt a subset of applicants in such a way that either at least one supported applicant is not corrupted or no supported applicant is corrupted\footnote{Otherwise there is a trivial, fundamental, attack to any protocol which consists in setting $\vect{d}_0^{(0)} = (01)$, $\vect{d}_0^{(1)}=(11)$, $\cM=\{2\}$ and then testing if the quantum state obtained by the party $2$ is a $\ket{+}$ or not. \arxivOnly{However, this attack is not possible anymore if the adversary is not in possession of one part of the \GHZ{} (for example if we replace $\vect{d}_0^{(0)} = (00)$ and $\vect{d}_0^{(1)}=(10)$ in the above example), that is the reason why we can provide a stronger security guarantee when no supported applicant is supported.}}, then the support status of honest applicants is completely hidden. More precisely, no \QPT{} adversary $\cA = (\cA_1,\cA_2,\cA_3)$ can win the game $\INDBLINDCANSUP{}$ (\cref{fig:gameIndBlindCanSup}) with probability better than $\frac{1}{2} + \negl[\lambda]$. In the following, $\cM$ is the set of malicious applicants corrupted by Bob, and the condition $\forall i \in \cM, \vect{d}_0^{(0)}[i] = \vect{d}_0^{(1)}[i]$ is required to avoid a trivial uninteresting distinguishing strategy.
  \begin{figure}[htb]
    \centering
    \begin {pcimage}
      {\normalfont\game[linenumbering]{$\INDBLINDCANSUP{}_{\{f_k\}} ^\adv (\lambda) $}{%
          (\cM, \vect{d}_0^{(0)},\vect{d}_0^{(1)}) \gets \cA_1(1^\lambda)\\
          \pcif \exists i \in \cM, \vect{d}_0^{(0)}[i] \neq \vect{d}_0^{(1)}[i]\pcthen \pcreturn \pcfalse\ \pcfi\\
          \pcif (\exists i \in \cM, \vect{d}_0^{(0)} = 1) \text{ and } ((\forall i \notin \cM, \vect{d}_0^{(0)}[i] = 0) \text{ or } (\forall i \notin \cM, \vect{d}_0^{(1)}[i] = 0))\\
          \pcthen \pcreturn \pcfalse\ \pcfi\\
          c \sample \bin\\
          (k, t_k) \gets \Gen(1^\lambda, \vect{d}_0^{(c)})\\
          (y,b) \gets \cA_2(k)\\
          v \gets \partialInfo(t_k,y)\\
          \pcif \bot \notin v \pcthen \alpha \eqdef \langle b, x \xor x' \rangle\ \pcelse \alpha \sample \bin \ \pcfi\\
          \hat{\alpha} \gets \{ \hat{\alpha} \mid \hat{\alpha} \in \{0,1\}^n, \bigoplus_{i \in \cS} \hat{\alpha}_i = \alpha \text{ or } \cS = \emptyset\}\\
          \pclinecomment{\text{The adversary has only access to the messages sent by Cupid to corrupted applicants:}}\\
          \tilde{c} \gets \cA_3(\{(i, \hat{\alpha}_i, v[i])\}_{i \in \cM})\\
          \pcreturn \tilde{c} = c
        }}
    \end{pcimage}
    \caption{\label{fig:gameIndBlindCanSup}Game $\INDBLINDCANSUP$ required in \cref{lem:securityBlindCanSup}}
  \end{figure}
\end{lemmaE}
\begin{proofE}
  The first step in the proof is to note that, due to the condition line 3, we have either:
  \begin{itemize}
    \item $\forall i \in \cM$, $\vect{d}_0^{(0)} = 0$, i.e.\ $\cM \cap \cS = \emptyset$. But since we send $\hat{\alpha}_i$ to the adversary only if $i \in \cM$, and since the line 10 does not put any restriction on the sampling $\hat{\alpha}_i$ when $i \notin \cS$, the line 9 and 10 can be replaced with a single line $\hat{\alpha} \gets \{0,1\}^n$.
    \item or there exists $j \notin \cM$ such that $\vect{d}_0^{(c)}[j] = 1$, i.e.\ such that $j \in \cS$. Therefore, $\hat{\alpha}$ can be sampled by choosing for all $i \neq j$, $\hat{\alpha}_i$ randomly, and finally by setting $\hat{\alpha}_j = \alpha \xor\bigoplus_{i \in \cS \setminus \{j\}} \hat{\alpha}_j$ (this is statistically indistinguishable). But since $\hat{\alpha}_j$ is never sent to $\cA$ because $j \notin \cM$, we can also remove lines 9 and 10 and replace them with $\hat{\alpha} \gets \{0,1\}^n$.
  \end{itemize}
  This gives us a new game \texttt{GAME1}:
  \begin {pcimage}
    {\normalfont\game[skipfirstln,lnstart=8,linenumbering]{ $\texttt{GAME1}^\cA$ }{
        \pclinecomment{\dots First 7 lines like \INDBLINDCANSUP}\\
        \hat{\alpha} \gets \{0,1\}^n\\
        \tilde{c} \gets \cA_3(\{(i, \hat{\alpha}_i, v[i])\}_{i \in \cM})\\
        \pcreturn \tilde{c} = c
      }}
  \end{pcimage}
  Since the two games are exactly equivalent, we have $\pr{\INDBLINDCANSUP^\cA} = \pr{\texttt{GAME1}^\cA}$. Then, define a new game \texttt{GAME2} by removing the condition line 3 and 4:
  \begin {pcimage}
    {\normalfont\game[skipfirstln,lnstart=2,linenumbering]{ $\texttt{GAME2}^\cA$ }{
        \pclinecomment{\dots Remove line 3 and 4 of \texttt{GAME1}}\\
        \hcancel[red]{\pcif (\exists i \in \cM, \vect{d}_0^{(0)} = 1) \text{ and } ((\forall i \notin \cM, \vect{d}_0^{(0)}[i] = 0) \text{ or } (\forall i \notin \cM, \vect{d}_0^{(1)}[i] = 0))}\\
        \hcancel[red]{\pcthen \pcreturn \pcfalse\ \pcfi}\\
        \pclinecomment{\dots Rest is like \texttt{GAME1}}
      }}
  \end{pcimage}
  We can remark that this condition cannot help the adversary to win since entering inside this condition always returns ``false'', therefore $\pr{\texttt{GAME1}^\cA} \leq \pr{\texttt{GAME2}^\cA}$. But now, we remark that \texttt{GAME2} is very similar to the game \INDPARTIAL{} defined \cref{def:GHZcanCapable}, except that we provide an additional random string $\hat{\alpha}$ to $\cA$. But since this string is random, it is easy to see that we can turn any adversary $\cA$ winning \texttt{GAME2} with probability $p$ into an adversary $\cA'$ winning the game \INDPARTIAL{} with the same probability $p$ by defining $\cA'(\{(i,v[i])\}_{i \in \cM})$ as an adversary sampling a uniformly random bit string $\hat{\alpha}$ and calling $\cA(\{(i,\hat{\alpha}_i,v[i])\}_{i \in \cM})$ (and reciprocally, any adversary that can win \INDPARTIAL{} without access to $\hat{\alpha}$ can win  \texttt{GAME2} with access to $\hat{\alpha}$ by simply forgetting this value). So we get:
  \begin{align}
    \max_{\text{QPT }\cA}\pr{ \texttt{GAME2}^\cA } = \max_{\text{QPT }\cA}\pr{ \INDPARTIAL{}^\cA }
  \end{align}
  But by assumption, for any QPT $\cA$, $\pr{ \INDPARTIAL{}^\cA } \leq \frac{1}{2} + \negl[\lambda]$. So by combining all the inequations we showed on games, we also get:
  \begin{align}
    \pr{\INDBLINDCANSUP{}_{\{f_k\}} ^\adv (\lambda)} \leq \frac{1}{2} + \negl[\lambda]
  \end{align}
  which ends the proof.
\end{proofE}

\subsection{The protocol \authBlindCanDist{}}\label{subsec:authBlindCanDist}
\pgfkeys{/prAtEnd/local custom defaults/.style={category=proofsAuthBlindCanDist}}

We can now define our main protocol \authBlindCanDist{}. Similarly to the \blindCanSup{} protocol, each supported applicant is supposed to end up with a canonical \GHZ{} state, and the support status of each applicant should be unknown to the other applicants and to the server. However, in this protocol the trusted party Cupid is not needed anymore: each applicant is supposed to choose themselves their own support status, and they will be assured that no malicious party (including the server) can learn it.

Moreover, the server can have some guarantees on the support status of the applicants: for example, the server can ensure that \emph{if} some applicants are supported, then they all know a classical secret (but the server has no way to know whether or not a given applicant is supported). This secret can be any witness of a \npol{} relation: it could be a password, a private key linked with some known public key, a signature from a third party Certification Authority, the proof of any famous theorem\dots{} We formalize it by defining $n$ deterministic functions $\Auth_i\colon \{0,1\} \times \{0,1\}^* \rightarrow \{0,1\}$ responsible of the ``authorization'' of the applicants: the server will allow applicant $i$ to be part of the protocol iff they can prove in a NIZK way that they know $w_c$ such that $\Auth_i(\vect{d}_0[i],w_c)=1$. For instance, we can use the $\Auth_i$ function to ensure that an applicant is part of the \GHZ{} iff they know a password whose hash by $h$ is $x$ by defining $\Auth_i(\vect{d}_0[i],\tilde{s}_i)) \eqdef (\vect{d}_0[i] = 0 \lor h(\tilde{s}_i) = x)$. Again, we emphasize that $\Auth_i$ does \emph{not} reveal the value of $\vect{d}_0[i]$: it just reveals that \emph{if} the user is supported, then they know the password. This verification only requires a single message from the client(s) and is therefore achieving NIZKoQS, as formalized in \cref{sec:NIZKoQS}.

\begin{figure}[htbp]
  \centering
  \begin{minipage}{1.0\linewidth}
    \begin{protocol}[H]
      \caption{\authBlindCanDist{}}\label{protocol:authBlindCanDist}
      \noindent\textbf{Inputs}: Each applicant $i$ gets $\lambda$. They also get $\vect{d}_0[i] \in \{0,1\}$ and $w_i \in \{0,1\}^*$ such that $\Auth_i(\vect{d}_0[i],w_i) = 1$ ($\vect{d}_0[i] = 1$ iff applicant $i$ wants to be supported). The authentication functions $\{\Auth_i\}_{i \in [n]}$ are also public.\\
      \textbf{Assumptions}: There exists a \AssumpFctCan{} distributable family of functions with $\delta = \negl[\lambda]$. Cupid is not required anymore, and instead we require the existence of a classical (but quantum-secure) Multi-Party Computation protocol.\\
      \textbf{Protocol}:%
      \begin{compressedList}
        \item Each applicant $a_i$: Run $ (k^{(i)},t_k^{(i)}) \leftarrow \GenLocal(1^\lambda,\vect{d}_0[i])$, send $k^{(i)}$ to the server, and continue the protocol.
        \item Server: Run (as a verifier) a Zero-Knowledge protocol with each applicant (the prover) to check that $k^{(i)}$ is well prepared, and that the applicant can authenticate the quantum state. More precisely, each applicant $i$ proves to the server that they know $(\vect{d}_0[i], t_k^{(i)},w^{(i)})$ such that $\CheckHonestTrapdoor_\lambda(\vect{d}_0[i], t_k^{(i)}, k^{(i)}) \land \Auth_i(\vect{d}_0[i],w^{(i)}) = 1$. If the protocol fails with at least one applicant, abort after sending $\bot$ to all applicants (the server can also output if needed the identity of the applicant who were malicious in case other actions should be performed with respect to them, e.g.\ in further runs). Otherwise, the protocol continues.
        \item Server: Compute $k \eqdef (k^{(1)},\dots,k^{(n)})$, run the quantum circuit already described in protocol \blind{}, and for all $i$, send the $i$-th qubit of the second register to applicant $a_i$ together with $(y,b)$.
        \item For each applicant $i$: Compute $v[i] \gets \partialInfoLocal(t_k^{(i)},y)$, and compute via a MPC protocol the function $\CombineAlpha$ which returns a secret share of $\alpha$ between supported applicants. More precisely, it returns to each applicant $i$ a random bit $\hat{\alpha}_i$ such that $\xor_{i \in \cS} \hat{\alpha}_i = \alpha = \bigoplus_i \partialAlpha(i, t_k^{(i)},y,b)$, where $\cS$ is the set of supported applicants (details in \cref{fig:functionComputeMPC}). The reason we use a MPC protocol is that $\partialAlpha(i, t_k^{(i)},y,b)$ can leak\footnote{The attack would be as follows: the malicious server Bob can run the QFactory \cite{CCKW_2019_qfactory} protocol with $k^{(i)}$ which gives him a BB84 state in the basis $\vect{d}_0[i]$: if $\vect{d}_0[i] = 0$ then it gets either $\ket{0}$ or $\ket{1}$, but if $\vect{d}_0[i] = 1$ then it gets the state $\ket{+}$ if $\alpha_p^{(i)}{\alpha}_i = 0$ and a $\ket{-}$ otherwise. So the trick is to measure the state in the Hadamard basis: if the measurement is different from $\alpha_p^{(i)}$ then we know that $\vect{d}_0[i] = 0$, and otherwise the server will randomly guess the value of $\vect{d}_0[i]$. It is easy to see that if $\vect{d}_0[i]$ is chosen uniformly at random, then the server has a non-negligible advantage in guessing $\vect{d}_0[i]$.} information about the bit $\vect{d}_0[i]$, so this bit should not be revealed directly.
        \item For each applicant $i$: If the outcome $\hat{\alpha}_i$ of the MPC is $\bot$, abort. Otherwise, similarly to the last step of the \blindCanSup protocol: if $v[i] \in \{0,1\}$, apply the correction $X^{v[i]}Z^{(\hat{\alpha}_i)}$ on the qubit, else discard the qubit.
      \end{compressedList}
    \end{protocol}
  \end{minipage}
\end{figure}

We will now prove the correctness and security of the \authBlindCanDist{} protocol. Note that an honest server can obtain guarantees on the distributed state even in the presence of malicious or noisy applicants. Assuming here an honest server is not absurd, notably when the server wants to use this \GHZ{}, for example to share a quantum state or if a verification is done afterwards. This centralization is also useful in the presence of many noisy clients (a single hardware needs to be noiseless, while in a decentralized MPC computation the protocol is likely to always abort if a single client is noisy). We provide here an informal version of the correctness property, more details can be found in \cref{lem:correctnessCanDist}.
\begin{lemmaE}[Correctness of \authBlindCanDist{} in the presence of malicious applicants, Informal]
  In the presence of malicious applicants, an honest server is guaranteed that with overwhelming probability, the protocol will either abort, or the applicants will obtain a \GHZ{} state, up to some unavoidable local deviation performed by supported malicious applicants on their own parts of the \GHZ{}.
\end{lemmaE}

\inAppendixIfPublished{
  \begin{figure}[htb]
    \centering
    \begin {pcimage}\label{fct:mpcCombineAlpha}
      {\normalfont\game[linenumbering]{$\CombineAlpha((k,y,b),(t_k^{(1)},\vect{d}_0[1],k^{(1)},y^{(1)}, b^{(1)}),\dots,(t_k^{(n)},\vect{d}_0[n],k^{(1)},y^{(n)},b^{(n)}))$}{%
          \pclinecomment{Check if the input are honestly prepared}\\
          \pcif k \neq (k^{(1)},\dots,k^{(n)}) \text{ or }\\
          \t \exists i, y^{(i)} \neq y \text{ or } b^{(i)} \neq b \text{ or } \neg \CheckHonestTrapdoor_\lambda(\vect{d}_0[i], t_k^{(i)}, k^{(i)})\\
          \pcthen \pcreturn \bot^{n+1} \ \pcfi\\
          \pclinecomment{Compute the correction $\alpha$, the set of supported applicants, sample a first version of $\hat{\alpha}$}\\
          \alpha = \bigoplus_i \partialAlpha(i, t_k^{(i)},y,b);
          \quad \cS = \{i \mid \vect{d}_0[i] = 1\};
          \quad \forall i, \hat{\alpha}_i = \bigoplus_l r^{(l)}[i]\\
          \pcif \cS \neq \emptyset \pcthen \pclinecomment{If at least one person is supported, ensure $\xor_{i \in \cS} \hat{\alpha_i} = \alpha$}\\
          \t j = \max_{i \in \cS} i \pclinecomment{Pick an arbitrary $j \in \cS$ to change}\\
          \t \hat{\alpha_j} = \alpha \xor \smashoperator{\bigoplus_{i \in \cS \setminus \{j\}}} \hat{\alpha_i}\\
          \pcfi \pclinecomment{Return $\hat{\alpha}_i$ to applicant $i$ and $\top$ to the server to indicate no problem occurred.}\\
          \pcreturn (\top,\hat{\alpha}_1,\dots,\hat{\alpha}_n)
        }}
    \end{pcimage}%
    \caption{The function to compute in the \authBlindCanDist{} protocol in a MPC way. The first input is the input of the server, and the other inputs are from the applicants (the $y^{(i)}$ and $b^{(i)}$ are supposed to be equal to $y$ and $b$ and are just used to ensure that the server provided coherent inputs in the MPC, and $r^{(i)} \in \{0,1\}^n$ is a string supposed to be sampled uniformly at random).}
    \label{fig:functionComputeMPC}
  \end{figure}
}

\begin{lemmaE}[Correctness of \authBlindCanDist{} in the presence of malicious applicants][all end if published]\label{lem:correctnessCanDist}
  If the server is honest, and if we allow an attacker to corrupt an arbitrary subset $\cM$ of applicants, then with overwhelming probabilities, the protocol either aborts, or the $k$ received by the server belongs to $\cK$, and for all applicant $i$, there exists $w_i$ such that $\Auth_i(\vect{d}_0[i],w_i) = 1$. If the ZK protocol is also a Proof of Knowledge protocol, then the applicant ``knows'' $w_i$, in the sense that if the adversary can pass the test with non-negligible probability, there exists an extractor that can extract $w_i$ given the applicant's circuit with non-negligible probability (this is a direct application of \cref{def:QPoK})\footnote{In particular, if it is impossible to forge with non-negligible probability a $w_i$ such that $\Auth_i(\vect{d}_0[i],w_i)=1$ (for instance because this $w_i$ is a signature coming from an unforgeable signature scheme), then it means that with overwhelming probability the applicant is indeed in possession of $w_i$.}.

  Moreover, if the protocol did not abort before, at the end of the protocol, with probability $1-\delta-\negl[\lambda]$ (i.e.\ overwhelming if $\delta$ is negligible), the protocol will either abort, or a state will be obtained by applicants. In this later case, if we denote by $\rho_{A,\cM}$ the joint state of the honest applicants (register $A$) and of the adversary (register $\cM$) obtained at the end of the protocol, then $\rho_{A,\cM}$ can be written as a Completely Positive Trace Preserving (CPTP) map\footnote{\label{note:prooflemma57}Note that for simplicity, we just require the existence of this CPTP map, and therefore it can depend on any quantity, including $k$, $\vect{d}_0$\dots Therefore we don't require this map to be efficiently computable since it will not be useful for us. However, a similar ``efficient'' version should be derivable if we make sure our ZK protocol is extractable, i.e.\ that the trapdoor of each $k^{(i)}$ can be extrated by a simulator. This is however out of the scope of this paper.} applied on a \GHZ{} state shared among all parties $i$ such that $\vect{d}_0[i]=1$, in such a way that the CPTP map leaves untouched the qubits of the \GHZ{} state owned by honest applicants $i$. In particular, if all supported parties are honest, they all share a \GHZ{} state.
\end{lemmaE}
\begin{proofE}
 The first action of the server (which is assumed to be honest here) is to run a ZK protocol to check that $\forall i, k^{(i)} \in \cK_{\lambda,\texttt{Loc}}$ and $\Auth_i(\vect{d}_0[i],w_i)=1$. Therefore, we can use the soundness property of the ZK protocol to claim that with overwhelming probability $\forall i, k^{(i)} \in \cK_{\lambda,\texttt{Loc}}$ and there exist $w_i$ such that $\Auth_i(\vect{d}_0[i],w_i)=1$ (since the provers are the applicants, they are bounded so we can rely on both computational or statistical soundness). The fact that $w_i$ is actually ``known'' to the applicant comes directly from the fact that the ZK protocol is a Proof of Knowledge and is extractable. So with overwhelming probability, $k \eqdef (k^{(1)},\dots,k^{(n)})$ belongs to $\cK \eqdef \cK_{\lambda,\texttt{Loc}}^n$. Therefore, since the server is honest, with probability $1-\delta$ it will measure a $y$ such that $|f_k^{-1}(y)|=2$, and due to the properties of the family $f_k$, the two preimages $x$ and $x'$ are such that $h(x) \xor h(x') = \vect{d}_0$. So the state sent by the server is $\ket{h(x)} + (-1)^a\ket{h(x')}$ with $a = \langle b, x \xor x' \rangle$. Then, the MPC protocol will be performed. If the MPC aborts, then we are already in the setting of the theorem. If the MPC does not aborts, then, due to the fact that in quantum mechanics, operations performed by two non-communicating parties commute, without any loss of generality we can assume that the honest applicants will apply the correction before the deviation of the malicious party. Moreover, since the honest corrections are unitary, we can also assume without any loss of generality that the first step of the malicious party is to apply the honest correction on the state received from the server and then deviate (eventually by starting to undo the correction). Note that we do not even ask this correction to be efficiently computable by the adversary (see \cref{note:prooflemma57}) since we just claim that such deviation exists. Due to the definition of $\partialInfo$, after applying the $X$ correction, the parties $i$ for which $\vect{d}_0[i]=1$ share a state $\ket{0 \dots 0} + (-1)^\alpha\ket{1 \dots 1}$. Now, we have two cases:
  \begin{enumerate}
    \item If there exists at least one supported applicant which is malicious, then the proof is done: no matter what are the values of $\hat{\alpha}_i$ which will be used by the honest applicant to correct the state, we can always include in the CPTP map a first step that applies $Z^{\alpha \xor\bigoplus_{i \in \cS, i \notin \cM} \hat{\alpha}_i}$ on the qubit of the malicious applicant to map the step back to a \GHZ{} state. Again, this is possible since we just require the existence of the CPTP map. Then, any CPTP deviation can be applied on the state owned by the malicious adversary, including undoing the previous $Z$ and $X$ corrections.
    \item If there exists no malicious supported applicant, and if the probability of having no abort and no malicious supported applicant is non-negligible\footnote{If on the other hand this quantity is negligible, then this second case occurs with negligible probability so it is absorbed in the $\negl[\lambda]$ of the theorem.}, then with overwhelming probability we must have $\bigoplus_{i \in \cS}\hat{\alpha}_i = \alpha$. Indeed, if it is not the case, then it is possible to distinguish the real world from the ideal world of the MPC computation. Therefore, after the $Z$ correction the honest applicants having $\vect{d}_0[i]=1$ will share a canonical \GHZ{}, which ends the proof.
  \end{enumerate}
\end{proofE}

\begin{lemmaE}[Blindness of \authBlindCanDist{} in the presence of malicious applicants][end if published]\label{lem:blindAuthBlindCanDist}
  If the server corrupts a set of applicants, in such a way that at least one supported applicant is not corrupted, or that no supported applicant is corrupted, then the support status of the honest applicants is hidden in the \authBlindCanDist{} protocol, beyond the fact that server knows whether or not they can pass the authorization step. More formally, no adversary can win the game $\INDBLINDCANAUTH$. 
    {\normalfont \game[linenumbering]{$\INDBLINDCANAUTH{}_{\{f_k\}} ^\adv (\lambda) $}{%
        (\cM, \vect{d}_0^{(0)}, \{(i,w_i^{(0)})\}_{i \in [n] \setminus \cM}),\vect{d}_0^{(1)}, \{(i,w_i^{(1)})\}_{i \in [n] \setminus \cM}) \gets \cA_1(1^\lambda)\\
        \pcif \exists i \in \cM, \vect{d}_0^{(0)}[i] \neq \vect{d}_0^{(1)}[i]\pcthen \pcreturn \pcfalse\ \pcfi\\
        \pcif (\exists i \in \cM, \vect{d}_0^{(0)} = 1) \text{ and } ((\forall i \notin \cM, \vect{d}_0^{(0)}[i] = 0) \text{ or } (\forall i \notin \cM, \vect{d}_0^{(1)}[i] = 0))\\
        \pcthen \pcreturn \pcfalse\ \pcfi\pcskipln\\
        \pclinecomment{Check that the adversary did not gave wrong witnesses $w_i$:}\\
        \pcif \exists i \in [n], \exists c \in \{0,1\}, \Auth_i(\vect{d}_0^{(c)}[i],w_i^{(c)}) \neq 1 \pcthen\pcreturn \pcfalse \ \pcfi\\
        c \sample \bin; \vect{d}_0 \eqdef \vect{d}_0^{(c)}; w_i \eqdef w_i^{(c)}\\
        \text{Run with $\cA_2$ the protocol \authBlindCanDist{}.}\\
        \tilde{c} \gets \cA_3\\
        \pcreturn \tilde{c} = c
      }}
\end{lemmaE}
\begin{proofE}
  The above game is more formally defined in \cref{fig:gameFormalBlindAuthBlindCanDist}.
  \begin{figure}[htb]
    \centering
    {\normalfont\game[linenumbering]{$\texttt{GAME1}^\adv (\lambda) $}{%
        (\cM, \vect{d}_0^{(0)}, \{(i,w_i^{(0)})\}_{i \in [n] \setminus \cM}),\vect{d}_0^{(1)}, \{(i,w_i^{(1)})\}_{i \in [n] \setminus \cM}) \gets \cA_1(1^\lambda)\\
        \label{my:line:securauth1d00eqd01}\pcif \exists i \in \cM, \vect{d}_0^{(0)}[i] \neq \vect{d}_0^{(1)}[i]\pcthen \pcreturn \pcfalse\ \pcfi\\
        \label{my:line:securauth1CorrptGoodSet}\pcif (\exists i \in \cM, \vect{d}_0^{(0)} = 1) \text{ and } ((\forall i \notin \cM, \vect{d}_0^{(0)}[i] = 0) \text{ or } (\forall i \notin \cM, \vect{d}_0^{(1)}[i] = 0))\\
        \pcthen \pcreturn \pcfalse\ \pcfi\pcskipln\\
        \pclinecomment{Check that the adversary did not gave wrong witnesses $w_i$:}\\
        \label{my:line:securauth1TestAuthOk}\pcif \exists i \in [n], \exists c \in \{0,1\}, \Auth_i(\vect{d}_0^{(c)}[i],w_i^{(c)}) \neq 1 \pcthen\pcreturn \pcfalse \ \pcfi\\
        c \sample \bin; \vect{d}_0 \eqdef \vect{d}_0^{(c)}; w_i \eqdef w_i^{(c)}\\
        \forall i \notin \cM, (k^{(i)}, t_k^{(i)}) \gets \GenLocal(1^\lambda, \vect{d}_0[i])\\
        \label{my:line:secureauth1A2}\cA_2(\{k^{(i)}\}_{i \notin \cM})\\
        \label{my:line:secureauth1ProofZK}\pcfor i \notin \cM \pcdo\\
        \t \text{Prove in ZK with $\cA_{2,i}$ that $\CheckHonestTrapdoor_\lambda(\vect{d}_0[i], t_k^{(i)}, k^{(i)}) \land \Auth_i(\vect{d}_0[i],w_i) = 1$}\\
        \pcendfor\\
        \label{my:line:secureauth1A2aborts}\pcif \text{$\cA_2$ aborts} \pcthen \text{wait for $\tilde{c}$ from $\cA_2$. } \pcreturn \tilde{c} = c \ \pcfi\\
        (y,b) \gets \cA_3\\
        \label{my:line:securauth1MPC}\text{Compute in a MPC way the \CombineAlpha{} function, where $\cA_4$ controls adversaries in $\cM$.}\pcskipln\\
        \pclinecomment{All others operations are not sent to the adversary, and operations}\pcskipln\\
        \pclinecomment{applied on the quantum state do not change anything due to non-signaling.}\\
        \label{my:line:securauth1Guess}\tilde{c} \gets \cA_5\\
        \pcreturn \tilde{c} = c
      }}
    \caption{}
    \label{fig:gameFormalBlindAuthBlindCanDist}
  \end{figure}
  We will prove the above theorem by using a hybrid argument. First, we can easily see that if the adversary corrupts all applicants ($\cM = [n]$), then it cannot win the game with probability better than $\frac{1}{2}$. Indeed, the line \ref{my:line:securauth1d00eqd01} forces $\vect{d}_0^{(0)}=\vect{d}_0^{(1)}$ (and both maps $w_i$ are empty), therefore the view of the adversary is exactly the same for $c=0$ and $c=1$. So we can define a new hybrid game $\texttt{GAME2}$ in which we return $\pcfalse$ if $\cM = [n]$:
  \begin {pcimage}
    {\normalfont\game[skipfirstln,lnstart=1,linenumbering]{ $\texttt{GAME2}^\cA$ }{
        \pclinecomment{Just update the line \ref{my:line:securauth1d00eqd01} of $\texttt{GAME1}$ as follows: }\\
        \pcif \textcolor{green!50!black}{\cM = [n] \text{ or } } \exists i \in \cM, \vect{d}_0^{(0)}[i] \neq \vect{d}_0^{(1)}[i]\pcthen \pcreturn \pcfalse\ \pcfi\pcskipln\\
        \pclinecomment{Rest is like \texttt{GAME1}\dots}
      }}
  \end{pcimage}
  Then, we can turn any adversary $\cA$ winning $\texttt{GAME1}$ with probability $p$ into another adversary $\cA'$ winning $\texttt{GAME2}$ with probability $p$. To do so, $\cA'$ runs first $\cA_1$: if $\cA_1$ returns $\cM \neq [n]$, then $\cA'$ continues normally with $\cA$, otherwise if $\cM = [n]$ then $\cA'$ removes one element of $\cM$ (of course, $[n]$ is assumed to be non empty\dots), and $\cA'$ aborts when $\cA_2$ is supposed to run, and output a random $\tilde{c}$. Therefore, we have:
  \begin{align}
    \max_{\QPT{} \cA} \pr{\texttt{GAME1}^\cA}(\lambda) = \max_{\QPT{} \cA'} \pr{\texttt{ GAME2}^{\cA'}}(\lambda)
  \end{align}
  The second hybrid game that we define is the same as the game \texttt{GAME2}, except that we replace lines \ref{my:line:securauth1MPC} to \ref{my:line:securauth1Guess} with one line $(\tilde{c}, \vec{y}) \gets \Real_{\Uppi,\cA'}(\lambda,\vec{x},\rho_3)$, where $\rho_3$ is the final internal state of the adversary $\cA_3$, $\vec{x}$ contains the honest inputs of the MPC computation for the non-corrupted adversaries and dummy inputs for the corrupted adversaries (we will ignore them anyway), and $\cA'$ is the adversary that outputs the corrupted set $\cM$, that runs $\cA_4(\rho_3)$ followed by $\tilde{c} \gets \cA_5(\rho_4)$, where $\rho_4$ is the final internal state of $\cA_4$, and that finally returns $\tilde{c}$. This defines a new game $\texttt{GAME3}$:
  \begin {pcimage}
    {\normalfont\game[lnstart=13,linenumbering]{ $\texttt{GAME3}^\cA$ }{
        \hcancel[red]{\text{Compute in a MPC way the \CombineAlpha{} function, where $\cA_4$ controls adversaries in $\cM$.}}\\
        \hcancel[red]{\tilde{c} \gets \cA_5}\quad\textcolor{green!50!black}{(\tilde{c}, \vec{y}) \gets \Real_{\Uppi,\cA'}(\lambda,\vec{x},\rho_3)}
      }}
  \end{pcimage}
  Since this is perfectly equivalent from the point of view of the adversary (due to the definition or $\Real$), the probability of winning these two games are exactly the same: $\pr{\texttt{GAME2}^\cA} = \pr{\texttt{GAME3}^\cA}$. Now, because the MPC protocol is secure, there exists a simulator $\Sim$ fitting \cref{def:MPC}. Therefore, we can now define a new game ${\texttt{GAME4}}$, in which we replace the real world with the ideal world:
  \begin {pcimage}
    {\normalfont\game[lnstart=14,linenumbering]{ $\texttt{GAME4}^\cA$ }{
        \hcancel[red]{(\tilde{c}, \vec{y}) \gets \Real_{\Uppi,\cA'}(\lambda,\vec{x},\rho_3)}\quad\textcolor{green!50!black}{(\tilde{c}, \vec{y}) \gets \Ideal_{\CombineAlpha,\Sim}(\lambda,\vec{x},\rho_3)}
      }}
  \end{pcimage}
  Then, we have $\pr{\Real_{\Uppi,\cA'}(\lambda,\vec{x},\rho_3)[0] = c} \leq  \pr{\Ideal_{\CombineAlpha,\Sim}(\lambda,\vec{x},\rho_3)[0] = c} + \negl[\lambda]$ (otherwise we could distinguish between the real and ideal worlds), and therefore $\pr{\texttt{GAME3}^\cA} \leq \pr{\texttt{GAME4}^\cA} + \negl[\lambda]$. Now, we can define $\CombineRandom$, which is an adaptation of $\CombineAlpha$ that does not depend anymore on the secret values of the honest parties:
  \begin {pcimage}\label{fct:combineRandom}
    {\normalfont\game[linenumbering]{$\CombineRandom((k,y,b),(t_k^{(1)},\vect{d}_0[1],r^{(1)},k^{(1)},y^{(1)}, b^{(1)}),\dots,(t_k^{(n)},\vect{d}_0[n],r^{(n)},k^{(1)},y^{(n)},b^{(n)}))$}{%
        \pclinecomment{Check if the input are honestly prepared}\\
        \pcif k \neq (k^{(1)},\dots,k^{(n)}) \text{ or } \exists i, y^{(i)} \neq y \text{ or } b^{(i)} \neq b \text{ or } \exists i \in \cM, \neg \CheckHonestTrapdoor_\lambda(\vect{d}_0[i], t_k^{(i)}, k^{(i)})\\
        \pcthen \pcreturn \bot^{n+1} \ \pcfi\\
        \forall i, \hat{\alpha}_i = \bigoplus_l r^{(l)}[i]\\
        \pcreturn (\top,\hat{\alpha}_1,\dots,\hat{\alpha}_n)
      }}
  \end{pcimage}
  We can now define a new game $\texttt{GAME5}$ in which we substitute the $\CombineAlpha$ function with the $\CombineRandom$ function:
  \begin {pcimage}
    {\normalfont\game[lnstart=14,linenumbering]{ $\texttt{GAME5}^\cA$ }{
        (\tilde{c}, \vec{y}) \gets \Ideal_{\hcancel[red]{\CombineAlpha}\textcolor{green!50!black}{\CombineRandom},\Sim}(\lambda,\vec{x},\rho_3)
      }}
  \end{pcimage}
  Then, we have $\pr{\texttt{GAME4}^\cA} = \pr{\texttt{GAME5}^\cA}$. Indeed, by construction, the inputs of the honest parties always pass the $\CheckHonestTrapdoor$ test, so removing this test for the honest parties cannot help the adversary to distinguish the two games. Moreover, since at least one applicant is honest, the string $\bigoplus_l r^{(l)}$ is indistinguishable from a random string. Therefore, we can use the same trick used already in the proof of \cref{lem:securityBlindCanSup}: the condition line \ref{my:line:securauth1CorrptGoodSet} gives us two cases.
  \begin{itemize}
    \item If all malicious applicants are not supported, then, since the output of honest applicants are never given back the adversary, we don't need to update $\hat{\alpha}_j$
    \item Similarly, if at least one honest applicant is supported, then instead of updating $\hat{\alpha}_j$, we can update the $\hat{\alpha}_j$ or this applicant\dots{} But since the output of honest applicants are never given to the adversary, we don't even need to update it.
  \end{itemize}
  Therefore, $\pr{\texttt{GAME4}^\cA} = \pr{\texttt{GAME5}^\cA}$. In the next hybrid, we are going to remove completely the MPC computation, and the previous line that can now be merged in a single one:
  \begin {pcimage}
    {\normalfont\game[lnstart=11,linenumbering]{ $\texttt{GAME6}^\cA$ }{
        \hcancel[red]{\pcif \text{$\cA_2$ aborts} \pcthen \text{wait for $\tilde{c}$ from $\cA_2$. } \pcreturn \tilde{c} = c \ \pcfi}\\
        \hcancel[red]{(y,b) \gets \cA_3}\setcounter{pclinenumber}{11}\\
        \label{my:line:game6newA3}\hcancel[red]{(\tilde{c}, \vec{y}) \gets \Ideal_{\CombineRandom,\Sim}(\lambda,\vec{x},\rho_3)}\quad\textcolor{green!50!black}{\tilde{c} \gets \cA_3}
      }}
  \end{pcimage}
  The reason is that now, since the $\CombineRandom{}$ function does not depend on any secret own by honest parties, the input of honest parties can be replaced with wrong trapdoors $t_k^{(i)}$ and $\vect{d}_0[i]$. Therefore, since the adversary knows already $k^{(i)}$, it can simulate locally the ideal world. More precisely, from an adversary $\cA$ winning the game $\texttt{GAME5}^\cA$ with probability $p$, we can create another adversary $\cA'$ winning the game $\texttt{GAME6}^\cA$ with the same probability $p$: $\cA'$ will run $\cA_1$ and $\cA_2$ against the challenger, keeping locally the $k^{(i)}$. If $\cA_2$ aborts and sends $\tilde{c}$, then it returns $\tilde{c}$ directly. Otherwise $\cA'$ runs as a blackbox $\cA_3$ to obtain $(y,b)$, and then it locally runs $\Sim$ to compute $\Ideal_{\CombineRandom,\Sim}(\lambda,\vec{x},\rho_3)$, by feeding the input of honest parties input $\vec{x}$ with the $k^{(i)}$ that it got before, $y^{(i)} \eqdef y$, $b^{(i)} \eqdef b$, $r^{(i)} \sample \{0,1\}^n$ and since the values of $t_k^{(i)}$ do not matter anymore, it can put any value here. Finally $\cA'$ outputs the $\tilde{c}$ obtained from the simulation $\Ideal$. Therefore, since we do not change what is done, but who's doing what, we get:
  \begin{align}
    \max_{\QPT{} \cA} \pr{\texttt{GAME5}^\cA}(\lambda) = \max_{\QPT{} \cA'} \pr{\texttt{GAME6}^{\cA'}}(\lambda)
  \end{align}
  So now, the game $\texttt{GAME6}$ is exactly like $\texttt{GAME1}$ except that the lines~\ref{my:line:secureauth1A2aborts}--\ref{my:line:securauth1Guess} are replaced with a single line $\tilde{c} \gets \cA_3$. We will now do a similar strategy to remove the ZK protocol (line~\ref{my:line:secureauth1ProofZK}). The first step is to formalize this line (we will also merge it with the next line). We define $\cR_\lang$ as the relation $(\vect{d}_0[i], t_k^{(i)}, k^{(i)}) \in \cR_\lang(k^{(i)})$ iff $\CheckHonestTrapdoor_\lambda(\vect{d}_0[i], t_k^{(i)}, k^{(i)}) \land \Auth_i(\vect{d}_0[i],w_i) = 1$, and $\P{}$ the honest ZK prover associated with $\cR_\lang$. If we define $\rho_1$ as the internal state at the end of $\cA_1$, and $\V{}^*(\{k^{(i)}\}_{i \notin \cM},\rho_1) \eqdef (\rho_2 \gets \cA_2(\rho_1); \tilde{c} \gets \cA_3(\rho_2))$ then we can merge the line~\ref{my:line:game6newA3} of $\texttt{GAME6}$ with the line~\ref{my:line:secureauth1ProofZK} as follows:
  \begin{pcimage}
    {\normalfont\game[lnstart=7,linenumbering]{ $\texttt{GAME7}^\cA$ }{
        \textcolor{green!50!black}{\rho \gets }\cA_2(\{k^{(i)}\}_{i \notin \cM})\textcolor{green!50!black}{\pclinecomment{We just explicit the internal state after $\cA_2$}}\\
        \pcfor i \notin \cM \pcdo\\
        \t \hcancel[red]{\text{Prove in ZK with $\cA_{2,i}$ that $\CheckHonestTrapdoor_\lambda(\vect{d}_0[i], t_k^{(i)}, k^{(i)}) \land \Auth_i(\vect{d}_0[i],w_i) = 1$}}\pcskipln\\
        \t \textcolor{green!50!black}{\rho \gets \OUT_{\cA_{2,i}}\langle \P{}(\vect{d}_0[i], t_k^{(i)},w_i), \cA_{2,i}(\rho)\rangle(k^{(i)})}\\
      \pcendfor
      }}
  \end{pcimage}
  Since both games are exactly identical (up to the notation), we get $\pr{\texttt{GAME6}^\cA}=\pr{\texttt{GAME7}^\cA}$. Now, due to the fact that the MPC protocol respects the property \emph{Quantum Zero Knowledge} defined in \cref{def:postquantumZK}, there exist for all $i \notin \cM$ a simulator $\Sim_i$ fitting \cref{def:postquantumZK}. To be completely formal, one should define a series of games in which we replace in the loop only one $\OUT$ at a time by the simulated version, and we can then claim that the probability of having $\tilde{c} = c$ in each hybrid game is negligibly close to the probability of having $\tilde{c}=c$ in the first game, otherwise we could distinguish between the real world and the ideal world. This gives us at the end a new game $\texttt{GAME8}$:
  \begin{pcimage}
    {\normalfont\game[lnstart=9,linenumbering]{ $\texttt{GAME8}^\cA$ }{
        \hcancel[red]{\rho \gets \OUT_{\cA_{2,i}}\langle \P{}(\vect{d}_0[i], t_k^{(i)},w_i), \cA_{2,i}(\rho)\rangle(k^{(i)})}\quad\textcolor{green!50!black}{\rho \gets \Sim_i(k^{(i)},\V{}^*,\rho)}
      }}
  \end{pcimage}
  And using the above argument, $\pr{\texttt{GAME7}^\cA} \leq \pr{\texttt{GAME8}^\cA} + \negl[\lambda]$. Now, the simulators $\Sim_i$ can be fully simulated by the adversary since there is no more secret information. So, exactly like we did for the MPC computation, we can move the loop into the adversary:
  \begin{pcimage}
    {\normalfont\game[lnstart=7,linenumbering]{ $\texttt{GAME7}^\cA$ }{
        \hcancel[red]{\rho \gets \cA_2(\{k^{(i)}\}_{i \notin \cM})}\quad\textcolor{green!50!black}{\tilde{c} \gets \cA_2(\{k^{(i)}\}_{i \notin \cM})}\\
        \hcancel[red]{\pcfor i \notin \cM \pcdo}\\
        \t \hcancel[red]{\rho \gets \Sim(k^{(i)},\V{}^*,\rho)}\\
        \hcancel[red]{\pcendfor}\setcounter{pclinenumber}{11}\\
        \hcancel[red]{\tilde{c} \gets \cA_3}
      }}
  \end{pcimage}
  and we get:
  \begin{align}
    \max_{\QPT{} \cA} \pr{\texttt{GAME8}^\cA}(\lambda) = \max_{\QPT{} \cA'} \pr{\texttt{GAME9}^{\cA'}}(\lambda)
  \end{align}
  So now, $\texttt{GAME9}$ is like $\texttt{GAME1}$ except that all lines starting from line~\ref{my:line:secureauth1A2} are replaced with a single line $\tilde{c} \gets \cA_2(\{k^{(i)}\}_{i \notin \cM})$. We see now that the conditions line~\ref{my:line:securauth1TestAuthOk} and line~\ref{my:line:securauth1CorrptGoodSet} can only decrease the probability of winning the game. Therefore we can remove them (as well as $w_i$'s which are not used anymore). This gives us this new game (after removing empty lines):
  {\normalfont\game[linenumbering]{$\texttt{GAME10}^\adv (\lambda) $}{%
      (\cM, \vect{d}_0^{(0)}, \vect{d}_0^{(1)}) \gets \cA_1(1^\lambda)\\
      \pcif \exists i \in \cM, \vect{d}_0^{(0)}[i] \neq \vect{d}_0^{(1)}[i]\pcthen \pcreturn \pcfalse\ \pcfi\\
      c \sample \bin; \vect{d}_0 \eqdef \vect{d}_0^{(c)}\\
      \forall i \notin \cM, (k^{(i)}, t_k^{(i)}) \gets \GenLocal(1^\lambda, \vect{d}_0[i])\\
      \tilde{c} \gets \cA_2(\{k^{(i)}\}_{i \notin \cM})\\
      \pcreturn \tilde{c} = c
    }}
  Since we only increase the probability of success, we have:
  \begin{align}
    \max_{\QPT{} \cA} \pr{\texttt{GAME9}^\cA}(\lambda) \leq \max_{\QPT{} \cA'} \pr{\texttt{GAME10}^{\cA'}}(\lambda)
  \end{align}
  Similarly, we can decide to give more advices to the adversary, by running $\GenLocal$ on all $i \in [n]$ instead of only on the $i \notin \cM$:
  \begin{pcimage}
    {\normalfont\game[lnstart=3,linenumbering]{ $\texttt{GAME11}^\cA$ }{
        \hcancel[red]{\forall i \notin \cM}\textcolor{green!50!black}{\forall i \in [n]}, (k^{(i)}, t_k^{(i)}) \gets \GenLocal(1^\lambda, \vect{d}_0[i])\\
        \tilde{c} \gets \cA_2(\{k^{(i)}\}_{i \hcancel[red]{\notin \cM}\textcolor{green!50!black}{\in [n]}})
      }}
  \end{pcimage}
  Since we give more advices to the adversary, its probability of winning the game can only increase (it can always decide to drop this additional information). Therefore
  \begin{align}
    \max_{\QPT{} \cA} \pr{\texttt{GAME10}^\cA}(\lambda) \leq \max_{\QPT{} \cA'} \pr{\texttt{GAME11}^{\cA'}}(\lambda)
  \end{align}
  However, since $\Gen$ is defined as the concatenation of $\GenLocal$, then this game is actually exactly the $\INDFCT_{\Gen{}} ^\adv (\lambda)$ game defined \cref{game:indfct}, and by assumption, the probability of winning this game is smaller than $\frac{1}{2} + \negl[\lambda]$. Therefore, we also have:
  \begin{align}
    \max_{\QPT{} \cA} \pr{\texttt{GAME1}^\cA}(\lambda) \leq \frac{1}{2} + \negl[\lambda]
  \end{align}
  which ends the proof.
\end{proofE}


\section{Function construction}\label{sec:functionConstruction}
\pgfkeys{/prAtEnd/local custom defaults/.style={category=functionConstructions}}

\subsection{Construction of a \AssumpFct{} family}\label{sec:construction_f_superpoly}

In this section, we will explain how to derive a \AssumpFct{} family. See the \cref{subsec:quick_overview} to get an intuitive explanation of our method, which extends the construction introduced in~\cite{CCKW_2019_qfactory} (itself based on~\cite{CCKW18}). In the following, $\MPGen$ and $\MPDec$ are the functions defined in \cite{MP11} (to, respectively, generate a couple public key/trapdoor $(\vect{A},\vect{R})$ and to invert the function $g_{\vect{A}}(\vect{s}, \vect{e}) \eqdef \vect{A} \vect{s} + \vect{e}$). Details are in \cref{sec:MP11construction}.

\begin{definition}\label{def:ourConstruction}
  For the parameters $\cP \eqdef (k,N,\alpha,r_{max},n,\ccX)$ with $(k,N,n,r_{max}) \in \N$, $0<\alpha<1$, $\ccX \subseteq \Z_q^N \times \Z_q^{M+n}$ where $q \eqdef 2^k$ and $M \eqdef N(1+k)$, we define \cref{fig:fct_construction_def} the algorithms $\Gen_\cP$, $\Dec_\cP$, $\Eval_\cP$ (to compute $f_k$) and $h$. We use $\vect{d}_0 \in \{0,1\}^n \subseteq \Z_q^n$ (same for $\vect{d}$), $\vect{s}_0 \in \Z_q^{N}$, $\vect{e}_0 \in \Z_q^{M}$, $(\vect{s},\vect{e}) \in \ccX$, $c \in \{0,1\}$, $\vect{A} \in \Z_q^{(M+n) \times N}$ and $\vect{R}$ is the trapdoor obtained via the \cite{MP11} algorithm.

  \begin{figure}[htb]
    \centering
    \begin{pcvstack}[boxed, center,space=0.3cm]
      \begin{pchstack}[space=0.3cm]
        {\normalfont\procedure[linenumbering]{$\Gen_\cP(1^\lambda, \vect{d}_0)$}{
            (\vect{A}_u, \vect{R}) \gets \MPGen(1^\lambda)\\
            \vect{A}_l \sample \Z_q^n\\
            \vect{A} \eqdef \SmallBlockMatrix{\vect{A}_u}{\vect{A}_l}\\
            \vect{s}_0 \gets \disGaussAQ^N\\
            \vect{e}_0 \gets \disGaussAQ^{M+n}\\
            \vect{y}_0 \eqdef \vect{A} \vect{s}_0 + \vect{e}_0 + \frac{q}{2} \SmallBlockMatrix{\vect{0}^M}{\vect{d}_0}\\
            k \eqdef (\vect{A}, \vect{y}_0)\\
            t_k \eqdef (\vect{R}, \vect{d}_0, \vect{s}_0, \vect{e}_0, \vect{A})\\
            \pcreturn (k, t_k)
          }}
        {\normalfont\procedure[linenumbering]{$\Dec_\cP(t_k \eqdef (\vect{R},\vect{d}_0,\vect{s}_0,\vect{e}_0,\vect{A}), \vect{y})$}{
            \SmallBlockMatrix{\vect{y}_u \in \Z_q^M}{\vect{y}_l \in \Z_q^n} \eqdef \vect{y}; \quad \SmallBlockMatrix{\vect{A}_u}{\vect{A}_l} \eqdef \vect{A}\\
            (\vect{s},\vect{e}_u) \gets \MPDec(\vect{R}, \vect{A}, \vect{y}_u)\\
            \pcif \vect{s} = \bot \pcthen \pcreturn \bot \ \pcfi\\
            \vect{d} \eqdef \RoundMod_q\left(\vect{y}_l - \vect{A}_l \vect{s}\right)\\
            \vect{e} \eqdef \SmallBlockMatrix{\vect{e}_u}{\vect{y}_l - \vect{A}_l \vect{s} - \vect{d}}\\
            \vect{s}' \eqdef \vect{s} - \vect{s}_0;\quad \vect{e}' \eqdef \vect{e} - \vect{s}_0\\
            \pcif (\vect{s},\vect{e}) \notin \ccX \text{ or } (\vect{s}',\vect{e}') \notin \ccX \pcthen\\
            \t \pcreturn \bot \  \pcfi\\
            \pcreturn ((\vect{s},\vect{e},0,\vect{d}), (\vect{s}',\vect{e}',1,\vect{d}'\xor \vect{d}_0))
          }}
      \end{pchstack}
      \begin{pchstack}[space=0.3cm]
        {\normalfont\procedure[linenumbering]{$\Eval_\cP(k \eqdef (\vect{A}, \vect{y}_0), x \eqdef
            (\vect{s},\vect{e},c,\vect{d})) \eqqcolon f_k(x)$}{
            \pcreturn \vect{A} \vect{s} + \vect{e} + \SmallBlockMatrix{\vect{0}^M}{\vect{d}} + c \times \vect{y_0}
          }}
        {\normalfont\procedure[linenumbering]{$h(x \eqdef (\vect{s},\vect{e},c,\vect{d}))$}{
            \pcreturn \vect{d}
          }}
      \end{pchstack}
    \end{pcvstack}
    \caption{Definition of the \AssumpFct{} family}
    \label{fig:fct_construction_def}
\end{figure}
\end{definition}

While most of the correctness and security proofs are similar to \cite{CCKW_2019_qfactory}, we also need a worst-case analysis (for malicious applicants) and ensure that the functions are $\negl[\lambda]$-$2$-to-$1$ for appropriate parameters.

\inAppendixIfPublished{
  We obtain conditions to check that our function is $\delta$-$2$-to-$1$ and have the \XOR{} property explained in \cref{def:fct_requirements}, which is exactly what the $\HCheckHonestTrapdoor$ function (required later in \cref{thm:compiler}) must do (i.e.\ check that the trapdoor is well formed, and that it has the properties given in \cref{lem:conditionsFkDelta2To1}). Note that the $\delta$ given here is a kind of worst-case analysis (which is required if the key is maliciously sampled): on average, if we sample random functions we expect to have a smaller $\delta$.
  \begin{lemmaE}[Conditions for $f_k$ to be $\delta$-$2$-to-$1$]\label{lem:conditionsFkDelta2To1}
    Let $\cP$ be like in \cref{def:ourConstruction}. We define $\rsafe = r_{\max} - \alpha q \sqrt{N+M+n}$, $\ccX+(\hat{\vect{s}}_0,\hat{\vect{e}}_0) \eqdef \{ (\vect{s}+\hat{\vect{s}}_0, \vect{e}+\hat{\vect{e}}_0) \mid (\vect{s},\vect{e}) \in \ccX \}$ and
    \begin{align}
      \delta \eqdef1- \min &\left\{\frac{|\ccX \cap (\ccX + (\hat{\vect{s}}_0,\hat{\vect{e}}_0))|}{|\ccX|}\right. \nonumber\\
                           &\left.\bigm| (\hat{\vect{s}}_0, \hat{\vect{e}}_0) \in \Z_q^N \times \Z_q^{M+n}, \left\|\SmallBlockMatrix{\hat{\vect{s}}_0}{\hat{\vect{e}}_0}\right\|_2 \leq \alpha q \sqrt{N+M+n} \right\}\label{eq:defDeltaInter}
    \end{align}
    Let $\vect{s}_0 \in \Z_q^N$, $\vect{e}_0 \in \Z_q^{M+n}$, $\vect{d}_0 \in \{0,1\}$, $\hat{\vect{A}} \in \Z_q^{N \times N}$, $\vect{A}_l \in \Z_q^{n \times n}$ and $\vect{R} = \HorizBlockMatrix{\vect{R}_1}{\vect{R}_2} \in \Z_q^{Nk \times 2N}$. We define as before:
    \begin{align}
      \vect{A_u} \eqdef \SmallBlockMatrix{\hat{\vect{A}}}{\vect{G} - \vect{R}_2\hat{\vect{A}} - \vect{R}_1}
      &\qquad \vect{A} \eqdef \SmallBlockMatrix{\vect{A}_u}{\vect{A}_l} \in \Z_q^{(M+n)\times N}
    \end{align}
    together with $\vect{y}_0 \eqdef \vect{A}\vect{s}_0+\vect{e}_0+\frac{q}{2} \SmallBlockMatrix{\vect{0}^M}{\vect{d}_0}$ and $k \eqdef (\vect{A}, \vect{y}_0)$.

    If $\sqrt{\sigma_{\max}(\vect{R})^2 +1} \leq \frac{q}{4 r_{\max}}$, $\left\|\SmallBlockMatrix{\vect{s}_0}{\vect{e}_0}\right\|_2 \leq \alpha q \sqrt{N+M+n}$, and
    \begin{align}
      \ccX \subseteq \left\{(\vect{s}, \vect{e}) \in \Z_q^N \times \Z_q^{M+n} \middle| \left\|\SmallBlockMatrix{\vect{s}}{\vect{e}}\right\|_2 \leq \rsafe \right\}\label{eq:X2rsafe}
    \end{align}
    then the function $f_{k}(x)$ described in \cref{def:ourConstruction} is \mbox{$\delta$-$2$-to-$1$}, trapdoor, and for any $y$ having exactly two preimages $x$ and $x'$, we have $x \neq x'$, $h(x) \xor h(x') = \vect{d}_0$.

    On the other hand, if $(k, t_k)$ is sampled according to $\Gen(1^\lambda, \vect{d}_0)$, if \cref{eq:X2rsafe} is true, and if:
    \begin{align}
      \sqrt{\left(C \times \alpha q \times \sqrt{N}(\sqrt{k}+\sqrt{2}+1)\right)^2 + 1} \leq \frac{q}{4 r_{\max}}\label{eq:linkWithRmax2}
    \end{align}
    then with overwhelming probability on $N$, the function $f_k$ is $\delta$-$2$-to-$1$, trapdoor, and for any $y$ having exactly two preimages $x$ and $x'$, we have $x \neq x'$, $h(x) \xor h(x') = \vect{d}_0$.
  \end{lemmaE}
  \begin{proofE}
    Let us first prove that for all $c \in \{0,1\}$, the function $f_k(\cdot,\cdot,c,\cdot)$ is injective. Let $c \in \{0,1\}$, and $\vect{s},\vect{e},\vect{d},\vect{s}',\vect{e}',\vect{d}'$ be such that $f_k(\vect{s},\vect{e},c,\vect{d}) = f_k(\vect{s}',\vect{e}',c,\vect{d}')$. Then, if we consider only the upper part of this equation (and denote $\vect{e}_u$ the upper part of the vector $\vect{e}$), we get $\vect{A}_u \vect{s} + \vect{e}_u + c \times \vect{y}_{0,l} = \vect{A}_u \vect{s}' + \vect{e}'_u + c \times \vect{y}_{0,l}$, i.e.\ $\vect{A}_u \vect{s} + \vect{e}_u = \vect{A}_u \vect{s}' + \vect{e}'_u$. But because $\sqrt{\sigma_{\max}(\vect{R})^2 +1} \leq \frac{q}{4 r_{\max}}$ and due to the condition on $\ccX$ given in \cref{eq:X2rsafe} and the fact that $\rsafe \leq r_{\max}$, according to \cref{lem:MPInjectiveSingularValue} the function $(\vect{s},\vect{e}) \rightarrow \vect{A}_u \vect{s} + \vect{e}$ is injective. So $\vect{s} = \vect{s}'$ and $\vect{e}_u = \vect{e}'_u$. Now, we focus on the upper part of the above equation: we have $\vect{A}_l \vect{s} + \vect{e}_l + \frac{q}{2} \vect{\vect{d}}+ c \times \vect{y}_{0,l} = \vect{A}_l \vect{s'} + \vect{e}'_l + \frac{q}{2} \vect{\vect{d}'}+ c \times \vect{y}_{0,l}$. Because $\vect{s}=\vect{s}'$, we obtain $\vect{e}_l + \frac{q}{2} \vect{\vect{d}} = \vect{e}'_l + \frac{q}{2} \vect{\vect{d}'}$. Because $1 \leq \sqrt{\sigma_{\max}(\vect{R})+1} < \frac{q}{4 r_{\max}}$, we have $r_{\max} < \frac{q}{4}$. Therefore, we get for all $i$, $|\vect{e}_l[i]| < \frac{q}{4}$, so $\RoundMod_q(\vect{e}_l[i] + \frac{q}{2} \vect{d}[i]) = \RoundMod_q(\vect{e}_l'[i] + \frac{q}{2} \vect{d}'[i])$, i.e.\ $\vect{d}[i] = \vect{d}'[i]$. So $\vect{d} = \vect{d}$ and therefore we also get $\vect{e}_l=\vect{e}'_l$: the function $f_k(\cdot, \cdot, c, \cdot)$ is injective.

    Therefore, $f_k$ has at most two preimages, one for $c =0$ and one for $c=1$. We prove now that $f_k(\vect{s},\vect{e},0,\vect{d}) = f_k(\vect{s}',\vect{e}',1,\vect{d}')$, iff $(\vect{s}', \vect{e}',\vect{d}') = (\vect{s}-\vect{s}_0,\vect{e}-\vect{e}_0,\vect{d} \xor \vect{d}_0)$. One implication is trivial: if $(\vect{s}', \vect{e}',\vect{d}') = (\vect{s}-\vect{s}_0,\vect{e}-\vect{e}_0,\vect{d} \xor \vect{d}_0)$ then because $q$ is even, $f_k(\vect{s},\vect{e},0,\vect{d}) = f_k(\vect{s}',\vect{e}',1,\vect{d}')$. We prove now the second implication. By definition of $f_k$, if we consider again the upper part of the equation and replace $y_0$ by its definition, we have $\vect{A}_u \vect{s}+\vect{e}_u = \vect{A}_u( \vect{s}'+\vect{s}_0) + (\vect{e}_u'+\vect{e}_{0,u})$. But the triangle inequality gives:
    \begin{align*}
      \left\|\SmallBlockMatrix{\vect{s}'+\vect{s}_0}{\vect{e}_u'+\vect{e}_{0,u}}\right\|_2 \leq \left\|\SmallBlockMatrix{\vect{s}'}{\vect{e}_u'}\right\|_2+\left\|\SmallBlockMatrix{\vect{s}_0}{\vect{e}_{0,u}}\right\|_2 \leq r_{safe} + \alpha q \sqrt{N+M+n} = r_{\max}
    \end{align*}
    Therefore, we can use again the injectivity property given in \cref{lem:MPInjectiveSingularValue}, which gives $(\vect{s},\vect{e}_u)=(\vect{s}'+\vect{s}_0,\vect{e}'+\vect{e}_{0,u})$. We can now analyse the lower part of the equation: $\vect{A}_l \vect{s}+\vect{e}_l + \frac{q}{2} \vect{d} = \vect{A}_l (\vect{s}'+\vect{s}_0)+(\vect{e}'_l+\vect{e}_{0,l}) + \frac{q}{2} (\vect{d}' + \vect{d}_0)$. Because $\vect{s}=\vect{s}'+\vect{s}_0$ and $q$ is even, we have $\vect{e}_l + \frac{q}{2} \vect{d} = (\vect{e}'_l+\vect{e}_{0,l}) + \frac{q}{2} (\vect{d}'\xor \vect{d}_0)$. Using again the triangle inequality, we prove the same way that $\|\vect{e}'_l+\vect{e}_{0,l}\|_2 < \frac{q}{2}$. As before, by rounding the previous equation using $\RoundMod_q$, we obtain $\vect{d}=\vect{d}'\xor \vect{d}_0$ and $\vect{e}_l = \vect{e}'_l$ which concludes the proof. In particular, if $x$ and $x'$ are the two preimage, we have $h(x) \xor h(x')= \vect{d}\xor \vect{d}' = \vect{d}_0$. We remark that this proof follows exactly the algorithm $\Dec$, therefore the correctness of $\Dec$ follows quite directly and thus the function is trapdoor.

    Now, we prove that the function $f_k$ is $\delta$-$2$-to-$1$. The total number of elements in the domain of $f_k$ is $2|\ccX| \times 2^n$. Let $(\vect{s},\vect{e}) \in \ccX$ and $\vect{d} \in \{0,1\}^n$. Then, using the result proven above, $f_k(\vect{s},\vect{e},0,\vect{d})$ has a second preimages $(\vect{s}-\vect{s}_0, \vect{e}-\vect{e}_0,1,\vect{d} \xor \vect{d}_0)$ iff $(\vect{s}-\vect{s}_0,\vect{e}-\vect{e}_0) \in \ccX$, i.e.\ iff $(\vect{s},\vect{e}) \in \ccX + (\vect{s}_0,\vect{e}_0)$. So the number of elements $x$ such that $|f_k^{-1}(f_k(x))|=2$ is equal to $2 |\ccX \cap (\ccX + (\vect{s}_0,\vect{e}_0))| 2^n$, therefore if we define:
    \begin{align}
      \delta_k \eqdef 1-\frac{2 |\ccX \cap (\ccX + (\vect{s}_0,\vect{e}_0))| 2^n}{2 |\ccX| 2^n} = 1-\frac{|\ccX \cap (\ccX + (\vect{s}_0,\vect{e}_0))|}{|\ccX|}
    \end{align}
    this function is $\delta_k$-$2$-to-$1$. But $\left\|\SmallBlockMatrix{\vect{s}_0}{\vect{e}_0}\right\|_2 \leq \alpha q \sqrt{N + M + n}$, so by definition of $\delta$, $\delta_k \leq \delta$. So $f_k$ is also $\delta$-$2$-to-$1$.

    To prove the last part of the theorem, we use \cref{eq:linkWithRmax2} and \cref{lem:MPInjectiveSingularValue}: with overwhelming probability (on $N$), $\sqrt{\sigma_{\max}(\vect{R})^2 +1 } < \frac{q}{4 r_{\max}}$. Moreover, because $\SmallBlockMatrix{\vect{s}_0}{\vect{e}_0}$ is sampled according to $\disGauss{\alpha}{q}^{N+(M+n)}$, we get, using\footnote{Note that the original lemma applies to Gaussian distributions that are not reduced modulo $q$, but reducing the Gaussian distribution modulo $q$ can only decrease the length of the vector.} \cref{lem:boundGaussianDistrib}, that with overwhelming probability (on $N+M+n$): $\left\|\SmallBlockMatrix{\vect{s}_0}{\vect{e}_0}\right\|_2 \leq \alpha q \sqrt{N+M+n}$. We can now end the proof using the first (already proven) part of the theorem.
  \end{proofE}

  Note that we did not yet give an explicit definition of $\ccX$. The most natural way to define $\ccX$ may be to define it following \cref{eq:X2rsafe} as:
  \begin{align}
    \ccXSphere \eqdef \left\{(\vect{s}, \vect{e}) \in \Z_q^N \times \Z_q^{M+n} \middle| \left\|\SmallBlockMatrix{\vect{s}}{\vect{e}}\right\|_2 \leq \rsafe \right\}
  \end{align}
  However, for our protocol to work, one needs to be able to create quantumly a uniform superposition over all elements in $\ccX$. A first naive method would be a \emph{rejection sampling} method: we create a uniform superposition $\sum_x \ket{x}$ over the hypercube of length $2r_{safe}$ (see later how to do), add an auxiliary qubit, set (in superposition) this qubit to $1$ if $\|x\|_2 \leq r_{safe}$, and to $0$ otherwise, and we finally measure it. If the output is $0$, we discard the state and start again from scratch, otherwise we have the wanted state. Unfortunately, this method is inefficient as the probability of not rejecting is negligible when $N$ tends to the infinity. While some more efficient methods may exist, it will be easier to focus rather on $\cX$ being an hypercube.
  \begin{definition} For any $(\mu,N,M,n) \in \N^4$, we define the hypercube $\ccXCube[\mu]$ as:
    \begin{align}
      \ccXCube[\mu] \eqdef \left\{(\vect{s}, \vect{e}) \in \Z_q^N \times \Z_q^{M+n} \middle| \left\|\SmallBlockMatrix{\vect{s}}{\vect{e}}\right\|_\infty \leq \mu \right\}
    \end{align}
  \end{definition}

  \begin{remark}\label{rq:sampleXCube}
    It is now easy to sample from $\ccXCube[\mu]$: we can do a rejection sampling as explained above, except that we proceed coordinate per coordinate (this is much more efficient than doing a rejection sampling on the final high dimensional state): if we use the binary two's complement notation, we apply Hadamard gates on $\ceil{\log_2(2\mu+1)}$ qubits, and use an auxiliary qubit to check if the state is projected on the superposition of elements having size $\leq 2\mu+1$ (this should happen with probability $\frac{2\mu+1}{2^{\ceil{\log_2(2\mu+1)}}} \geq 1/2$). If the test passes, we add $\ket{0}$ ``significants qubits'' until having $ k = \log_2(q)$ qubits, and run the quantum unitary that substracts $\mu$ modulo $q$. We repeat until having $N+M+n$ successful projections (this require therefore $O(N+M+n)$ samplings). Moreover, we can even get completely rid of the rejection sampling if we slightly change the definition of $\ccXCube[\mu]$ and if we make it less symmetric by asking that there exists $k' \in \N$ such that for all $i$, $\SmallBlockMatrix{\vect{s}}{\vect{e}}[i] \in \left[-2^{k'},2^{k'}-1\right]$. The superposition procedure is the same except that we work on $k'+1$ qubits, and we don't need the rejection sampling. However this notation slightly complicates the computations with no clear benefit (if simplifies slightly the sampling part, but it may decrease the value of $\delta$ since the term $2^{k'}$ must be a power of $2$), so for simplicity we will keep our initial notation. Note also that in the following, for the sake of simplicity we won't try to give tight bounds.
  \end{remark}

  \begin{lemmaE}[]\label{lem:upperBoundDelta}
    Let $(N,M,n,\mu) \in \N^4$, $\cX = \ccXCube[\mu]$, $\alpha \in (0,1)$, $\delta$ be as in \cref{eq:defDeltaInter} and $\mu' \eqdef \floor*{\mu  - \alpha q \sqrt{N+M+n}}$. Then if $\mu' \geq 0$:
    \begin{align}
      \delta \leq 1-\left(\frac{2\mu'+1}{2\mu+1}\right)^{N+M+n} \leq \frac{(\alpha q + 1)(N+M+n)^{3/2}}{\mu+1/2}
    \end{align}
  \end{lemmaE}
  \begin{proofE}
    Let $ (\vect{s},\vect{e}) \in \ccXCube[\mu]$ and $(\hat{\vect{s}}_0,\hat{\vect{e}}_0) \in \Z_q^N \times \Z_q^{M+n}$ such that $\left\|\SmallBlockMatrix{\hat{\vect{s}}_0}{\hat{\vect{e}}_0}\right\|_2 \leq \alpha q \sqrt{N+M+n}$. Then, $(\vect{s},\vect{e})\in \ccXCube[\mu] + (\hat{\vect{s}}_0,\hat{\vect{e}}_0)$ iff $(\vect{s}-\hat{\vect{s}}_0,\vect{e}-\hat{\vect{e}}_0)\in \cX$, i.e.\ iff $\left\|\SmallBlockMatrix{\vect{s}-\hat{\vect{s}}_0}{\vect{e}-\hat{\vect{e}}_0}\right\|_\infty \leq \mu$.
    If we assume that $\left\|\SmallBlockMatrix{\vect{s}}{\vect{e}}\right\|_\infty \leq \mu'$ (this will not be super tight, but it is good enough for our analysis), then
    \begin{align}
      \left\|\SmallBlockMatrix{\vect{s}-\hat{\vect{s}}_0}{\vect{e}-\hat{\vect{e}}_0}\right\|_\infty
      \leq \left\|\SmallBlockMatrix{\vect{s}}{\vect{e}}\right\|_\infty + \left\|\SmallBlockMatrix{\hat{\vect{s}}_0}{\hat{\vect{e}}_0}\right\|_\infty
      \leq \left\|\SmallBlockMatrix{\vect{s}}{\vect{e}}\right\|_\infty + \left\|\SmallBlockMatrix{\hat{\vect{s}}_0}{\hat{\vect{e}}_0}\right\|_2 \leq \mu
    \end{align}
    Therefore, $\ccXCube[\mu'] \subseteq \ccXCube[\mu] \cap (\ccXCube[\mu] + (\hat{\vect{s}}_0,\hat{\vect{e}}_0))$. So $|\ccXCube[\mu] \cap (\ccXCube[\mu] + (\hat{\vect{s}}_0,\hat{\vect{e}}_0))| \geq |\ccXCube[\mu']| = (2\mu'+1)^{N+M+n}$, and because $|\ccXCube[\mu] = (2\mu+1)^{N+M+n}|$, we get:
    \begingroup
    \allowdisplaybreaks
    \begin{align}
      \delta
      &\leq 1-\left(\frac{2\mu'+1}{2\mu+1}\right)^{N+M+n} = 1-\left(1+\left(-1+\frac{2\mu'+1}{2\mu+1}\right)\right)^{N+M+n}\nonumber\\
      &\leq 1-\left(1+(N+M+n)\left(-1+\frac{2\mu'+1}{2\mu+1}\right)\right)\tag{Bernouilli's inequality}\\
      &= (N+M+n)\left(1-\frac{2\mu'+1}{2\mu+1}\right)\nonumber\\
      & \leq (N+M+n)\left(1-\frac{2\mu-2\alpha q \sqrt{N+M+n}+3}{2\mu+1}\right) \tag{Remove $\floor{\cdot}$}\\
      & = (N+M+n)\left(\frac{\alpha q \sqrt{N+M+n}+1}{\mu+1/2}\right)\nonumber\\
      & \leq \frac{(\alpha q + 1)(N+M+n)^{3/2}}{\mu+1/2}
    \end{align}
    \endgroup
  \end{proofE}

  \begin{lemmaE}[Conditions on parameters]\label{th:conditionsParam}
    Let $\lambda \in \N$ be a security parameter and let $(k,n) \in \N$ and $\alpha \in (0,1)$ be parameters that depend on $\lambda$, and $C \approx \frac{1}{\sqrt{2\pi}}$ (see \cref{lem:MPInjectiveSingularValue}). We define $N \eqdef \lambda$, $q=2^k$, $M \eqdef N(1+k)$,
    \begingroup
    \allowdisplaybreaks
    \begin{align}
      r_{\max} &\eqdef \frac{q}{4\sqrt{\left(C \times \alpha q \times \sqrt{N}(\sqrt{k}+\sqrt{2}+1)\right)^2 + 1}}\\
      \rsafe &\eqdef r_{\max} - \alpha q \sqrt{N+M+n}\\
      \mu &\eqdef \floor*{\frac{\rsafe}{\sqrt{N+M+n}}}\\
      \ccX &\eqdef \ccXCube[\mu]\\
      \delta_m &\eqdef \frac{(\alpha q + 1)(N+M+n)^{3/2}}{\mu+1/2}
    \end{align}
    \endgroup
    Then, if $\floor*{\mu  - \alpha q \sqrt{N+M+n}} \geq 0$, the construction given in \cref{def:ourConstruction} is $\delta_m-\AssumpFctNoDelta$ (see \cref{def:fct_requirements}) assuming the security of decision-$\LWE{}_{q,\disGaussAQ}$. Moreover, if we define:
    \begingroup
    \allowdisplaybreaks
    \begin{align}
      \alpha_0 &\eqdef \frac{1}{q} \sqrt{(\alpha q)^2 - \omega(\sqrt{\log N})^2}\\
      \gamma &\eqdef \tilde{O}\left(\frac{N}{\alpha_0}\right) \tag{See constants in \cite{PRS17_PseudorandomnessRingLWEAny}}
    \end{align}
    \endgroup
    and if $\alpha_0 q > 2\sqrt{N}$, then the construction is secure if $\GapSVP_\gamma$ is hard. In particular, we are interesting in the regime in which $\delta_m$ is negligible (correctness) and in which $\gamma = \tilde{O}(2^{N^\eps})$ for some $\eps \in (0,1/2)$ (security).
  \end{lemmaE}
  \begin{proofE}
    For the first part of the theorem, the efficient generation and computation properties are trivial to check. The fact that the function is trapdoor and the property on the \XOR{} is a direct consequence of \cref{lem:conditionsFkDelta2To1}. The $\delta_m$-$2$-to-$1$ property comes from \cref{lem:conditionsFkDelta2To1} and \cref{lem:upperBoundDelta}. The method to efficiently create a uniform superposition of elements in $\ccX$ is given in \cref{rq:sampleXCube}.

    To prove the indistinguishability property on game $\INDFCT$, we assume that there exists an adversary $\cA$ that can win this game with a non-negligible advantage. Because $\cA$ has only access to $\left(\vect{A}, \vect{y}_0 = \vect{A} \vect{s}_0 + \vect{e}_0 + \frac{q}{2} \SmallBlockMatrix{\vect{0}^M}{\vect{d}_0^{(c)}}\right)$, we can use $\cA$ to break the decision-\LWE{} problem: given a challenge $(\vect{A}', \vect{y})$, we run the adversary and send to $\cA$ the couple $(\vect{A}', \vect{y}+ \frac{q}{2} \SmallBlockMatrix{\vect{0}^M}{\vect{d}_0^{(c)}})$. If the guess $\tilde{c}$ of $\cA$ equals $c$, we guess that we get the non-uniform distribution normal-$A_{\vect{s},\cD_{\Z,\alpha q}}$ (i.e.\ the distribution where $\vect{s}$ is also sampled according to $\cD_{\Z,\alpha q}$), otherwise we guess that we get the uniform distribution $U$. We remark that if $\vect{y}$ is a vector chosen uniformly at random, then the distribution of $\vect{y}+ \frac{q}{2} \SmallBlockMatrix{\vect{0}^M}{\vect{d}_0^{(c)}}$ is statistically uniform. So in that case, $\cA$ cannot guess $c$ with probability better than $1/2$. Now, if $y$ is sampled from normal-$A_{\vect{s},\cD_{\Z,\alpha q}}$, then $\cA$ must guess correctly the value of $c$ with non-negligible advantage, otherwise it means that we can distinguish a matrix obtained by $\MPGen$ from a uniform matrix (we already know it is not possible, see \cref{lem:MPGenIndistinguishableMatrix}). Therefore, the probability $p$ of guessing the correct distribution is:
    \begin{align}
      p &\geq \frac{1}{2} \times \frac{1}{2} + \frac{1}{2} \left(\frac{1}{2} + \negl[\lambda]\right) = \frac{1}{2} + \negl[\lambda]
    \end{align}
    which is absurd since we assumed the hardness of \LWE{}. Note that we don't exactly have an instance of decision-$\LWE$, because it is an instance of the normal version of decision-$\LWE$ ($\vect{s}_0$ is sampled according to the same distribution as $\vect{e}_0$). However, \cref{lem:normalLWEreduction} shows that the normal problem is harder, and keeping only $K$ samples can also only make the problem harder. Therefore, no adversary can win $\INDFCT$ with non-negligible advantage.

    For the second part, if we assume $\GapSVP_\gamma$ to be hard, then using \cref{lem:GapSVPtoLWE} we get that decision-$\LWE_{q, \cD_{\alpha_0 q}}$ is also hard ($\alpha_0 q > 2\sqrt{N}$, and $0 < \alpha_0 \leq \alpha < 1$). We can now discretize the distribution using \cref{cor:continuousToDiscrete} ($\lambda = N$) to obtain that decision-$\LWE_{q, \cD_{\Z, \alpha q}}$ is hard. Indeed, if we assume the existence of an adversary $\cA$ that can distinguish with non-negligible advantage the distribution $U$ from $A_{\vect{s},\cD_{\Z, \alpha q}}$ for a vector $\vect{s}$ chosen uniformly at random, then $\cA$ has also a non-negligible advantage in distinguishing $\cA_{\vect{s},\chi}$ where $\chi$ is the marginal distribution of $\vect{e}$ in \cref{cor:continuousToDiscrete} (otherwise we could use $\cA$ to distinguish $\chi$ from $\cD_{\Z, \alpha q}$, which is impossible because they are statistically negligibly close). But it also means that $\cA$ can also be used to break $\LWE_{q,\cD_{\alpha_0 q}}$ by first discretizing $\cD_{\alpha_0 q}$ (it works because the transformation given in \cref{cor:continuousToDiscrete} also maps the uniform distribution on itself). Since we already prove the security when we assume the hardness of decision-$\LWE_{q, \cD_{\Z, \alpha q}}$ above, the proof is finished.
  \end{proofE}
}

\publishedVsArxiv{
  \begin{theoremE}[{Existence of a \mbox{$\negl[\lambda]$-\AssumpFctNoDelta{}} family}]
    Let $\eps \in (0,\frac{1}{2})$ be a constant, and $\lambda \in \N$ be a security parameter. Let $n = \poly(\lambda) \in \N$ and $N \eqdef \lambda$. If we assume the hardness of the $\GapSVP_\gamma$ problem for any $\gamma = \tilde{O}(2^{N^\eps})$, then there exists a $\negl[\lambda]$-\AssumpFctNoDelta{} family of functions.
  \end{theoremE}
  The precise instantiation of the parameters and the proof require a more in depth analysis and can be found in \cref{appendix:fctConstruction} (notably in \cref{th:conditionsParam,thm:ExistanceAssumpFctNoDeltaAdvanced}). We show that it is possible to find a regime in which $\alpha$ is small enough to provide correctness, and big enough to provide security.
}{}

\inAppendixIfPublished{
  We prove now that there exists an instantiation that fulfills the requirements of \cref{th:conditionsParam}. Note that for simplicity, we only verify the properties asymptotically. Moreover, we do not attempt to give any particularly optimized construction (note that there is a tradeoff between security, correctness, and simplicity of the quantum superposition preparation circuit).
  \begin{theoremE}[{Existence of a \mbox{$\negl[\lambda]$-\AssumpFctNoDelta{}} family}]\label{thm:ExistanceAssumpFctNoDeltaAdvanced}
    Let $\eps \in (0,\frac{1}{2})$ be a constant, and $\lambda \in \N$ be a security parameter. Let $n = \poly[\lambda] \in \N$ and $N \eqdef \lambda$. If we assume the hardness of the $\GapSVP_\gamma$ problem for any $\gamma = \tilde{O}(2^{N^\eps})$, then there exists a $\negl[\lambda]$-\AssumpFctNoDelta{} family of functions. More precisely, if we define the fixed function $\omega(\sqrt{\log N}) \eqdef \log N$, $k \eqdef \floor{N^\eps}$, $q \eqdef 2^k$,
    \begin{align}
      \alpha \eqdef \frac{\sqrt{4N+\omega(\sqrt{\log N})^2+1}}{q}\label{eq:alphaOmegaN}
    \end{align}
    and $M, r_{\max}, \rsafe, \mu, \ccX, \delta_m$ as in \cref{th:conditionsParam}, the construction given in \cref{def:ourConstruction} is $\delta_m$-\AssumpFctNoDelta{} (for sufficiently large\footnote{The function may not be well defined for too small $\lambda$ because the input set may be empty. For example, with this instantiation, if we take $\eps = \frac{1}{3}$, the function is well defined for $N \geq 7 \times 10^5$. Moreover, when $N = 6 \times 10^6$, we get $k = 181$ and $\delta_m < 2^{-80}$. There is surely place for optimisation, but only existence matters here.} $\lambda$), with $\delta_m = \negl[\lambda]$.
  \end{theoremE}
  \begin{proofE}
    We just need to check that for sufficiently large $\lambda$ the properties of \cref{th:conditionsParam} are respected. Because $1/q = \negl[N]$, it is easy to see that for sufficiently large $\lambda$, $a \in (0,1)$ since $\alpha = \poly[N]/q$. Then, $\alpha q > \sqrt{4N + \omega(\sqrt{\log N})^2}$, so using notation from \cref{th:conditionsParam}, we directly get $\alpha_0 q > 2\sqrt{N}$. Moreover, multiplying \cref{eq:alphaOmegaN} by $q$, we get $\alpha q = \poly[N]$, so it means that $\alpha_0 = \poly[N]/q = \negl[\lambda]$. Therefore $\gamma = \tilde{O}(N/\alpha_0) = \tilde{O}(\sqrt{N}q) = \tilde{O}(\sqrt{N}2^{N^\eps}) = \tilde{O}(2^{N^\eps})$ which is assumed to be hard. Next, let us study $\mu$ and $\delta_m$. Because $\alpha q = \poly[n]$ and $q = 1/\negl[\lambda]$ is superpolynomial, that $r_{\max}$ is also superpolynomial, and same for $r_{safe}$, $\mu$ and $\floor{\mu-\alpha q \sqrt{N+M+N}}$ (we only substract or divide by terms that are $\poly[N]$). Therefore, for a large enough $\lambda$ ($=N$), we have $\floor{\mu-\alpha q \sqrt{N+M+N}} \geq 0$. Finally, $\delta_m = \poly[N]/(\mu + 1/2)$. But we showed that $\mu$ is superpolynomial, so $\frac{1}{\mu}$ is negligible, and therefore $\delta_m$ is also negligible, which ends the proof.
  \end{proofE}
}

\subsection{Generic construction to create distributable \AssumpFctCanPrime{} primitives from \AssumpFct{} primitives}\label{sec:compilerAssumpFctCan}

We prove in this section that we can create a distributable \AssumpFctCanPrime{} family of functions from a \AssumpFct{} family having a small assumption.

\begin{theoremE}[][end if published]\label{thm:compiler}
  Let $\delta \in [0,1]$, and $\{f_k\}_{k \in \cK}$ be a \AssumpFct{} family\footnote{In fact we only require this function to work when $\vect{d}_0$ is a single bit.} of functions, such that $f_k(x)$ can be written as $f_k((c,\bar{x}))$ with $c \in \{0,1\}$ a bit labelling the preimage\footnote{This is quite similar to the concept of claw-free functions used in \cite{mahadev2017classical}.}, i.e.\ such that when a given $y$ has exactly two preimages, one preimage has the form $(0,\bar{x})$ and the other $(1,\bar{x}')$. Then there exists a family $\{f_k'\}_{k \in \cK'}$ which is a distributable \AssumpFctCanPrime{} family with $\delta'=1-(1-\delta)^n \leq \delta n$. In particular, if $\delta$ is negligible (and $n$ polynomial) then $\delta'$ is negligible. Moreover, if the family $\{f_k\}$ admits a circuit $\HCheckHonestTrapdoor_\lambda(\vect{d}_0, t_k, k)$ that returns $1$ iff $t_k$ is the trapdoor of $k$, $k \in \cK$ and $\vect{d}_0 = \vect{d}_0(t_k)$, then there exists a function $\CheckHonestTrapdoor$ for $\{f_k'\}$ having the properties from \cref{def:GHZdistCapable}.

  The family $\{f'_k\}_{k \in \cK'}$ can be obtained by generating $n$ independent functions in $\{f_k\}$ (one for each $\vect{d}_0[i]$). \publishedVsArxiv{Here is}{More precisely, we define in \cref{fig:constructionFromSimpleGHZ}} the precise construction (where $\HGen, \HEnc, \HDec$ and $\HEval$ are coming from the family $\{f_k\}$)\publishedVsArxiv{:}{.}
  \begin{figure}[htb]
    \centering
    \begin{pchstack}[boxed, center,space=0.3cm]
      \begin{pcvstack}[space=0.3cm]
        {\normalfont\procedure[linenumbering]{$\GenLocal(1^\lambda, \vect{d}_0[i])$}{
            \publishedVsArxiv{
              \pcreturn (k^{(i)}, t_k^{(i)}) \gets \HGen(1^\lambda, \vect{d}_0[i])
            }{%
              (k^{(i)}, t_k^{(i)}) \gets \HGen(1^\lambda, \vect{d}_0[i])
              \pcreturn (k^{(i)}, t_k^{(i)})%
            }
          }}
        {\normalfont\procedure[linenumbering]{$\Eval_\cP((k^{(1)},\dots,k^{(n)}), (c,\bar{x}^{(1)},\dots,\bar{x}^{(n)}))$}{
            \pcreturn (\HEval(k^{(1)},(c,\bar{x}^{(1)})),\\
            \t \dots,\HEval(k^{(n)},(c,\bar{x}^{(n)})))
          }}
        {\normalfont\procedure[linenumbering]{$\partialInfoLocal{}(t_k^{(i)}, y)$}{
            \pcif \vect{d}_0(t_k^{(i)}) = 0 \pcthen \pcreturn \cross \  \pcfi\\
            \{(0,\bar{x}),(1,\bar{x}')\} \gets \HDec(t_k^{(i),y})\\
            \pcif \bar{x} = \bot \text{ or } \bar{x}' = \bot\  \pcthen \pcreturn \bot\\
            \pcreturn \Hh((0,\bar{x}))
          }}
      \end{pcvstack}
      \begin{pcvstack}[space=0.3cm]
        {\normalfont\procedure[linenumbering]{$h((x^{(1)},\dots,x^{(n)}))$}{
            \pcreturn h(x^{1}) | \dots | h(x^{n})
          }}
        {\normalfont\procedure[linenumbering]{$\partialAlpha(i, t_k^{(i)},y ,b)$}{
            (y^{(1)},\dots,y^{(n)}) \eqdef y; \quad (b_c, b^{(1), \dots, b^{(n)}}) \eqdef b\\
            \{(0,\bar{x}),(1,\bar{x}')\} \gets \HDec(t_k^{(i),y})\\
            \pcif \bar{x} = \bot \text{ or } \bar{x}' = \bot \pcthen \pcreturn \bot \ \pcfi\\
            \pcif i = 1 \pcthen \pcreturn b_c \xor \langle b^{(i)}, \bar{x}\xor\bar{x'}\rangle\\
            \pcelse \pcreturn \langle b^{(i)}, \bar{x}\xor\bar{x'}\rangle \ \pcfi
          }}
        {\normalfont\procedure[linenumbering]{$\CheckHonestTrapdoor_\lambda(\vect{d}_0[i], t_k^{(i)}, k^{(i)})$}{
            \pcreturn \HCheckHonestTrapdoor_\lambda(\vect{d}_0[i], t_k^{(i)},k^{(i)})
          }}
      \end{pcvstack}
    \end{pchstack}
    \publishedVsArxiv{}{ 
      \caption{Construction distributable \AssumpFctCanPrime{} family}
      \label{fig:constructionFromSimpleGHZ}
    }%
  \end{figure}
\end{theoremE}
\begin{proofE}
  Most of the properties are simple to check. We just precise the $\delta'$-$2$-to-$1$ proof and blindness. First, we show that the function $f_k'$ are $\delta'$-$2$-to-$1$ with $\delta = 1-(1-\delta)^n$. Let $\diese_2(f)$ be the number of images having exactly $2$ preimages by $f$ (we will call this kind of preimages ``twin''), and $|\cK|$ the number of elements in $\cK$. Then by definition, for all $k$, $1- \delta \leq \diese_2(f_k)/|\cX|$. Let $k' = (k^{(1)},\dots,k^{(n)}) \in \cK^n = \cK'$, we want to show that $1-\delta'\eqdef (1-\delta)^n \leq \diese_2(f_{k'}')/(|\cX'|)$. First, we compute $\diese_2(f_{k'})$. Because of the assumption of the shape $(0,\bar{x})$ and $(1,\bar{x}')$ of all the couples of preimages, we can define for any $k^{(i)}$ the sets
  \begin{align}
    A_0^{(i)} \eqdef \{\bar{x} \mid f_{k^{(i)}}^{-1}(f_{k^{(i)}}(x)) = \{(0,\bar{x}), (1,\bar{x}')\}\}
  \end{align}
  and
  \begin{align}
    A_1^{(i)} \eqdef \{\bar{x}' \mid f_{k^{(i)}}^{-1}(f_{k^{(i)}}(x)) = \{(0,\bar{x}), (1,\bar{x}')\}\}
  \end{align}
  Moreover, due to this same condition, we have $|A_0^{(i)}| = |A_1^{(i)}| = \tfrac12 \diese_2(f_{k^{(i)}})$. Now, we compute a lower bound on the number of twin preimages of $f_{k'}$. Let $(\bar{x}^{(1)}, \dots, \bar{x}^{(n)}) \in A_0^{(1)} \times \dots \times A_0^{(n)}$: then for all $i$ there exists a unique $\bar{x}^{(1)\prime} \in A_1^{(i)}$ such that $f_{k^{(i)}}((0,\bar{x}^{(i)})) = f_{k^{(i)}}((1,\bar{x}^{(i)})$. So
  \begin{align}
    f_{k'}(0,\bar{x}^{(1)},\dots,\bar{x}^{(n)})
    &= f_{k^{(1)}}(0,\bar{x}^{(1)}) | \dots |f_{k^{(n)}}(0,\bar{x}^{(n)})\\
    &= f_{k^{(1)}}(1,\bar{x}^{(1)\prime}) | \dots |f_{k^{(n)}}(1,\bar{x}^{(n)\prime})\\
    &= f_{k'}(1,\bar{x}^{(1)\prime},\dots,\bar{x}^{(n)\prime})
  \end{align}

  So we found at least one different preimage, and due to the uniqueness of the above $\bar{x}^{(1)\prime}$ it is the only second preimage. Therefore:
  \begin{align}
    \frac{\diese_2(f_{k'})}{|\cX'|}
    &= \frac{2 \times |A_0^{(1)}| \times \dots \times |A_0^{(n)}|}{2 \left(\frac{|\cX|}{2}\right)^n}\\
    &= \frac{\diese_2(f_{k^{(1)}}) \times \dots \times \diese_2(f_{k^{(n)}})}{(|\cX|)^n}\\
    &\geq (1-\delta)^n
  \end{align}
  Which concludes the proof.

  To prove the inequality, we use the Bernoulli's inequality: since $\delta \in [0,1]$ and $n$ is a non-negative integer, we get: $(1-\delta)^n \geq 1-\delta n$, so
  \begin{align}
    \delta' = 1-(1-\delta)^n \leq 1-(1-\delta n) = \delta n
  \end{align}

  Since the keys of $\{f_k'\}$ are keys of $\{f_k'\}$ (except that $\vect{d}_0$ is a single bit), the properties of $\CheckHonestTrapdoor$ come directly from the properties of $\HCheckHonestTrapdoor$. All the other correctness properties are true by construction.

  The security is quite intuitive: since all trapdoors are independently sampled, if one can learn information about the $d_0$ sampled by another party, then it can break the $\INDFCT_{\Gen{}} ^\adv (\lambda) $ game of $f_k$. More formally, because of the properties on $\GenLocal$, the game $\INDPARTIAL_{\Gen{},\partialInfo{}} ^\adv (\lambda) $ can equivalently be rewritten as follows:
  \begin {pcimage}
    {\normalfont\game[linenumbering]{$\texttt{GAME1}$}{%
        (\cM, \vect{d}_0^{(0)},\vect{d}_0^{(1)}) \gets \cA_1(1^\lambda)\\
        \label{my:line:proofRedGame1Cond}\pcif \exists i \in \cM, \vect{d}_0^{(0)}[i] \neq \vect{d}_0^{(1)}[i]: \pcreturn \pcfalse\ \pcfi\\
        c \sample \bin\\
        \label{my:line:proofRedGame1GenLocal}\forall i, (k^{(i)}, t_k^{(i)}) \gets \HGen(1^\lambda, \vect{d}_0^{(c)}[i])\\
        y \gets \cA_2(k^{(1)},\dots,k^{(n)})\\
        \forall i, v[i] \gets \partialInfoLocal(i, t_k^{(i)}, y)\\
        \tilde{c} \gets \cA_3(\{(i,v[i])\}_{i \in \cM})\\
        \pcreturn \tilde{c} = c
      }}
  \end{pcimage}
  Then, we can define the following game in which the sampling of malicious $t_k^{(i)}$, the computing of $\partialInfoLocal$ and the initial condition are removed:
  \begin {pcimage}
    {\normalfont\game[linenumbering]{$\texttt{GAME2}$}{%
        (\cM, \vect{d}_0^{(0)},\vect{d}_0^{(1)}) \gets \cA_1(1^\lambda)\\
        c \sample \bin\\
        \label{my:line:proofRedGame2hgen}\forall i \in \cM, (k^{(i)}, t_k^{(i)}) \gets \HGen(1^\lambda, \vect{d}_0^{(c)}[i])\\
        \tilde{c} \gets \cA_2(\{k^{(i)}\}_{i \in \cM})\\
        \pcreturn \tilde{c} = c
      }}
  \end{pcimage}
  Then, because $\cA_2$ in \texttt{GAME2} can do itself the sampling and computing done in \texttt{GAME1}, and because the removing the condition line \ref{my:line:proofRedGame1Cond} can only increase the probability of winning the game, we have:
  \begin{align}
    \max_{\QPT{} \cA} \pr{\texttt{GAME1}^\cA}(\lambda) \leq \max_{\QPT{} \cA'} \pr{\texttt{GAME2}^{\cA'}}(\lambda)
  \end{align}
  Then, we define a series a hybrid games in which we gradually replace the $\HGen(1^\lambda,d_0^{(c)}[i])$'s with $\HGen(1^\lambda,0)$, which $\cZ$ starting from $\emptyset$ until $\cZ = \cM$.
  \begin {pcimage}
    {\normalfont\game[linenumbering,lnstart=\getrefnumber{my:line:proofRedGame2hgen}-1]{$\texttt{GAME3}_\cZ$}{%
        \hcancel[red]{\forall i \in \cM, (k^{(i)}, t_k^{(i)}) \gets \HGen(1^\lambda, \vect{d}_0^{(c)}[i])}\pcskipln\\
        \textcolor{green!50!black}{\forall i \in \cM, \pcif i \in \cZ \pcthen x \eqdef 0 \ \pcelse x \eqdef \vect{d}_0^{(c)}[i] \ \pcfi; (k^{(i)}, t_k^{(i)}) \gets \HGen(1^\lambda, x)}
      }}
  \end{pcimage}
  This is possible because the game $\INDFCT$ of $\HGen$ can be seen as a \CPA{} secure encryption (where $\Gen$ is the concatenation of the key generation and encryption of $\vect{d}_0$), itself equivalent to semantic security: i.e.\ the ciphertext does not give any advice on the clear text. Therefore we can replace the clear text with $0$ for example, which gives for all $j$:
  \begin{align}
    \max_{\QPT{} \cA} \pr{\texttt{GAME3}_{\cZ}^\cA}(\lambda) \leq \max_{\QPT{} \cA'} \pr{\texttt{GAME3}_{\cZ \cup \{j\}}^{\cA'}}(\lambda)
  \end{align}
  At the end of the hybrid series, when $\cZ = \cM$, we get a final game where no information about $c$ are given to the adversary: it is therefore impossible to win this game $\texttt{GAME3}_\cM$ with probability better than $\frac{1}{2}$. Therefore we get:
  \begin{align}
    \max_{\QPT{} \cA} \pr{\INDPARTIAL_{\Gen{},\partialInfo{}} ^\adv}(\lambda) \leq \frac{1}{2} + \negl[\lambda]
  \end{align}
  Which ends the proof.

  Finally, the security of the game $\INDPARTIAL$ implies the security of $\INDFCT$ since when $\cM = \emptyset$, both games are equivalent.
\end{proofE}

\arxivOnly{
  \section{Acknowledgement}
  We thank AirFrance for the unscheduled one night delay of a flight which was instrumental in starting this line of research. We are also very grateful to Alex Grilo, Chris Peikert, Florian Bourse, Geoffroy Couteau, the StackExchange user \href{https://crypto.stackexchange.com/users/45690/mark}{Mark}, Michael Reichle, Thomas Vidick and Romain Gay for very valuable discussions. This work has been supported in part by the French ANR Project ANR-18-CE39-0015 (CryptiQ), the French ANR project ANR-18-CE47-0010 (QUDATA), the EU Flagship Quantum Internet Alliance (QIA) project and the European Union as H2020 Programme under grant agreement number ERC-669891. EK acknowledges support from the EPSRC Verification of Quantum Technology grant (EP/N003829/1), the EPSRC Hub in Quantum Computing and Simulation (EP/T001062/1), and the UK Quantum Technology Hub: NQIT grant (EP/M013243/1).
}

\publishedVsArxiv{\printbibliography[segment=0,heading=bibintoc]}{\printbibliography[heading=bibintoc]}
  \newpage
  \appendix
  \newrefsegment 
  \begin{center}
    \Huge{\textsc{Supplementary Material}}
  \end{center}
  \publishedOnly{\printbibliography[title=Appendix References,filter=appendixOnlyFilter]}

  \publishedVsArxiv{\section{Details on MPC}\label{appendix:introMPC}
  \textEnd[category=introMPC]{}
  \printProofs[introMPC]}{}

  \section{Details on Zero-Knowledge}\label{appendix:introZK}

  \textEnd[category=introZK]{}
  \printProofs[introZK]

  \section{Details on \LWE{}}\label{appendix:introLWE}
  \textEnd[category=introLWE]{}
  \printProofs[introLWE]

  \publishedVsArxiv{
    \section{Detailed proofs for the different protocols}

    \subsection{Proofs of protocol \blind{}}\label{cref:appendixBlind}
    We present in this section the detailed proofs of \cref{subsec:blind}.

    \textEnd[category=proofsBlind]{}
    \printProofs[proofsBlind]

    \subsection{Proofs of protocol \blindSup{}}
    We present in this section the detailed proofs of \cref{subsec:blindSup}.
    Its correctness is guaranteed by \cref{lem:correctnessBlind}, proven above.

    \textEnd[category=proofsBlindSup]{}
    \printProofs[proofsBlindSup]

    \subsection{Proof of the impossibility of a protocol \blindCan{}}\label{subsec:blindCan}
    We present in this section the detailed proofs of \cref{subsec:blindCan}.

    \textEnd[category=proofsBlindCan]{}
    \printProofs[proofsBlindCan]

    \subsection{Proofs of protocol \blindCanSup{}}
    We present in this section the detailed proofs of \cref{subsec:blindCanSup}.

    \textEnd[category=proofsBlindCanSup]{}
    \printProofs[proofsBlindCanSup]

    \subsection{Proofs of protocol \authBlindCanDist{}}\label{appendix:authBlindCanDist}
    We present in this section the detailed proofs of \cref{subsec:authBlindCanDist}.

    \textEnd[category=proofsAuthBlindCanDist]{}
    \printProofs[proofsAuthBlindCanDist]

    \section{Proofs of the NIZKoQS}
    \textEnd[category=proofsNIZKoQS]{}
    \printProofs[proofsNIZKoQS]

    \section{Proofs of the function constructions}\label{appendix:fctConstruction}
    \textEnd[category=functionConstructions]{}
    \printProofs[functionConstructions]
  }

\end{document}